\documentclass[11pt]{article}

\usepackage[normalem]{ulem}

\RequirePackage[OT1]{fontenc}
\RequirePackage{amsthm,amsmath}

\usepackage[square,numbers]{natbib}
\usepackage{bm}

\RequirePackage[colorlinks,citecolor=blue,urlcolor=blue]{hyperref}
\usepackage[dvipsnames]{xcolor}
\usepackage{verbatim}
\usepackage{xcolor}
\usepackage{fullpage}
\usepackage{subfig}
\usepackage{layout}
\usepackage{dsfont}
\usepackage[mathscr]{eucal}
\usepackage[toc,page]{appendix}
\usepackage{mathrsfs}
\usepackage{color}
\usepackage{pifont}
\usepackage{bm}
\usepackage{latexsym}
\usepackage{amsfonts}
\usepackage{amssymb}
\usepackage{epsfig}
\usepackage{graphicx}
\usepackage{multirow}
\usepackage{hyperref}
\usepackage{enumitem}
\usepackage{amsmath}
\usepackage{amsthm}
\DeclareMathOperator*{\argmin}{arg\,min} 
\usepackage{units}
\usepackage{xfrac}  
\usepackage{breqn}
\usepackage{dsfont}
\usepackage{graphicx}

\usepackage{algorithm}
\usepackage{algpseudocode}

\newcommand {\tr}{\mathrm{tr}}
\newcommand {\R}{\mathbb{R}}
\newcommand{\PP}{\mathbb{P}}
\newcommand{\E}{\mathbb{E}}
\renewcommand {\H}{\mathbb{H}}
\newcommand {\1}{\textrm{\textbf{1}}}
\newcommand {\Id}{I}
\newcommand {\F}{\mathcal{F}}
\newcommand {\G}{\mathcal{G}}

\newcommand {\N}{\mathcal{N}}

\newcommand {\V}{\mathcal{V}}
\newcommand {\U}{\mathcal{U}}
\newcommand {\A}{\mathcal{A}}
\newcommand {\B}{\mathcal{B}}
\newcommand {\T}{\mathcal{T}}
\newcommand {\M}{\mathcal{M}}
\renewcommand {\P}{\mathcal{P}}
\renewcommand {\L}{\mathcal{L}}
\newcommand {\FF}{\mathscr{F}}
\newcommand {\BB}{\mathscr{B}}
\newcommand {\KK}{\mathscr{K}}
\newcommand {\RR}{\mathscr{R}}

\newcommand {\TT}{\mathscr{T}}
\newcommand {\WW}{\mathscr{W}}

\newcommand {\CC}{\mathscr{C}}
\newcommand{\range}{\mathrm{range}}
\newcommand{\Frechet}{Fréchet }
\newcommand{\Tan}{\mathscr{T}}
\newcommand{\der}{d}
\newcommand{\bJ}{\boldsymbol{\mathfrak{J}}}
\newcommand{\J}{\mathfrak{J}}

\usepackage{enumitem, hyperref}
\makeatletter
\def\namedlabel#1#2{\begingroup
  #2%
  \def\@currentlabel{#2}%
  \phantomsection\label{#1}\endgroup
}
\makeatother

\allowdisplaybreaks

\newtheorem{theorem}{Theorem}
\newtheorem{proposition}{Proposition}

\newtheorem{lemma}{Lemma}

\newtheorem{remark}{Remark}

\theoremstyle{plain}

\usepackage{hyperref}
\hypersetup{ colorlinks = true, linkcolor  = blue}

\numberwithin{equation}{section}
  \usepackage[affil-it]{authblk}

\makeatletter
\newcommand{\subjclass}[2][1991]{%
  \let\@oldtitle\@title%
  \gdef\@title{\@oldtitle\footnotetext{#1 \emph{Mathematics subject classification.} #2}}%
}
\newcommand{\keywords}[1]{%
  \let\@@oldtitle\@title%
  \gdef\@title{\@@oldtitle\footnotetext{\emph{Key words and phrases.} #1.}}%
}
\makeatother
\usepackage{abstract}

 \title{{Statistical Inference for Bures-Wasserstein Flows\footnote{Research supported by a Swiss National Science Foundation (SNF) grant.}}}

\author{
  Leonardo V.\ Santoro
  \qquad 
  Victor M. Panaretos   \\ {\footnotesize{\texttt{leonardo.santoro@epfl.ch}  \qquad\,\,\texttt{victor.panaretos@epfl.ch}}}}
  
  \affil{Institut de Math\'ematiques\\École Polytechnique Fédérale de Lausanne}

\begin{document}

\maketitle

\begin{abstract}
 We develop a statistical framework for conducting inference on collections of  time-varying covariance operators (covariance flows) over a general, possibly infinite dimensional, Hilbert space. We model the intrinsically non-linear structure of covariances by means of the Bures-Wasserstein metric geometry. We make use of the Riemmanian-like structure induced by this metric to define a notion of mean and covariance of a random flow, and develop an associated Karhunen-Lo\`eve expansion. We then treat the problem of estimation and construction of functional principal components from a finite collection of covariance flows, observed fully or irregularly.
 Our theoretical results are motivated by modern problems in functional data analysis, where one observes operator-valued random processes  -- for instance when analysing dynamic functional connectivity and fMRI data, or when analysing multiple functional time series in the frequency domain. {Nevertheless, our framework is also novel in the finite-dimensions (matrix case), and we demonstrate what simplifications can be afforded then}. We illustrate our methodology by means of simulations and data analyses.
  
\medskip
    \noindent \textbf{MSC2020 classes:} 62R10, 62R20, 62R30, 62G05, 60G57 

\end{abstract}

 \begin{footnotesize}
 \tableofcontents
 \end{footnotesize}

\section{Introduction}

\subsubsection*{Background}

 Covariance operators are crucial in  functional data analysis, encoding the second-order variation of a stochastic process, and providing the canonical means to analyse such variation by way of the Karhunen-Lo\`eve expansion. Further to their core operational role, covariances are increasingly becoming the main focus of analysis in modern applications involving more complex functional structures. 
Modern data collection techniques are now demanding of methodology going beyond finite collections of covariance operators, namely for circumstances where one has a dynamic evolution of a finite collection of covariances, i.e.\ \emph{covariance flows}: collections of time-evolving covariance operators that are expected to possesses some level of regularity with respect to the time parameter. These can be seen as functional data valued in the space of covariance operators. For example, neuroscientists study dynamic functional connectivity, where one records longitudinal samples of time-varying covariances at increasingly high resolutions \cite{dai2019analyzing}, whereas collections of multiple functional time series will generate multiple spectral density operators  \cite{panaretos2013fourier}.  
Furthermore, data in the form of time-varying covariance \textit{matrices} is increasingly common, and has recently attracted significant interest from the statistical community \cite{dai2019analyzing, huang2020statistical, zhang2018rate,petersen2019frechet} mostly due to a growing  interest in studying the brain functional connectome generated from functional MRI \cite{van2013wu}.

In light of these developments, the purpose of this paper is to develop a framework for conducting statistical analyses on flows of covariance operators over a general, possibly infinite dimensional, separable Hilbert space $\H$. 
The intrinsic non-linearity of covariance operators renders many techniques relying on a linear structure ineffective or inapplicable, and requires an adaptation of the tools of classical statistical analysis: see for instance \citet{schwartzman2006random}, \citet{dryden2009non} in finite dimensions or \citet{pigoli2014distances} and 
 \citet{panaretos2019statistical} in infinite dimensions. In fact, the infinite dimensional case is of particular mathematical interest, as it constitutes an example of an infinite-dimensional stratified space, with dense singularities.
 
 We consider the space $\KK=\KK_{\H}$ of self-adjoint, trace-class, and non-negative operators over $\H$ and, fully embracing the non-negative nature of these operators, we view $\KK$ as a (non-linear) metric space equipped with the Bures-Procrustes-Wasserstein metric (\citet{bhatia2019bures}, \citet{masarotto2019procrustes})
 $$\Pi(F,G)^2 =  \tr(F)+\tr(G)-2\tr(G^{\nicefrac{1}{2}}FG^{\nicefrac{1}{2}})^{\nicefrac{1}{2}}\quad F,G\in\KK.$$  This choice is emerging as canonical for the analysis of covariance operators: it enjoys an intrinsic connection to the theory of optimal transportation of (Gaussian) measures, and is interpretable via phase variation at the level of the corresponding processes. We define a space of \textit{regular} functions defined over a compact time-interval, say $[0,1]$, and valued in the metric space $(\KK,\Pi)$. We will refer to such such functions $\F :[0,1]\rightarrow \KK$, $\F:t\mapsto \F(t)$ as \textit{flows} of covariance operators. And, our purpose will be to extend the basic machinery of functional data analysis to deal with functional data of the form $\F$.

\subsubsection*{Contributions}
We develop a framework for conducting a second-order statistical analysis on collections of flows of covariance operators.
 In classical terms, we are not concerned with nonparametric regression (a single deterministic curve) as in \citet{krebs2022statistical} and \citet{chau2021intrinsic}, but with functional data analysis (multiple random curves) valued in a space that is both intrinsically non-linear and infinite-dimensional.  Past work has focussed on spaces that are either non-linear but finite dimensional (e.g. \citet{lin2019intrinsic}, \citet{dubey2020functional}) or infinite-dimensional but intrinsically flat (e.g. \citet{kim2020principal}).
Our work is the first to consider theory and methodology for covariance operator \emph{flows} at a genuinely infinite-dimensional level.

\medskip

 Our theoretical results are based on first-principles assumptions that are interpretable and verifiable at the level of the random flows themselves and/or their model: we do not make \emph{a priori} assumptions that are hard to verify (such as existence/uniqueness/separation of Fr\'echet means, metric entropy bounds, etc) but deduce such properties from the basic assumptions on the data generating mechanism.

\medskip

In more detail, in Section \ref{sec : flows} we define an $L^2$-type metric on flows in the Bures-Wasserstein space, which allows us to consider notions such as continuity/regularity of a flow,  Fr\'echet means of  collections of flows , and trace-class covariance of a collection of flows relative to their Fr\'echet mean. Doing so requires the study of the geometry of $L^2$-flows, which is related to the Bures-Wasserstein geometry. In the spirit of \citet{dai2018principal}, \citet{lin2019intrinsic} and \citet{zhou2021intrinsic}, we show how the image of the flow under a suitable log transform can be seen as an element of an abstract Hilbert space, obtained as the \textit{direct infinite sum} of the tangent spaces to the \Frechet mean flow.  This allows us, in Section \ref{sec : secondorder}, to extend classical functional tools such as Mercer's decomposition and the Karhunen-Loève expansion. We then develop the theory of estimation of the mean flow and the covariance of the flows, including their spectrum -- in effect yielding a functional principal component analysis for flows of covariances: see Section \ref{sec : estimation}. Owing to the subtleties of the infinite-dimensional transport problem, and contrary to past work, we do not have access to tools such as parallel transport. Rather, to assess convergence, we introduce a natural embedding of elements of the tangent bundle into the space of flows of square-integrable Hilbert-Schmidt operators via a canonical isomorphism. 
{Since our results are valid in general Hilbert spaces, they of course apply to the finite-dimensional (matrix) case; they are in fact novel in that case as well, and we explain what simplifications finite dimensions afford.

We consider estimation from discrete, noise corrupted and sparsely observed covariance flows in Section \ref{sec : irregular obs}, and we show in Section \ref{sec : computation} that our methods are stable with respect to discretisation, both in time and in space. We also give an explicit description of the practical implementation.

We apply our methods to functional time series in Section \ref{sec : fts}, where we make use of the spectral representation of stationary sequences to conduct a frequency domain based analysis.

We illustrate our methodology by way of simulations in Section \ref{sec : simulations} and the analysis of real data, specifically to EEG data in Section \ref{subsect : EEG} and mortality data 
in Section \ref{subsect : FTS mortality}. 

{The proofs of all statements are collected in the Appendix.}

\subsubsection*{Related Work}
The modeling and analysis of functional data valued in spaces beyond $\mathbb{R}^d$ has been a very active research field  in recent years,  and covers a wide range of objects with distinct geometrical properties. 
Arguably, almost any statistical technique in a non-linear space will rely on some form of linearization, all the way back to classical Euclidean shape theory in the 70s . The question is of course \emph{how to do so} especially in complex spaces that have both geometric and analytical structure.
 Past work has focussed on spaces that are either non-linear but finite dimensional \cite{lin2019intrinsic, dubey2020functional} or infinite-dimensional but intrinsically flat \cite{zhou2021intrinsic, kim2020principal, santoro2023karhunen}.
In more detail, \citet{dai2018principal} and \citet{lin2019intrinsic} have considered functional data taking values on a Riemannian manifold; \citet{zhou2021intrinsic} have similarly approached the analysis of functional data in the space of probability measures on $\mathbb{R}$; \citet{kim2020principal, santoro2023karhunen} have considered \textit{Hilbertian} functional data, i.e.\ functional data taking values in a general Hilbert space, extending the celebrated Karhunen-Loève theorem. 
\citet{dubey2020functional} have proposed a different perspective, and introduced the novel concept of \textit{metric covariance} to model curves in general metric space. 

Many more examples may be found present in the literature, all sharing the goal of overcoming the limits of classical FDA methodologies, and proposing new ways to account for the non-linearity or double infinite dimensionality that arises when handling functional data in more general spaces.

\smallskip

 Our work is the first to consider theory and methodology for covariance operator \emph{flows} at a genuinely infinite-dimensional level; at a high level, our approach is comparable to that of \citet{lin2019intrinsic}. However our context goes well beyond their setting: we work with an infinite-dimensional, stratified space, rather then a finite-dimensional Riemmanian manifold. Indeed, even if we were to consider flows of positive-definite covariance matrices (which constitute elements in a proper finite-dimensional Riemmanian manifold), it is not at all clear if and how the assumptions required by \cite{lin2019intrinsic} could be shown to hold in this setting. In that sense, our work is of interest even in the case of full rank covariance matrices.  In more detail: 
    \begin{enumerate}
   \item The space we work with is neither finite-dimensional nor a Riemmanian manifold. Instead we have an \textit{infinite-dimensional}, \textit{stratified} space, filled with \textit{singularities}. From a methodological stand point, this implies the following three major challenges. First, well definition of logarithm maps is non-trivial, but needs careful verification at each step, as the space is fraught with singularities. It is a delicate analytical problem. Second, we do not have local charts to do differential calculus, but need to employ the concept of \Frechet differentiability. Third,  parallel transport -- which is the instrumental tool in all of the theory developed in \cite{lin2019intrinsic,shao2022intrinsic, zhou2021intrinsic} -- is not available at all. Instead, we make use of a tool that is peculiar to our setting, i.e.\ a natural isometric embedding into a common Hilbert space, by means of which we distinctly carry out the computation and theory.

    \item 
        Our theoretical results are based on first principles assumptions that are intuitively interpretable and, most importantly, {verifiable}.  We do not make a priori assumptions that are hard to verify, such as existence/uniqueness of \Frechet means (assumption A.1 in \cite{lin2019intrinsic}) and well-separation thereof (assumption B5 in \cite{lin2019intrinsic}), or bounds on the spectrum of the Hessian of the squared distance function (assumption B.6 in \cite{lin2019intrinsic}).
        Instead we \textit{deduce} such properties from our elementary assumptions. This is in contrast with a recent trend in non-Euclidean data analysis where a very general setting is sought, not committing to a specific space, and as a result requires heavy and technical downstream theoretical assumptions to apply general M-estimation results \cite{van1996weak}); by contrast, we prefer to engage and make use of the specific structural and geometric properties of our underlying space, which we believe ultimately yields more informative theory.

    \end{enumerate}

\section{Background and Notation}

Let $(\H, \langle\cdot,\cdot\rangle)$ be an arbitrary separable Hilbert space.  The {trace} of an operator $\F:\H\to\H$ is defined by $\tr(F)= \sum _{i\geq 1}\langle F e_i,e_i\rangle $, where $\{e_i\}_{i\geq 1}$ is any Complete Orthonormal System (CONS) of $\H$.
The operator norm, trace norm and Hilbert-Schmidt norm are respectively defined as:
 $$
 \|F\|_\infty:=\sup_{\|h\|=1}\|F h\|,\quad 
 \|F\|_2:= \sqrt{\tr(F^*F)}, \quad 
 \|F\|_1:= \tr(\sqrt{F^*F}).
 $$
 In particular, any operator $F$ satisfies:
 $$
  \|F\|_\infty \leq
 \|F\|_2\leq
 \|F\|_1.
 $$
We denote by $\BB,\BB_2$ and $\BB_1$ the space of bounded, Hilbert-Schmidt and trace class linear operators on $\H$. We denote the adjoint of an operator $F$ by $F^*$, and say the operator $F$ is self adjoint if $F=F^*$. We say that $F$ is non-negative definite and write $F\succeq 0$, if it satisfies $\langle F h,h\rangle\geq 0$ for all $h\in\H$. We say that it $F$ strictly positive definite (or regular) and write $F\succ 0$ if the last inequality is strict except at $h=0$. An operator is compact if for any bounded sequence $\{h_n\}\subset  \H$, $\{F h_n\}\subset  \H$ contains a convergent subsequence. Given a compact, non-negative operator $F$, there exists a unique square root $F^{\nicefrac{1}{2}}$ operator, which satisfies $(F^{\nicefrac{1}{2}})^2=F$.
The \textit{kernel} of an operator $F$ is denoted by $\ker(F) = \{h\in\H\::\: F h= 0\}$, and its \textit{range} by $\range(F)=\{F h \::\: h \in \H\}$. 
We write $\Id$ for the identity operator on $\H$. 
  For $f,g\in\H$, the \textit{tensor product} $f\otimes g \::\: \H\rightarrow\H$ is the linear operator defined by:
$$
(f\otimes g)u = \langle g,u\rangle f,\qquad u\in\H.
$$ 

\bigskip 
We denote the space of square integrable $\H$-valued functions by
 $$L^2_{\H} = L^2( [ 0 , 1 ] ,  \H)=\left\{f:[0,1]\rightarrow\H \:\Big\vert\: \int_{0}^{1}\|f(t)\|^{2} d t<\infty\right\},
 $$  
 with inner product $\langle f, g\rangle_{L^2_{\H}}=\int_{0}^{1}\langle f(t), g(t)\rangle d t$.

\bigskip

For a time-indexed collection of covariance operators $\{F_t\}_{t\in[0,1]}$, we denote by $\F$ the corresponding flow of covariances: $\F(t) = F_t$ for $t\in[0,1]$.  This subtle distinction is similar to how one distinguishes a random process from the corresponding random element in a function space. We use capital letters to denote operators, and calligraphic capital letters to denote flows. Script capital letters are reserved for spaces of operators or operator flows.

\subsection{Optimal Transport and the Wasserstein Space}

Let $\WW_2(\H)$ denote the space of Borel probability measures on $\H$ with finite second moment. 
 For $\mu,\nu \in\WW_2(H)$, let $\Gamma(\mu,\nu)$ be the set of couplings of $\mu$ and $\nu$, i.e.\ the set of probability measures $\pi$ on $\H \times \H$
 with marginals $\mu$ and $\nu$.
The Wasserstein distance between $\mu$ and $\nu$ is defined in terms of the Kantorovich problem \cite{kantorovich1942translocation} of optimal transportation:
\begin{equation}\label{eq : wasserstein distance}
    W_{2}(\mu, \nu)=\left( \inf _{\pi \in \Gamma(\mu, \nu)} \int_{\H \times \H}\|x-y\|^{2} \mathrm{~d} \pi(x, y)\right)^{\nicefrac{1}{2}} 
\end{equation}
and defines a proper metric on $\WW_2(\H)$. 
It may be shown that a minimising coupling $\pi$ always exists, for any marginal pair of measures $\mu, \nu \in \WW_2(\H)$. In fact, when $\mu$ is sufficiently regular,
the optimal coupling is unique and given by a deterministic coupling, $\pi=\left(\Id, T_{\mu}^{\nu}\right) \# \mu$, for some  deterministic optimal (transportation) map $T_{\mu}^{\nu}: \H \rightarrow \H$. Here `$T\#\mu$' represents the operation of \emph{pushing forward} a measure $\mu$ by  a measurable transformation $T$ of its domain, $T\#\mu(A)=\mu(T^{-1}(A))$. 
 When an optimal map exists, the Kantorovich problem is equivalent to the {Monge} problem \cite{monge1781memoire} of optimal transportation, which minimises the transportation cost only over deterministic couplings:
 \begin{equation}\label{eq:monge}
 \inf _{T \::\: T_{\#}\mu = \nu } \int_{\H}\|x-T(x)\|^{2} \mathrm{~d} \mu(x),
 \end{equation}
in the sense that the optimal coupling $\pi$ in \eqref{eq : wasserstein distance} is induced by a map, which minimises \eqref{eq:monge}.

\medskip

The transport structure of the Wasserstein distance induces a manifold-like geometry. Namely, for $\mu\in\WW_2(\H)$, a measure $\nu\in\WW_2(\H)$ can be identified with the optimal map $T_{\mu}^{\nu}$ provided such optimal map exists and is unique. In fact, by subtracting the identity $\Id$, this creates a correspondence between a subset of the Wasserstein space with a subset of the linear space $L^2(\mu)$ in such a way that $\mu$ itself corresponds to $0 \in L^2(\mu)$. This motivates the following definition of a \textit{formal tangent space} at $\mu$ in $\WW_2(\H)$ , which is due to \citet[][Definition 8.5.1]{ambrosio2005gradient}:
\begin{equation}\label{eq : wass tangent space}
\Tan_{\mu}=\overline{\left\{\lambda(T-\Id): T \in L^2(\mu) ; \;T \text { optimal between } \mu \text { and } T _ \# \mu ; \lambda>0\right\}}^{L^2(\mu)}
\end{equation}
As a subset of $L^2(\mu)$, the tangent space inherits the inner product:
$$
\langle T, S \rangle_{\mu}=\int_{\H}\langle S (x), R(x)\rangle \der \mu(x), \quad S, R \in L^2(\mu) .
$$
Note that since optimality of $T$ is independent of $\mu$, the only part of this definition that depends on $\mu$ is this inner product and the corresponding closure operation. Though not obvious from the definition, this is a linear space: see \citet{ambrosio2005gradient}.

\medskip
\begin{remark}[Tangent vectors and geodesics]
As happens for the case of manifolds, there is a one-to-one correspondence between geodesics paths and vectors in the tangent bundle: indeed, for any point $p$ in a manifold $M$ and tangent vector $v \in \mathrm{Tan}_p$, there exists a geodesic $\gamma_{v}: [0,1] \rightarrow$ $M$, where $\gamma_{v}(0)=p$ and $\gamma'_{v}(0)=v$. Something similar occurs in the Wasserstein case: given $\mu, \nu \in \WW_2(\H)$  such that the optimal map from $\mu$ to $\nu, T_{\mu}^{\nu}$, exists, we may define a curve:
$$
\mu({\lambda})=\left[\Id+\lambda\left(T_{\mu}^{\nu}-\Id\right)\right] \# \mu, \quad \lambda \in [0,1]
$$
which is known in the literature as McCann's interpolation (\citet{mccann1997convexity}), and \textit{interpolates} between the two measures, in the sense that $\mu({0})=\mu$ and $\mu({1}) = \nu$. In fact, such curve is a constant speed geodesic, and represents the \textit{shortest path} between $\mu,\nu$ in the Wasserstein geometry, and generalises the notion of a ``straight line" to the (non-linear) Wasserstein space $\WW_2(\H)$.  
\end{remark}

\subsection{Gaussian Optimal Transport and the Bures-Wasserstein Space}\label{sec : BW space}

Our interest in optimal transport stems from the fact that, when restricting attention to \emph{centred Gaussian measures} on $\H$, the Wasserstein distance and its geometry induce a natural distance and geometry on \emph{convariance operators}. 

Concretely, consider the space $\KK_{\H}$ of non-negative definite, self-adjoint, and trace-class operators on the Hilbert space $\H$. Notice that $F \in \KK$ if and only if there exists an $\H$-valued centred Gaussian element $X$ satisfying $\E\|X\|^2<\infty$ with covariance operator $\E[X\otimes X] = F$. Equivalently, the space $\KK$ can be identified with the space of centered Gaussian measures on $\H$. The latter is a closed and convex subset in $\WW_2(\H)$. Thus it is fruitful to make use the theory of optimal transportation of measures to define a distance and geometry between covariance operators: given two covariances $G$ and $F$, we can consider the corresponding Gaussian measures $\mu_{F} \equiv \N(0,F)$ and $\nu_{G} \equiv \N(0,G)$, and define the Bures-Wasserstein metric $\Pi$ as:
\begin{equation}
    \label{eq : bures-wasserstein distance}
     \Pi(F,G) :=  W_2(\mu_{F},\nu_{G}).
\end{equation}
Gaussian measures are quite exceptional, constituting a rare case where the 2-Wasserstein distance admits a closed form solution, namely
$$
\Pi(F,G)=W_2(\mu_{F},\nu_{G})  = \sqrt{ \tr(F)+\tr(G)-2\tr(G^{\nicefrac{1}{2}}FG^{\nicefrac{1}{2}})^{\nicefrac{1}{2}}}.
$$
The use of this metric in an infinite-dimensional setting was initiated (under a different  expression) by \citet{pigoli2014distances}, who extended the finite-dimensional Procrustes metric --- a central concept in statistical shape theory, also previously shown to be very well adapted to the statistical analysis of covariance \emph{matrices} (see e.g. \citet{dryden2009non}). \citet{masarotto2019procrustes} later showed that the Procrustes metric between two covariances in fact coincided with the 2-Wasserstein metric between the corresponding centred Gaussians, i.e.   the infinite-dimensional Bures--Wasserstein metric. This allowed them to use the rich structure of the optimal transport problem to establish key geometrical and statistical properties of this metric. We now review those properties that are elemental for this paper. For a complete account, see \citet{masarotto2019procrustes} and references therein.

Optimal transport maps between Gaussians also admit a closed-form representation: writing $T_{F}^{G}$ in lieu of the optimal transport map $T_{\mu}^{\nu}$, one may show that:
\begin{equation}
    \label{eq : transport map}
    T_{F}^{G} := 
        F^{-\nicefrac{1}{2}} (F^{\nicefrac{1}{2}}G F^{\nicefrac{1}{2}})^{\nicefrac{1}{2}}F^{-\nicefrac{1}{2}}.
\end{equation}
This formula holds whenever such a map exists. A necessary condition for existence is that the ranges of $F$ and $G$ are appropriately nested: if $\operatorname{ker}\left(F\right) \subseteq \operatorname{ker}\left(G\right)$,  then the map \eqref{eq : transport map} is well defined on a subspace of $\H$ with $\mu_{F}$--measure 1 , containing the range of $F^{\nicefrac{1}{2}}$.
 Conversely, any linear operator  $T\::\:\H\rightarrow\H$ is an optimal map between Gaussian measures provided it satisfies: 
\begin{enumerate}
    \item[i)] $T \geq 0$, $T^*=T$,
 
    \item[ii)]$T F^{\nicefrac{1}{2}} \in \BB_2$ for some $F \in\KK$.
\end{enumerate} 
In other words, if $\mu\equiv \N(0,F)$ and $T$ satisfies the two properties above, then the push-forward of $\mu$ through $T$ is a centered Gaussian measure with covariance $TF T$: $T_{\#}\mu \equiv \N(0, TF T)$.

The geometry induced on $\KK$ by the metric $\Pi$ inherits the peculiar the geometric structure of the Wasserstein space $\WW_2(\H)$, but specialised to the Gaussian case.  In particular, the \textit{tangent space} at any $F\in\KK$, can be defined as
\begin{equation}
    \label{eq : tangent space}
    \begin{split}
    \Tan_{F}\equiv   \Tan_{F}(\KK):=&  \overline{\{ \lambda \cdot (T - \Id) \::\: \lambda\geq 0, \: T \geq 0, \: \| TF^{\nicefrac{1}{2}} \|_2<\infty \}} \\
    =& \overline{\{\Gamma \::\: \Gamma=\Gamma^*,\:\|\Gamma F^{\nicefrac{1}{2}}\|_2 <\infty\}}
    \end{split}
\end{equation}
where the closure is taken with respect to the inner product on $\Tan_{F}$,
\begin{equation}
    \label{eq : inner product on tangent space}
    \begin{split}
    \langle \Gamma, \Gamma' \rangle_{\Tan_{F}} := & \tr(\Gamma F\Gamma') 
                            = \E [\langle \Gamma X , \Gamma' X\rangle], \qquad \text{where } 
                            X \sim \N(0,F).
    \end{split}
\end{equation}
Endowed with such inner product, and  quotienting-out the equivalence relation $\Gamma \sim \Gamma' \Leftrightarrow (\Gamma-\Gamma')F = 0$, the tangent space $(\Tan_{F}, \langle \cdot,\cdot \rangle_{\Tan_{F}})$ is a separable Hilbert space. 
We denote by $\TT_{\KK}$ the tangent bundle, i.e., the disjoint union of all tangent spaces:
\begin{equation} \label{eq : tangent bundle}
    \TT_{\KK}:= \bigcup_{F\in\KK}\{( F, \Tan_{F})\}.
\end{equation}

\noindent For any $F\in\KK$, we may define an \textit{exponential map} by: 
\begin{equation}
    \label{eq : exp map}
\exp_{F} \::\: \Tan_{F} \rightarrow \KK , \qquad  \exp_{F}(\Gamma) = (\Gamma+\Id )F(\Gamma+\Id ).
\end{equation}
and a \textit{logarithm map} as its left inverse,  projecting $\KK$ onto $\Tan_{F}$. Specifically, for $G\in\KK$ satisfying $\ker F \subset \ker G$, we may define $
\log_{F}(G) = T_{F}^{G} - \Id.
$

\medskip

For $F_0,F_1\in\KK$,  a constant-speed geodesic interpolation $F_{\lambda}$  can be defined, employing McCann's interpolation on Gaussian measures:
\begin{equation}\label{eq : McCann's interpolant}
    \begin{split}
        F_{\lambda} := &  
        \left[ \lambda T_{F_0}^{F_1} + (1-\lambda) \Id \right] F_0 \left[ \lambda T_{F_0}^{F_1} + (1-\lambda) \Id \right] \\
                          = &  
        \lambda^2F_1 + (1-\lambda)^2 F_0 + \lambda(1-\lambda)(TF_0 + F_0T).
    \end{split}
\end{equation}
The path $\lambda \mapsto F_{\lambda}$, defined on the unit interval, indeed interpolates between $F_0$ and $F_1$,
and has constant speed in that $\Pi(F_{\lambda+h}, F_{\lambda}) = h\cdot \Pi(F_{0},F_{1})$.

\begin{remark}
Although the space $(\KK,\Pi)$ has manifold-like properties, it is \textit{not} an infinite dimensional Riemannian manifold. Indeed, it is not in general possible to define an exponential map from a neighbourhood of the origin in $\Tan_{\F}$ into $\KK$ which is a homeomorphism (see \citet{ambrosio2004gradient}). Even when $\dim(\H)<\infty$, one has an Alexandrov geometry with singularities stratified according to the rank of the corresponding covariances   (\citet{takatsu2010wasserstein}).
\end{remark}

\begin{figure}
\centering

\subfloat[]{ \includegraphics[width=0.88\textwidth]{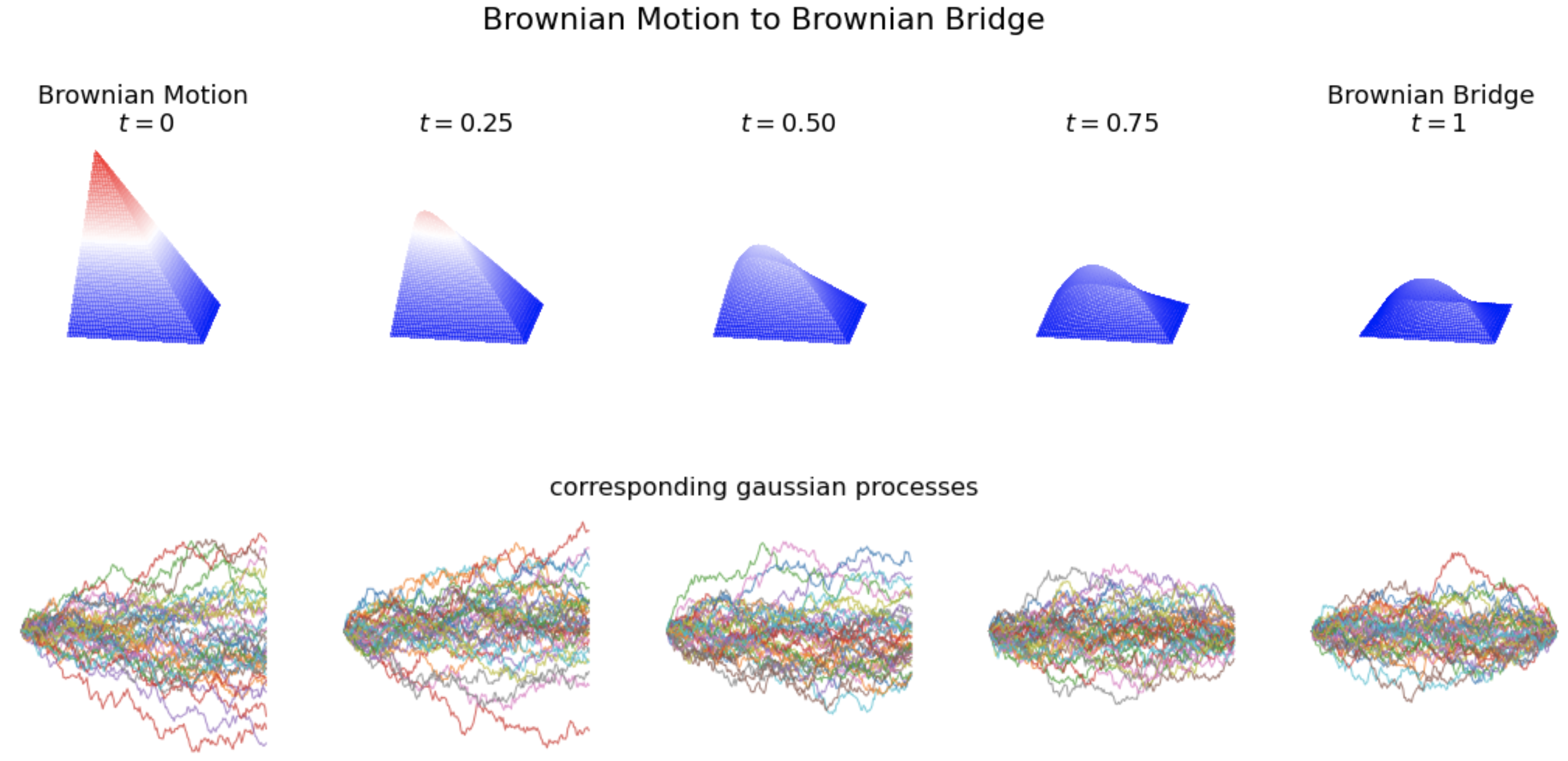}} 

\vspace{.5cm}

\subfloat[]{ \includegraphics[width=0.88\textwidth]{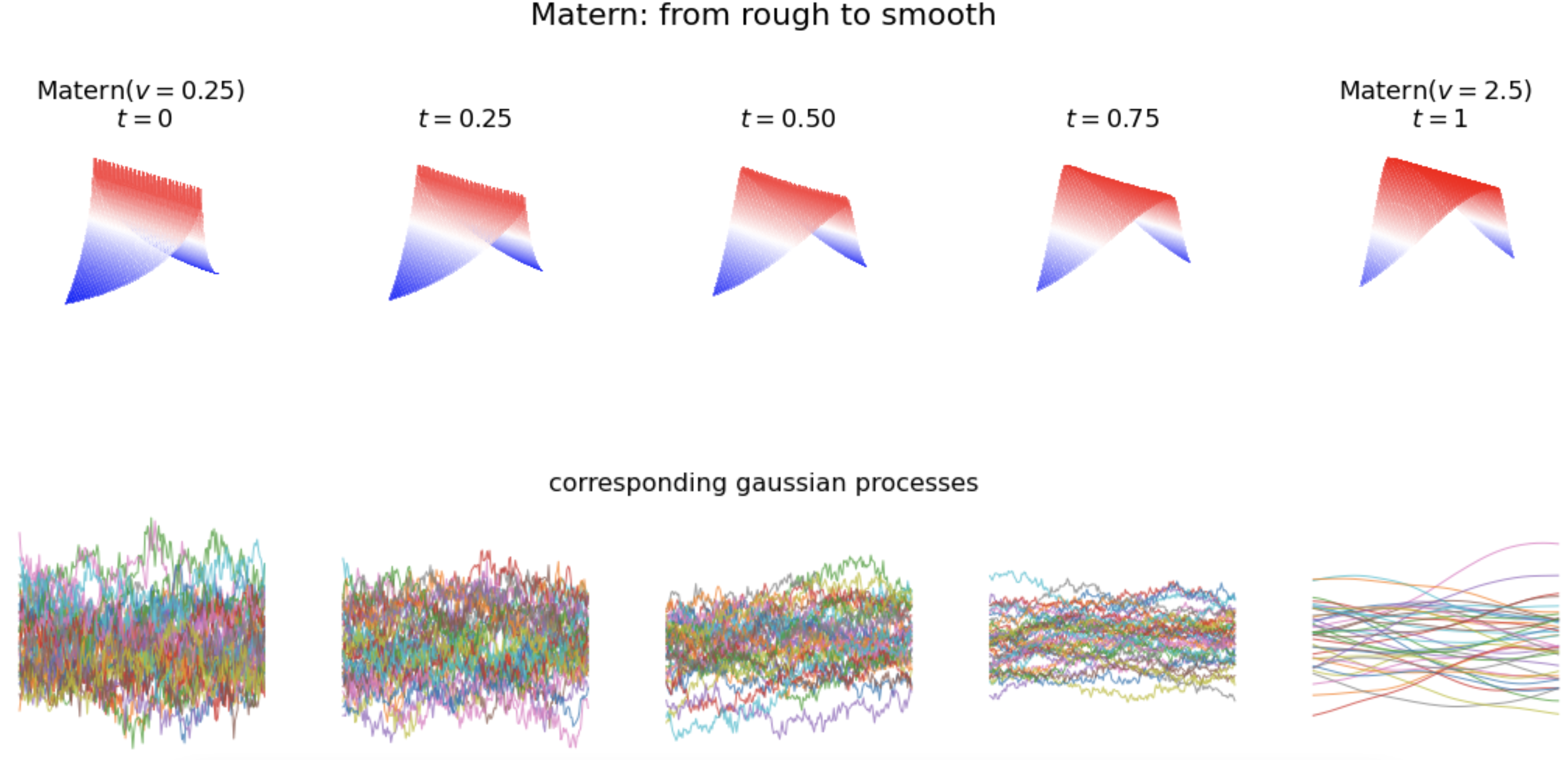}}

\caption[Two geodesic interpolations between covariances, with randomly generated samples of the corresponding gaussian processes]{(a) geodesic interpolation between the Brownian Motion covariance and the Brownian Bridge covariance, with ``spaghetti" plot of 50 gaussian processes with the corresponding covariance below.
    (b) geodesic interpolation between two Matèrn covariances, respectively with $\nu = 0.25$ and $\nu = 2.5$ degrees of freedom. Below, 50 gaussian processes with the corresponding covariance, showing a rough-to-smooth transition in the trajectories, consequence to the progressively smoother behaviour around the diagonal of the covariance.}
    \label{fig : spaghetti}
\end{figure}

\section{Functional Data in $(\KK,\Pi)$: Flows of Covariance Operators}
\label{sec : flows}
Armed with the geometry and rich analytical structure furnished by optimal transport, we now turn to consider a family $\{F_t\}_{t\in[0,1]}\subset \KK$ of covariances, indexed by a compact time interval, say $[0,1]$ without loss of generality.  We say that the map $\F : t \mapsto F_t$ is a \textit{flow of covariance operators} or a \textit{covariance-valued functional datum} if $F_t\in\KK$ for every $t\in[0,1]$,  and the map $t\mapsto F_t$ is measurable.
In particular, equipping $\KK$ with the metric $\Pi$, we say that a flow $\F=\{F_t\}_{t\in[0,1]}\subset \KK$ is continuous at $t\in[0,1]$ if:
$$
\Pi\big(F_{t+h},F_t \big) \rightarrow 0, \quad  \text{ as }h\rightarrow 0.
$$
We define  $\FF_C$ to be the space of \textit{continuous} flows of covariances:
$$
    \FF_{C} = \{ \F \::\: [0,1] \rightarrow \KK, \:\F\text{ continuous on } [0,1]\}.
$$
In particular, we may measure the distance between any pair $\F=\{F_t\}_{t\in[0,1]}$ and $\G =\{G_t\}_{t\in[0,1]}$ of continuous flows, $\F,\G\in\FF_C$, by integrating their (squared) pointwise $\Pi$-distance:
\begin{equation*}
\label{eq : distance between flows}
   d(\F,\G) := \left(  \int_{0}^1\Pi(F_t,G_t)^2 dt\right)^{\nicefrac{1}{2}}.
\end{equation*}
This yields a valid metric:
\begin{proposition}
\label{prop : space pf continous flows is a metric space}
$(\FF_C, d)$ is a metric space.
\end{proposition}

Continuity of a flow $\F$ allows us to make point-wise statements: at any fixed time point $t\in[0,1]$, the \textit{operator} $\F(t)=F_t\in\KK$ and associated notions (its kernel, its spectrum, and so on) are well-defined.

\begin{remark}[Additional Smoothness]
For the purposes of this paper, continuity suffices. But for other functional tasks, one occasionally asserts differentiability. For $p\in[1,\infty]$, we say that $\F =\{F_t\}_{t\in[0,1]}$ is $p$-absolutely continuous if there exists $v\in{L}^p(0,1)$ such that:
\begin{equation}
    \label{eq : AC flow}
    \Pi(F_s,F_t)\leq \int_s^t v(u)\der u, \qquad \forall \: s,t\in[0,1],
\end{equation}
 and denote the space of such flows by $  \FF_{AC^p}$. Absolutely continuous flows in \textit{complete} metric spaces are, in some sense, differentiable. Indeed, by \citet[Theorem 2.2 in ][]{ambrosio2004gradient}, the completeness of $(\KK,\Pi)$ shows that for any $\F \in \FF_{AC^p}$ there exists $|v'|\in L^p(0,1)$ satisfying:
\begin{equation}
    \label{eq : metric derivative}
    |v'|(t) = \lim_{s\rightarrow t} \frac{\Pi(F_s,F_t)}{|s-t|}, \qquad \forall \: t\in[0,1],
\end{equation}
where $|v'|$ is additionally a \textit{minimal} admissible integrand for the right hand side of \eqref{eq : AC flow}.
\end{remark}

Now let $\{F_t\}_{t\in[0,1]}$ be family of random covariance operators defined on a probability space $(\Omega,  \PP)$, and assume that the process $t\mapsto F_t$ has almost surely continuous trajectories and that for each $t\in[0,1]$ the mapping $\omega\mapsto F_t(\omega)\in \KK$ is measurable. Given such a collection of time-indexed covariances, we may define a \emph{random flow} $\F$ by $\F(t) = F_t$ for $t\in[0,1]$, which may be viewed as a \textit{random element} in the space $\FF_C$ of covariance-valued continuous functional data. {In the presence of our assumptions, it can also be seen as a stochastic process $\{F_t\}_{t\in [0,1]}$, with the two perspectives being interchangeable.

\subsection{Fr\'echet Mean Flow}
\label{sec : fmf}

Given a random flow $\F\in\FF_C$, we may define a corresponding \textit{\Frechet mean flow} with respect to the integrated Bures-Wasserstein metric $d$ defined in \eqref{eq : distance between flows}. That is, a \Frechet mean flow of $\F\in\FF_C$ is  a minimizer of the functional
\begin{equation}
\label{eq : pop frechet mean flow}
   \G \mapsto \E \left[ d^2(\G,\F) \right],\qquad \G \in \FF_C.
\end{equation}

A \Frechet mean flow need not exist nor be unique, if it does exist, but we will soon discuss sufficient conditions for existence and uniqueness, in which case we speak of \emph{the} \Frechet mean flow. As for any functional datum, estimating and constructing the mean is fundamental in itself, but also elemental for tasks such as regression, classification, and testing. It is, in particular, the first step in defining a \emph{covariance}, which we define in the next section, and is the second key ingredient for functional data analysis. 

For the purpose of statistical inference and computation, thus, we will make the following first-principles assumptions on the random flows $\F$:
\label{assumptions : flows}
\begin{enumerate}
    \item[\namedlabel{assumption_flow_trbound}{\textbf{A(1)}}]  $\E \int_{0}^1 \| \F(t)\|_1 dt < \infty. $
    \item[\namedlabel{assumption_flow_posdef}{\textbf{A(2)}}] $\PP\{\F(t) > 0 \} > 0, \quad \forall\: t\in[0,1]$.
    \end{enumerate}

The first assumption is the usual $L^1$-moment condition, a standard requirement for the existence of  a mean, but stated in the natural topology for our context. The second assumption is a regularity condition, asserting that at any $t\in[0,1]$, the flow has positive probability of attaining injectivity. These conditions will ensure, respectively, existence and uniqueness.

When a \Frechet mean  flow $\{\M(t)\}_{t\in[0,1]}$ does exist uniquely, we will occasionally use a second set of conditions which can be seen as more refined pair of moment and regularity conditions.

    \begin{enumerate}

    \item[\namedlabel{assumption_flow_boundedmap}{\textbf{B(1)}}]  The random maps 
        $T_{\M(t)}^{\F(t)}$  satisfy
     $\E \left\{\left\| T_{\M(t)}^{\F(t)}\right\|^2_{\infty}\right\}< \infty$ for every $t\in[0,1]$.
    \end{enumerate} 

\begin{enumerate}
 \item[\namedlabel{assumption_invertible_maps}{\textbf{B(2)}}]  
{ $\mathbb{P}\left\{ \varepsilon^{-1}\mathcal{R}(t)\preceq \F(t) \preceq \varepsilon \mathcal{R}(t) \right\}=1$, 
 for some $\mathcal{R}\in\FF_C$, some  $0<\varepsilon<\infty$, and all $t\in [0,1]$.}
\end{enumerate}

While the first pair of assumptions is used for existence/uniqueness of \Frechet mean flows, the second pair will be made use of when developing asymptotic theory, in the form of rates and limiting laws. 
Similarly to Assumption \ref{assumption_flow_trbound}, Assumption \ref{assumption_flow_boundedmap} is a finite moment assumption, but this time formulated on the tangent space of the \Frechet mean flow at time $t$. It states that the mean absolute deviation of $\F(t)$ (interpreted on the tangent space, as the image of $\F(t)$ via the log map at $\M(t)$) is finite in operator norm, at all time instances $t\in[0,1]$. {For example, this will be the case in the classical \emph{template deformation model}, or rather its dynamical version -- see model \eqref{eq : template}.}

As for Assumption \ref{assumption_invertible_maps}, like Assumption \ref{assumption_flow_posdef}, it   is a regularity condition on the possible ranges of the operators $\F(t)$, albeit a more refined one. It is equivalent to stating that the range of $\F^{1/2}(t)$ is constant over outcomes $\omega\in \Omega$ -- i.e.\ that every realisation of $\F$ generates the same Reproducing Kernel Hilbert Space (RKHS) at time $t$, almost surely. {The most natural case is when $\mathcal{R}(t)=R$ is time-invariant, and so the RKHS is invariant with respect to $t$ and $\omega$: the regularity is the same along the flow and over realizations. But invariance over $t$ is not necessary for our development, and indeed, one can imagine flows where the covariances are becoming increasingly or decreasingly regular as time progresses.}

When $\mathrm{dim}(\H)=\infty$, \ref{assumption_flow_trbound} follows from \ref{assumption_flow_boundedmap}, provided the latter holds true for \emph{some} flow $\M(t)$ (not necessarily a \Frechet mean flow). When $\mathrm{dim}(\H)<\infty$, the two conditions are equivalent (again with $\M(t)$ arbitrary in the formulation of \ref{assumption_flow_boundedmap}).

In general,  \ref{assumption_flow_posdef} does not follow from  \ref{assumption_invertible_maps}, except if  $\mathcal{R}$ is such that $\mathcal{R}(t)\succ 0$  for all $t\in[0,1]$. On the other hand, \ref{assumption_flow_posdef} and \ref{assumption_invertible_maps} \emph{together} will imply that the enveloping flow $\mathcal{R}$ can be taken to be everywhere strictly positive, implying in turn that $\F$ is everywhere strictly positive.

\medskip

Computing the \Frechet mean flow in \eqref{eq : pop frechet mean flow} in principle requires the solution of a minimisation problem over the space $\FF_C$ of continuous flows. However, continuity allows us to reduce the optimisation problem in \eqref{eq : pop frechet mean flow} to a \textit{pointwise} optimisation in the space of covariance operators, with respect to the Bures-Wasserstein distance, yielding the following result.

\begin{lemma}\label{lemma : continuity for BW barycenter}
Let $\F =\{F_t\}_{t\in[0,1]}$ be a random continuous flow satisfying Assumptions \ref{assumption_flow_trbound} and \ref{assumption_flow_posdef}. Given any $t\in[0,1]$, the pointwise minimiser 
$$
M_t : = \argmin_{G\in\KK}\:\E \left[ \Pi^2(G,F_t) \right].
$$
exists uniquely, and induces a continuous flow $\M\::\:t \mapsto M_t$ that uniquely minimises \eqref{eq : pop frechet mean flow} and is thus the \Frechet mean flow of $\F$.
\end{lemma}

In practice, Lemma \ref{lemma : continuity for BW barycenter} implies that the (unique) solution to \eqref{eq : pop frechet mean flow} can be obtained by means of \textit{pointwise} minimisation, ``stitching" the \Frechet means of $\F(t)$ over time, for $t\in[0,1]$. That is, the solution to \eqref{eq : pop frechet mean flow} evaluated at time $t$ may be expressed at $M_t = \argmin_{G \in\KK}\:\E \left[ \Pi(G, F_t )^2 \right]$.The actual computation relies on the point-wise use of a (provably convergent) Procrustes algorithm. See Section \ref{sec : computation}.

\section{Second-Order Structure of Random Flows}
\label{sec : secondorder}
The statistical analysis of random objects crucially relies on a notion of \textit{covariance} of the random object under study. For instance, when considering random elements in a separable Hilbert space, say $L^2(0,1)$, covariance operators encode the second-order variation structure, and are crucial in virtually all inferential tasks by way of the Karhunen-Lo\'eve expansion. This approach, however, is inextricably linked to the linear structure of the ambient space: traditionally, covariances of random objects are defined as the (Bochner) mean tensor product of the centred process with itself. The process of centering does not directly extend to non-linear spaces, such as the space of covariance operator flows introduced in the previous Section \ref{sec : flows}. And, even if centering is suitably defined (e.g. via lifting to the tangent space) it is not always possible to compare two centred objects as they might live on different spaces.

Fortunately, however, the availability of a tangent bundle for the Bures-Wasserstein space of operators allows for the definition of a tangent bundle in the space of corresponding operator flows. This, in turn, defines a centering operation by way of a suitable logarithmic map, and the subsequent definition of a tensor product of the logarithmic images of random flows with themselves. Finally, by a process of embedding in a common space, one can compare elements belonging to different tangent spaces, as is necessary for asymptotic theory. 

\medskip

The purpose of this section is to rigorously carry out this construction. The section consists of three parts. We first rigorously define the tangent bundle corresponding to the space of flows, and show that a flow $\F$ may be associated to a linear space  \eqref{eq:tensorspace}, which is in fact a Hilbert space (Proposition \ref{prop : tensor Hilbert space is complete and separable}).

Next, we introduce a tool that allows for comparing elements belonging to different tangent spaces. In the case of proper Riemannian manifolds, this is done via parallel transport, but such a tool is unavailable in our context. We thus develop a different solution, in the form of an isometric embedding \eqref{eq : embedding}. This embedding will be particularly useful in assessing the convergence of the empirical covariance operator to the corresponding population operator.

Then, we consider a logarithmic transformation mapping a given flow onto the tangent space at the \Frechet mean flow \eqref{eq : log-process}, and introduce the covariance operator of the image of random flow via the log transform \eqref{eq : log-process covariance}; this enables us to establish a Karhunen-Loève expansion for the random flow (Proposition \ref{prop:KarhunenLoeve-pop}) and the corresponding functional PCA.

\subsection{The Direct Integral Hilbert Space Along a Flow}
\label{subsec : tensor space}
We introduce a tangent bundle structure to the space of flows of covariance operators $\FF_C$. Heuristically speaking, we take an ``uncountable direct sum" of the tangent Hilbert spaces along a flow $\F\in\FF_C$. This is a well-defined notion known as the \textit{direct integral} of the corresponding tangent spaces: $\TT_{\F} = {\oplus}_{t\in[0,1]} \Tan_{\F(t)}$, see \citet[][Chapter 8]{takesaki2003theory}. This sort of construction was first made statistical use of in \citet{lin2019intrinsic}, in the context of flows on finite-dimensional Riemannian manifolds, under the name of \textit{tensor Hilbert space} (which, however, is a misnomer), and we will make similar use here. 
\medskip

Let $\V : t\mapsto \V(t)$ be a time-indexed family of elements in the tangent bundle $\TT_{\KK}$, as defined in \eqref{eq : tangent bundle}. We say it is a \textit{vector field} along the flow $\F\in\FF_C([0,1]$ if $\V(t) \in \Tan_{\F(t)}$ for all $t\in[0,1].$ It is straightforward to see that the set of all vector fields along a measurable flow $\F$ forms a vector space, with standard addition and scalar multiplication. Furthermore, given a vector field $\V$ along $\F$ we may define its norm $\|\V\|_\F$ as:
$$
\|\V\|_{\Tan_{\F}}
:=\left(\int_{0}^1\|\V(t)\|_{\Tan_{\F(t)}}^2 dt\right)^{\nicefrac{1}{2}}
=\left(\int_{0}^1 \tr( \V(t) \F(t) \V(t) ) dt\right)^{\nicefrac{1}{2}},
$$
where $\|\cdot\|_{\Tan_{\F(t)}}$ is the norm associated to the inner product  $\langle \cdot,\cdot \rangle_{F_t}$ on the tangent Hilbert space $\Tan_{\F(t)}$ as defined in \eqref{eq : inner product on tangent space}.

Given a flow $\F\in\FF_C$, we naturally define $\TT_{\F}$ to be the collection of measurable vector fields $\V$ along $\F$ satisfying $\|\V\|_{\Tan_{\F}}<\infty$:
\begin{equation}\label{eq:tensorspace}
    \TT_{\F} : = 
    \left\lbrace  \V \::\: \V(t) \in \T_{\F(t)} \text{ for } t\in[0,1], \: \int_{0}^1 \|\V(t)\|_{\Tan_{\F(t)}}dt  < \infty \right\rbrace,
\end{equation}
with identification between $\U,\V \in \TT_{\F}$ whenever the set $\{t\in[0,1]\::\: \|\V(t)-\U(t)\|_{\Tan_{\F(t)}} \neq 0 \}$ has null Lebesgue measure. We can endow $\TT_{\F}$ with the structure of an inner product space by means of the form
$$
\langle\U,\V\rangle_{\F} 
= \int_{0}^1\langle \V(t),\U(t) \rangle_{\Tan_{\F(t)}} dt
= \int_{0}^1\tr( \V(t)\F(t)\U(t)) dt,
$$
defined for all $\U,\V$ vector fields along $\F$, with induced norm $\|\cdot\|_{\Tan_{\F}}.$

\begin{proposition}
\label{prop : tensor Hilbert space is complete and separable}
The inner product space $(\TT_{\F}, \; \langle \cdot,\cdot \rangle_{\F}) $ is a separable Hilbert space.
\end{proposition}

\begin{remark}[Tangent Spaces] Note that we use $\TT$ generically to indicate a tangent space, with the index clarifying to what space it is tangent. In particular, $\TT_{F}$ represents a Hilbert space of operators tangent to $\KK$ at a covariance operator $F\in\KK$, whereas $\TT_\F$ represents a Hilbert space of operator flows tangent to $\FF_C$ at a covariance operator flow $\F\in\FF_C$.
\end{remark}

\subsection{Comparing Elements Across Tangent Spaces by Canonical Embedding}
\label{subsec : embedding}
In the case of Riemmanian manifolds, to compare elements at different tangent spaces, one may employ a tool intrinsic to manifolds known as \textit{parallel transport}. In some sense, parallel transport is a transformation between tangent spaces, that canonically transports a vector $v\in \mathrm{Tan}_p$ tangent to a point $p$, to a vector $\P_{p,q} v\in \mathrm{Tan}_q$ tangent to a point $q$, by  transporting it along the geodesic connecting $p$ and $q$ in \emph{parallel fashion}. In connection to our setting, such geometrical tool was employed for instance by \citet{lin2019intrinsic} and \citet{zhou2021intrinsic}.

\medskip

Although the Wasserstein space of measures is \textit{not} an infinite dimensional Riemmanian-manifold, its quasi-Riemannian nature allows for a similar geometric construction: in fact, \citet{ambrosio2008construction} have shown that a parallel transport map acting on the tangent bundle of $\WW_2(\R^d)$ does exists, and in the case of of $1$-d measures it is even available in closed form: see \citet{chen2021wasserstein} and \citet{zhou2021intrinsic} -- who consider functional data taking values in the 1-d Wasserstein space. However, while existence of parallel transport in the finite-dimensional case has indeed been established, it is only available implicitly, even in the Gaussian case (see for instance \citet{takatsu2010wasserstein}), and -- critically -- not in closed form.  Furthermore, it is not even clear whether the results in \citet{ambrosio2008construction} on parallel transport, fundamental in all the above mentioned literature, hold in the infinite dimensional setting. 

\medskip

Therefore, given covariances $F,G\in\KK$ and tangent vectors $U\in\Tan(F)$ and $V\in\Tan(G)$, it is not obvious how one may assess a notion of (dis)similarity between the tangent vectors $U$ and $V$, as these \textit{live} on different spaces, with different geometries. This poses a significant obstacle, and even more so when one considers elements belonging to different \textit{tensor} spaces, i.e.\ spaces in the tangent bundle structure we have introduced to the space of flows in the previous section, and will be particularly  relevant in the next section, when we will compare need to compare operators acting on different tensor spaces.

\medskip

\begin{figure}
    \centering
    \includegraphics[width= .7\textwidth]{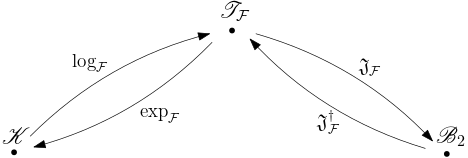}
    \caption{On the left side of the diagram, the logarithm map projecting $\KK$ onto the tangent space at a covariance operator $\F$, and its left-inverse, the exponential map at $\F$. On the right side of the diagram, the isometric embedding $\J$, identifying the tangent space $\Tan_\F$ with a (subset of) the space of Hilbert-Schimdt operators, $\BB_2$, and its left-inverse.}
    \label{fig:diagram}
\end{figure}

The solution we propose employs a similar geometrical strategy, and consists in canonically embedding elements of the tangent bundle into the space of Hilbert-Schmidt operators. For $F\in\KK$, considering the characterisation of the elements in the tangent space at $F$ given in \eqref{eq : tangent space}, we see that the following canonical embedding:
$$
\J_F \::\: \Tan_{F}\rightarrow\BB_2,\qquad \J_F U = U F^{\nicefrac{1}{2}}
$$
is well defined. In fact, it preserves the geometry of tangent vector:
$$
U,V \in \Tan_{F},\qquad 
\langle U,V\rangle_{\Tan_{F}} 
= \langle \J_F U, \J_F V \rangle_{2}
$$
where $\langle \cdot , \cdot \rangle_2$ denotes the Hilbert-Schmidt inner product.

\medskip

It is straightforward to then extend this construction to tensor Hilbert spaces: given a \textit{flow} of operators $\F\in \FF_C$, we define an embedding:
\begin{equation}\label{eq : embedding}
{\bJ}_\F \::\: \TT_{\F}\rightarrow L^2_{\BB_2},
\qquad ({\bJ}_\F \U)(t) = \U(t)\F(t)^{\nicefrac{1}{2}},
\end{equation}
where $L^2_{\BB_2}$ denotes the space of squared-norm integrable functions on the compact interval $[0,1]$ taking values in the space of Hilbert-Schmidt operators $\BB_2$:
$$
L^2_{\BB_2} = \left\lbrace \A \::\: [0,1] \rightarrow \BB_2 \::\: \int_{0}^1 \| \A(t)\|_{\BB_2}^2 dt <\infty \right\rbrace.
$$
which is a Hilbert space endowed with the inner product:
$$
\A,\B \in L^2_{\BB_2}, \qquad \langle \A,\B \rangle_{L^2_{\BB_2}}
=  \int_{0}^1 \langle \A(t), \B(t)\rangle_{\BB_2} dt = \int_{0}^1 \tr( \A(t) \B(t)^*) dt.
$$
Such construction thus defines an isomorphic embedding which provides a canonical tool for comparing vector fields at different tangent (tensor) spaces.

\subsection{Covariance and Karhunen-Loève Expansion of Random Flows}
\label{subsec : PCA}

Let $\F \: : \: t \mapsto \F(t)$ be a {random element} in $\FF_C$, i.e.\ a random flow of covariance operators. 
We wish to investigate the stochastic behaviour of $\F$ encoded in the second-order variation around its \Frechet mean flow. Our approach is consistent with that by \citet{dai2018principal}, \citet{lin2019intrinsic} and \citet{zhou2021intrinsic}, and leans on the quasi-Riemannian structure of $\KK$: we project the process $\F$ onto the tensor Hilbert space generated by its (Fr\'echet) mean $\M$, and consider the projected \textit{log-process}:
\begin{equation}
    \label{eq : log-process}
    t\mapsto \log _{\M(t)} (\F(t)) = T_{\M(t)}^{\F(t)} - \Id
    \in \Tan_{\M(t)},\qquad t\in[0,1],
\end{equation} 
For each $t\in[0,1]$, the log-process $\log _{\M(t)} (\F(t))$ takes values in the vector space $\Tan_{\M(t)}$. In fact, an important observation is that the log-process -- which we denote  by $\log _{\M} \F$ for simplicity -- is a random vector field along $\M$.

In some sense, the log-process originates from projecting the random flow along its mean trajectory. In fact, \citet[][Proposition 1]{santoro2023large} shows that:
\begin{equation}
    \label{eq : null average proj}
        \E T_{\M(t)}^{\F(t)} = I
    \quad \mu_{\M(t)} \text{- almost surely},\qquad \forall \: t\in[0,1]
\end{equation}

Consequently, the log-process has mean zero, and its (auto)covariance operator $\CC$ is
 defined by:
\begin{equation}\label{eq : log-process covariance}
\CC\::\:\TT_{\M} \rightarrow\TT_{\M},
\qquad \quad
\CC  = \E \left[ 
\log_{\M} \F \otimes 
\log_{\M} \F
\right]
\end{equation}
or, equivalently:
$$
\langle\CC U, V\rangle_{\M}:=
\mathbb{E}\left[
\langle\log_{\M} \F, U\rangle_{\M}
\langle\log_{\M}\F, V\rangle_{\M}
\right], \quad \text { for } U, V \in\TT_{\M}    
$$
Note that $\CC$ is nonnegative-definite and trace-class. Therefore, by \citet[Theorem 7.2.6][]{hsing2015theoretical}, $\CC$ admits a Mercer's decomposition in terms of its eigenvalue-eigenfuction pairs: 
$$
\CC=\sum_{k=1}^{\infty} \lambda_{k} {\Phi}_{k} \otimes {\Phi}_{k},
$$
where $\lambda_{1}>\lambda_{2}>\cdots>0$ are eigenvalues and ${\Phi}_{k}\:  : \TT_{\M}  \: \rightarrow  \TT_{\M}$ the corresponding eigenfunctions for $\CC$, forming a complete orthonormal system for $\mathscr{T}(\M)$.
In consequence, the log-process $\log _{\M} \F$ admits the following decomposition:
$$
\log _{\M} \F=\sum_{k=1}^{\infty} \langle \log _{\M}, {\Phi}_{k}\rangle_{\TT_{\M}} \boldsymbol{\Phi}_{k}
$$
where $( \langle \log _{\M} \F, {\Phi}_{k}\rangle_{\TT_{\M}})_{k\geq 1}$ are independent random variables with zero mean. 
In particular, by \citet[][Theorem 7.2.8]{hsing2015theoretical}, the above decomposition is optimal in the following sense:
$$
\lim_{N\rightarrow \infty} 
\E \| \log _{\M} \F - \sum_{k=1}^{N} \langle \log _{\M} \F, {\Phi}_{k}\rangle_{\TT_{\M}} {\Phi}_{k}\|_{\TT_{\F}}^2 = 0 
$$

Employing the canonical embedding we introduced in Subsection \ref{subsec : embedding} and the generalisation of Mercer's expansion proven in \citet{santoro2023karhunen}, we may in fact strengthen such result and obtain a \textit{pointwise} interpretation to the previous expansion, i.e.\ a Karhunen-Loève theorem, decomposing the covariance flows into uncorrelated components which achieve optimal dimensionality reduction uniformly over time.

\begin{proposition}[Karhunen-Loéve Expansion]\label{prop:KarhunenLoeve-pop}
Let $\F_1,\dots,\F_n$ be i.i.d.\ copies of $\F\in \FF_C$ satisfying \ref{assumption_flow_posdef}, \ref{assumption_flow_trbound} and \ref{assumption_flow_boundedmap}.
Define:
$$
\chi(t) := {\J}_{\M(t)}\log_{\M(t)} \F(t),
$$
Then, the following convergence holds uniformly:
$$
\lim_{N\rightarrow\infty}  \sup_{t\in[0,1]} \E  \left[\chi(t) - \sum_{j=1}^{N}
\langle\chi, \varphi_{j}\rangle_{L^2_{\H}}
\varphi_{j}(t)\right],
$$
where $\{\varphi_{j}\}_{j\geq 1}$ is a CONS of eigenfunctions of the integral operator associated to the operator-valued kernel $ K\;:\;[0,1]^2\rightarrow \BB_1(L^2_{\BB_1(\H)})$,
$
 K(s,t) := \E  \chi(s) \otimes \chi(t).
$
In particular, $\varphi_{j} = {\J}_{\M}\Phi_j, j\geq 1$.
\end{proposition}

\section{Estimation}
\label{sec : estimation}
\subsection{Infinite-Dimensional Case}

We now investigate the problem of estimation of the \Frechet mean and the spectral characteristics of a random flow of covariances, from independently and identically distributed (i.i.d.) replications. 
Given a sample of i.i.d.\ copies $\F_{1}, \cdots, \F_{n}$ of $\F\in\FF_C$, the population \Frechet mean flow is naturally estimated by a sample \Frechet mean flow, i.e. a minimizer of the empirical \Frechet functional:
 \begin{equation}
 \label{eq : sample frechet mean flow}
 \FF_C \ni \G \mapsto
 \frac{1}{n} \sum_{i=1}^{n} d^{2}\left(\F_{i}, \G\right).
 \end{equation}

As was the case with the population \Frechet functional, the optimisation problem \eqref{eq : sample frechet mean flow} may be reduced to a pointwise minimisation problem. For a given $t\in[0,1]$,  \citet[][Theorem 11]{masarotto2019procrustes} show that -- when one of the $n$ sample flows is non-degenerate at time $t$ -- there is a unique \Frechet mean operator at $t$, namely $\widehat\M_n(t) :=\argmin_{G \in \KK}\: \frac{1}{n} \sum_{i=1}^n \Pi^2(\F_i (t),G)$. By continuity and our assumption \ref{assumption_flow_posdef}, this allows us to obtain a continuous flow $\widehat\M_n\::\:t\mapsto \widehat\M_n(t)$, constituting the unique solution to \eqref{eq : sample frechet mean flow}, by means of \textit{pointwise} minimisation, followed by ``stitching" the \Frechet means of $\F_1(t),\dots,\F_n(t)$ over  $t\in[0,1]$. 

\begin{lemma}\label{lemma : empirical FMF via pointwise min}
Let $\{\F_i\}_{i=1}^n$ be an i.i.d.\ sample of random continuous flows in $\FF$, satisfying Assumptions \ref{assumption_flow_trbound} and \ref{assumption_flow_posdef}. The pointwise minimiser
$$\widehat{\M}(t)=\argmin_{G\in \KK} \frac{1}{n} \sum_{i=1}^n \Pi^2(\F_i (t),G) $$
exists uniquely for all sufficiently large $n$, with probability 1. It induces a continuous flow $t\mapsto \widehat\M(t)$ that uniquely minimizes \eqref{eq : sample frechet mean flow}, and is thus the empirical \Frechet mean flow of $\F_1,...,\F_n$.
\end{lemma}

With the existence and uniqueness of the empirical \Frechet mean flow at hand, we consider its consistency. We show that the empirical \Frechet mean flow $\widehat\M_n$ converges to the population \Frechet mean $\M$ with respect to the integral metric $d$. This essentially follows from the pointwise convergence result in Theorem \ref{thm : FM consistency and CLT} and continuity.

\begin{theorem}\label{thm : FM consistency and CLT}
Let $\F_1,\dots,\F_n$ be i.i.d.\ copies of $\F\in \FF_C$ satisfying \ref{assumption_flow_trbound} and \ref{assumption_flow_posdef} . 
Let  $\M$ and $\widehat\M_n$ be the population and sample barycenters, respectively 
defined in \eqref{eq : pop frechet mean flow} and \eqref{eq : sample frechet mean flow}. Then: 
\begin{equation}\label{eq:consistency}
 d(\M, \widehat\M_n)= o_{p}(1).
\end{equation}
If furthermore  \ref{assumption_flow_boundedmap} and \ref{assumption_invertible_maps} hold:
\begin{equation}
    \label{eq:consistency+rate}
d(\M, \widehat \M_n) = O_{p}(n^{-\nicefrac{1}{2}}).
\end{equation}
\end{theorem}

Having proven consistency of the empirical \Frechet mean flow, we can now consider the estimation of the covariance structure. By \citet[][Theorem 2]{masarotto2022transportation},  the  optimal map from the empirical \Frechet mean flow $\widehat\M_n(t)$ to each $\F_i(t)$ is a proper bounded operator for all $t\in [0,1]$. This guarantees that the log process  $t\mapsto \log _{\widehat{\M}_n(t)} (\F_i (t))$ is well-defined as a (mean-zero) random element in $\mathscr{T}_{\widehat{\M}}$.
 The random flow's covariance operator  \eqref{eq : log-process covariance} may thus be estimated by its finite-sample version:
\begin{equation}
\label{eq : estimated covariance}
\widehat \CC\::\:\TT_{\widehat\M_n} \rightarrow\TT_{\widehat\M_n}, \quad  \qquad 
\widehat{\CC}=\frac{1}{n} \sum_{i=1}^{n}\left(\log _{ \widehat{\M}} \F_{i}\right) \otimes\left(\log _{ \widehat{\M}} \F_{i}\right),    
\end{equation}
which in turn admits the Mercer's eigendecomposion $\widehat{\CC}=\sum_{k=1}^{\infty} \widehat{\lambda}_{k} \widehat{{\Phi}}_{k} \otimes \widehat{{\Phi}}_{k}$ with eigenvalues $\widehat{\lambda}_{1}>\widehat{\lambda}_{2}>\dots>0$ corresponding to the CONS of eigenfunctions $\widehat{{\Phi}}_{k}$, and Karhunen-Loève expansion for the log-process given by:
$$
\log _{\widehat\M_n} \F=\sum_{k=1}^{\infty} \widehat\xi_{k} \widehat{{\Phi}}_{k}.
$$

Establishing convergence of the Karhunen-Loève expansion is more involved, and essentially translates to showing the convergence of the empirical covariance operator $\widehat{\CC}$ to the population covariance ${\CC}$, which are operators acting on different tangent spaces, mainly $\TT_{\M}$ and $\TT_{\widehat\M_n}$, respectively.
However, we may lift the isometric embedding \eqref{eq : embedding} to define a linear mapping which converts a bounded operator acting on the tensor Hilbert space $\TT_{\F}$, for some $\F\in\FF_C$, to a bounded operator acting on $L^2_{\BB_2}$. 
We define such operation on rank one operators by ${\bJ}^{\otimes}_\F\mathscr{G}\::\:L^2_{\BB_2}\rightarrow L^2_{\BB_2}$ . For $U,V\in \TT_{\F}$ define:
$$
{\bJ}^{\otimes}_\F(U\otimes V) :=
{\bJ}_\F(U)\otimes  {\bJ}_\F(V).
$$
Such construction provides a canonical tool for comparing the population and sample log-process auto-covariance operators. Indeed, we may consider:
\begin{equation}\label{eq : pop - samp embedded covs}
{\bJ}^{\otimes}_\M\CC  = \E \left[ {\bJ}_\M\log_{\M} \F \otimes 
{\bJ}_\M\log_{\M} \F
\right],
\qquad
{\bJ}^{\otimes}_{\widehat\M}\widehat{\CC}=\frac{1}{n} \sum_{i=1}^{n}\left({\bJ}_{\widehat\M}\log _{ \widehat{\M}} \F_{i}\right) \otimes\left({\bJ}_{\widehat\M}\log _{ \widehat{\M}} \F_{i}\right),    
\end{equation}
which are operators acting on \textit{the same space}  $L^2_{\BB_2}$. In particular, ${\bJ}^{\otimes}_\M\CC $ and ${\bJ}^{\otimes}_{\widehat\M}\widehat\CC_n$ may be compared using the natural norm of the space ${\BB( L^2_{\BB_2})}$ of bounded linear operators on $L^2_{\BB_2}$.

\begin{theorem}\label{thm : cov consistency}
Let $\F_1,\dots,\F_n$ be i.i.d.\ copies of a time-varying covariance operator $\F\in \FF_C$ satisfying  \ref{assumption_flow_trbound}, \ref{assumption_flow_posdef} and \ref{assumption_flow_boundedmap}. Then:
\begin{equation}
    \label{eq:covariance_consistency}
    \big\Vert {\bJ}^{\otimes}_\M\CC  - {\bJ}^{\otimes}_{\widehat\M_n}\widehat\CC_n  \big\Vert_{{\BB( L^2_{\BB_2})}}= o_{p}(1).
\end{equation}
\end{theorem}

\subsection{Finite-Dimensional Case: Covariance Matrix Flows}

In this section we consider in detail the finite dimensional case $\dim(\H)=d<\infty$, i.e.\ of flows of covariance \textit{matrices}, explaining the simplification that arises compared to the functional case, and the stronger result that follow.

The space of finite-dimensional, non-negative definite covariance matrices does not constitute a Riemmanian manifold, except in the full rank case. It is rather a \textit{stratified} space in which the different strata correspond to the Reproducing Kernel Hilbert Spaces (RKHS) generated, 
 as explained in \citet{takatsu2010wasserstein}. Therefore, previous results as \citet{lin2019intrinsic} do not apply here. However, the finite dimensional Bures-Wasserstein space undoubtedly provides various simplifications when considering time-varying covariances. Firstly, in finite dimensions, strictly positive (regular) covariances are \textit{invertible}. The \Frechet mean of a family of covariances that are regular with positive probability is guaranteed to be regular (\citet{kroshnin2021statistical}). Optimal maps between regular covariances always define proper, bounded linear mappings. In fact, the finite dimensional optimal transportation problem is stable, in that optimal maps are \Frechet differentiable -- while even continuity often fails in infinite dimensions (\citet{kroshnin2021statistical}, \citet{santoro2023large}). Furthermore, all regular covariance matrices are equivalent, and  -- consequently -- all tangent spaces to regular covariances encompass the same set of elements. These elements drastically simplify the workload in the finite dimensional case, especially when dealing with elements in the tangent bundle.

\medskip
 Below, we state stronger results we may prove for time-varying covariance matrices, in particular concerning estimation of their \Frechet mean flows and second order operator.

\begin{theorem}\label{thm : FM consistency - finite dim}
Let $\F_1,\dots,\F_n$ be i.i.d.\ copies of a time-varying covariance matrices $\F\in \FF_C$ over $\R^d$ for some $1\leq d<\infty$, satisfying  \ref{assumption_flow_trbound} and \ref{assumption_flow_posdef}. Then:
\begin{equation}\label{eq:consistency-finite dim}
d(\M, \widehat\M_n)= O_p(n^{-\nicefrac{1}{2}}).
\end{equation}
Furthermore, if \ref{assumption_invertible_maps} is satisfied with $\mathcal{R}(t) = R\succ 0$, then:
\begin{equation}\label{eq:consistency-finite dim-uniform}
    \sup_{t\in[0,1]}\E \left[ \Pi\big(\M(t), \widehat\M_n(t)\big)^2\right]= O(n^{-1})
\end{equation}
\end{theorem}
The first part of the Theorem's statement follows similarly to the infinite dimensional case, but requires weaker assumptions. In the second part of the statement, note that the rate is uniform, rather then integral. Here, the assumption \ref{assumption_invertible_maps} allows us to bound the minimum and maximum eigenvalues of $\F(t)$ almost surely, uniformly over $t$, from above and below, away from $0$ and $\infty$ respectively. This, in particular, allows us to apply the finite-dimensional results in \citet{le2022fast} to deduce the uniform convergence result.

\medskip

Next we state the following stronger result concerning estimation of the second-order operator in finite dimensions.
\begin{theorem}\label{thm : cov consistency - finite dim}
Let $\F_1,\dots,\F_n$ be i.i.d.\ copies of a time-varying covariance operator $\F\in \FF_C$ over a $\R^d$ for some $1\leq d<\infty$, satisfying  \ref{assumption_flow_trbound}, \ref{assumption_flow_posdef} and \ref{assumption_flow_boundedmap}. Then:
\begin{equation}
    \label{eq:covariance_consistency_FINDIM}
    \big\Vert {\bJ}^{\otimes}_\M\CC  - {\bJ}^{\otimes}_{\widehat\M_n}\widehat\CC_n  \big\Vert_{{\BB( L^2_{\BB_2})}}= O_{p}(n^{-\nicefrac{1}{2}})
\end{equation}
\end{theorem}

The stronger rates of convergence we obtain in the finite-dimensional case follow from the following result.

\begin{lemma}\label{lemma : smoothness of emblog}
Let  $F$ be a covariance operator. Then the functional:
$$
M^{\nicefrac{1}{2}}\mapsto \J_M\log_M F
$$
defined onto $\BB_2$ has densely defined \Frechet differentiable at any $M\succ 0$, with derivative given by:
$$
{d_{M^{\nicefrac{1}{2}}}}[\J_M\log_{M}F](H) =T_M^F H - M^{-\nicefrac{1}{2}}\der_{ Q^2}g (QH M^{-\nicefrac{1}{2}}Q + QM^{-\nicefrac{1}{2}} HQ) - H
$$
In particular, if the optimal map $T_M^F$ is bounded, we have that:
$$
\|{d_{M^{\nicefrac{1}{2}}}}[\J_M\log_{M}F]\| \leq 3+\|T_M^F\|.
$$
\end{lemma}
In finite dimensions, the above theorem yields that the (embedded) logarithm map \textit{always} defines a smooth function. To the contrary, in infinite dimensions such mapping has only densely defined derivative which does not, in general, constitute a proper bounded functional, as the optimal transport map appearing above might be unbounded.

\medskip

Finally, we make two observations regarding estimation of the eigendecomposition of the second-order process $\CC,$ which are direct consequence of the perturbation results  in \citet[][Lemma 4.2, Theorem 4.4 and Lemma 4.3]{bosq2000linear}.
\begin{remark}
    Denote by $\{\varphi_k,\lambda_k\}_{k\geq 1}$ and $\{\widehat\varphi_k,\widehat\lambda_k\}_{k\geq 1}$ the eigenvalue/eigenfunction pairs corresponding to $ {\bJ}^{\otimes}_\M\CC, {\bJ}^{\otimes}_{\widehat\M_n}\widehat\CC_n  $ , with the latter oriented so that $\mathrm{sign}(\langle \varphi_k,\widehat\varphi_k\rangle) =1$. Then:
    \begin{equation}\label{eq : covflow cov consistency : eig-values}
    \sup_{k\geq 1} |\lambda_k -\widehat\lambda_k | = O_p(n^{-\nicefrac{1}{2}}).
    \end{equation} 
    Let $k\geq 1$ such that $\lambda_k$ has multiplicity one, i.e\ the corresponding eigenspace is one 
    dimensional. Then, setting $\Delta_k:= \min\{ (\lambda_{k-1}- \lambda_{k}), \: (\lambda_{k}- \lambda_{k+1})\}$:
    \begin{equation}
        \|\varphi_k - \widehat\varphi_k\|_{L^2_{\BB_2}}
        = O_p(\Delta_k^{-1}n^{-\nicefrac{1}{2}}).
    \end{equation} 
\end{remark}

\section{Irregular observations}
\label{sec : irregular obs}
In this section we show how the methods we developed in the case of fully observed flows may be adapted to more realistic data-collection scenarios: indeed, while in the exposition above we considered the covariance flows to be fully and perfectly observed, real-data is often available only discretely -- densely or sparsely -- while noise corruption is a reality that cannot be escaped and should be taken into account when developing a statistical method.

\subsection{Discrete Observations}

When functional observations are available discretely over the interval $[0,1]$ but the grid is sufficiently \textit{dense}, a standard approach is that of smoothing the observations to recover the trajectories.

In our context, let $\F$ be some flow, which we assume to only observe on a grid 
$ \{T_i\}_{i=1,\dots,r}$ of randomly distributed points over $[0,1]$, and furthermore assume to have only access to \textit{estimates} $\hat F_i$ of the covariance $\F(T_i)$, for $i=1,\dots,r$. This naturally occurs for instance when covariances are constructed on the basis of a finite sample of the corresponding processes,  under complete, regular or sparse observation sampling regimes. That is, in practice -- for instance when data is coming from fMRI data \cite{dai2019analyzing} -- each covariance is $\F(T_i)$ is estimated on the basis of $K$ identical realisations $\{X^i_\ell\}_{\ell=1,\dots,K}\in L^2([0,1],\R)$ say, with covariance $\E [X^i_\ell\otimes X^i_\ell ] = \F(T_i)$. On the basis of these samples, we usually have an estimator $\widehat F_i$ of the latent covariance $\F(T_i)$ satisfying 
\begin{equation}
    \label{eq:covobsrate}
\E \| \widehat F_i - \F(T_i) \|_1 = \mathcal{O}(\mathcal{R}(K)),
\end{equation}
with convergence rate depending on the specific observation regime of the curves $X_\ell^i$, but generally decreasing as $K$ diverges. For instance, assuming the observations $X^i_\ell$ are independent across indices, if each trajectory  $X^i_\ell$ is observed {fully}, across the time interval $[0,1]$ with no noise corruption, then $\mathcal{R}(K) =K^{-1/2}$ by the standard Central Limit Theorem. 
Of course many other rates, corresponding to other observational regimes (dependent curves, discrete/sparse observations, functional fragments,...) can arise.

\medskip

Regardless, if a flow $\F \::\: [0,1]\rightarrow \KK$ is observed via proxies $\hat F_1,\dots,\hat F_K$ satisfying \eqref{eq:covobsrate} corresponding to (random) timepoints $T_1,\dots,T_K$, then
-- as in classical Functional Data Analysis -- the flow $\F$ may be estimated by smoothing the perturbed observations  over the grid of points $\{T_i\}_{i=1,\dots,K}$, for instance by the classical Nadaraya-Watson estimator:
\begin{equation}\label{eq:NW}
\hat{\F}^{\mathrm{NW}}_r(t)= \frac{1}{ \sum_{i=1}^r W_h\left(T_{i}-t\right)} \sum_{i=1}^r W_h\left(T_{i}-t\right)\hat F_i.
\end{equation}
where $W_h(\cdot) = \frac{1}{h}W(\cdot/h)$ for some positive kernel $W(\cdot)$ and bandwith parameter $h$.
Some simple math shows consistency of this estimator as $r, K$ diverge, provided some regularity conditions are satisfied.

\begin{proposition}\label{prop:NWsmoothing}
    Assume that the flow $\F$ is Lipschitz continous with constant $L$, the kernel $W(\cdot)$ is compactly supported and the sampling times $T_j$ are uniformly distributed on the time inerval. Then, if $h_r \asymp r^{-\nicefrac{1}{3}}$:
    $$
    \E\left[ d(\F(t), \hat{\F}_r^{\mathrm{NW}})^2\right]  = 
    \mathcal{O}\left(\mathcal{R}(K) + r^{-\nicefrac{2}{3}}\right).
    $$
\end{proposition}

    \subsection{Sparse Longitudinal Data}
In practice, the discrete and irregular nature of functional observation may be more extreme, and only longitudinal and \textit{sparse} observations could be available. Indeed, assume to have access to $n$ random flows $\F_1,\dots,\F_n$ through observations:

\begin{equation}\label{eq:sparseobs}
F_{ij} =  F(T_{ij}), \qquad i=1,\dots,n, \quad j = 1,\dots, r(i)
\end{equation}

Such an observational regime does arise in practice for covariance-valued data  (see \cite{lin2019intrinsic}) and is in fact quite common in functional connectivity studies: ``\textit{for most longitudinal biomedical and social studies, data are not continuously observed in time, and more often than not the time points at which measurements are available are irregular and sparse}" \cite{dai2018principal}.

\medskip

We revisit the principal component analysis through conditional
expectation (PACE) for longitudinal data due to \citet{yao2005functional} and in particular its adaptation to manifold-valued data by \citet{dai2021modeling}.
This allows us to conduct principal component analysis for sparsely observed  longitudinal covariance flows, and consequently recover  best fits for individual trajectories  on the basis of the estimated principal components.

\medskip

We will assume the following:
\begin{enumerate}
    
    \item[\namedlabel{assumption_sparse_time}{\textbf{S(1)}}] $\left\{T_{i j}\right\}$ is the triangular array of the ordered random design points $\left(T_{i k}<T_{i j}\right.$, for $1 \leq k<j \leq r(i)$, and $i=1,2, \ldots, n)$, each drawn independently from a uniform density on $[0,1]$.
    
    \item[\namedlabel{assumption_sparse_r}{\textbf{S(2)}}] $\{r(i)\}$ is the sequence of grid sizes, determining the denseness of the sampling scheme, with $r_i\geq r\geq 2.$

    \item[\namedlabel{assumption_sparse_independence}{\textbf{S(3)}}]  $\left\{\F_i\right\}$ and $\left\{{T}_{i j}\right\}$  are totally independent across all indices.
\end{enumerate}

\medskip

As in classical FDA, the \Frechet mean flow may be estimated by a \textit{pool and smooth} approach, longitudinally grouping observations and smoothing. For this we will resort to  
 Local \Frechet Regression  (LFR, see \citet{petersen2019frechet}):
 \begin{equation}\label{eq:LFRMeq}
\hat\M^{LFR}(t) = \argmin_{G\in\KK}\sum_{i=1}^n\sum_{j=1}^{r(i)} s(T_{ij},t,h )\Pi^2(G,F_{ij}).
 \end{equation}
which generalises local linear smoothing to general metric spaces .
Here
$s\left(T, t ,h \right) := \frac{1}{\mu_0\mu_2-\mu_1^2}K_h(T-t)[\mu_2-\mu_1(T-t)]$ and $\mu_j = \frac{1}{r} \sum_{i=1}^r K_h(T-t)(T-t)^j$, 
and we refer to the original reference for the details motivating this expression.

\begin{proposition}\label{prop:loclin FM from sparse}
  Let $\F_1,\dots,\F_n$ be random i.i.d.\ flows  satisfying  \ref{assumption_flow_trbound}, \ref{assumption_flow_posdef}, \ref{assumption_flow_boundedmap} and \ref{assumption_invertible_maps}. Assume that the flows are observed sparsely as in \eqref{eq:sparseobs}, and satisfy
\ref{assumption_sparse_time},\ref{assumption_sparse_r},\ref{assumption_sparse_independence}. Let $K$ be compactly supported. Then:
$$\Pi(\hat\M^{LFR}(t)  , \M(t))\: \overset{p}{\longrightarrow} \:0 , \qquad \forall\:t\in[0,1], \quad \text{as } n\rightarrow\infty, h\rightarrow 0.$$
In particular if $\mathrm{dim}(\H)<\infty$:
  $$\sup_{t\in[0,1]}\Pi(\hat\M^{LFR}(t)  , \M(t)) = O_p(h^2 + \sqrt{nh}).$$
\end{proposition}
Note that the rate of convergence is the usual one for local regression with real valued responses.

\begin{remark} 
The smooth estimator is defined in terms of the local linear estimate suggested by  \citet{petersen2019frechet}. However, our setting escapes the assumptions under which they prove consistency. In fact, our strategy of proof for establishing a rate of convergence is distinctly different and exploits the additional structure inherent to our context, which in turn allows us to avoid making difficult and opaque assumption on the data generation process.
\end{remark}

For estimating the  principal components and principal components scores, we proceed as follows.
First, observe that the result in \citet{ocana1999functional} yields that the PCA of the flows $\F_1,\dots,\F_n$ in the tangent bundle is equivalent to the PCA of the Hilbertian embedded log-projected flows $\bJ_{\M}\log_{\M}\F_1,\dots\bJ_{\M}\log_{\M}\F_n \in \L^2_{\BB_2}$, in the sense that the eigenvalues (i.e., the principal component scores) remain the same, and the eigenvectors (i.e., the principal components) change predictably according to the embedding map. This  motivates us to cosider the (perturbed) embedded logarithmic projections:
$$
\chi_{ij} := \J_{\hat M_{ij}}\log_{\hat M_{ij}}F_{ij} \in \BB_2,  \qquad i=1,\dots,n, \quad j = 1,\dots, r(i)
$$
where $\hat M_{ij} := \hat \M^{LFR}(T_{ij})$.
It is then natural to estimate the corresponding covariance by scatterplot smoothing of the tensor products. That is, consider
$\Gamma_{ij\ell} := \chi_{ij} \otimes \chi_{i\ell}$ for $i=1,\dots,n$ and $j,\ell = 1,\dots,r(i)$, and the estimator $\hat\Gamma(s,t):= \widehat A_0$ where:
\begin{equation}\label{eq:covPaceEq}
 \left(\hat{A}_0, \hat{A}_1, \hat{A}_2\right) :=\argmin_{A_0, A_1, A_2} \sum_{i=1}^n \sum_{ j \neq l } K_{h_{\Gamma}}\left(T_{i j}-s\right) K_{h_{\Gamma}}\left(T_{i l}-t\right)
\times\left\|\Gamma_{i j l}-A_0-\left(T_{i j}-s\right) A_1-\left(T_{i l}-t\right) A_2\right\| .
\end{equation}
where $h_{\Gamma}>0$ is a bandwidth.

\medskip

 \begin{proposition}\label{prop:pace}
 Let $\F_1,\dots,\F_n$ be random i.i.d.\ flows  satisfying  \ref{assumption_flow_trbound}, \ref{assumption_flow_posdef}, \ref{assumption_flow_boundedmap} and 
\ref{assumption_invertible_maps}. Assume that the flows are observed sparsely as in \eqref{eq:sparseobs}, and satisfy
\ref{assumption_sparse_time},\ref{assumption_sparse_r},\ref{assumption_sparse_independence}. Then, if $\mathrm{dim}(\H)<\infty$:
 $$
 \sup _{s, t \in \mathcal{T}} \| \hat{\Gamma}(s, t)- \Gamma(s, t) \|^2=
O_P\left(h_{\M}^4  + {h_{\M}n} + h_{\Gamma}^4 +   \frac{\log nh_{\Gamma}^{-4}}{n}\left(h_{\Gamma}^4+\frac{h_{\Gamma}^3}{r}+\frac{h_{\Gamma}^2}{r^2}\right)\right)
$$
where $h_{\M}$ and $h_{\Gamma}$ respectively denote the bandwidths parameters in \eqref{eq:LFRMeq} and \eqref{eq:covPaceEq}. 
If $\mathrm{dim}(\H)=\infty$ then $\hat{\Gamma}(s, t)$ is pointwise consistent as $n\rightarrow \infty$.
\end{proposition}

By the Karhunen-Loéve theorem for hilbertian functional data, we may write:
$$
\bJ_{\M}\log_{\M}\F_i =  \sum_{j=1}^{N}
\xi_{ij}
\varphi_{j}(t), \qquad i=1,\dots,n
$$
where $\xi_{ij} := \langle\bJ_{\M}\log_{\M}\F_i, \varphi_{j}\rangle$ as we proved in Proposition \ref{prop:KarhunenLoeve-pop}. The idea of the PACE method is to construct estimates for the full trajectories on the basis of a truncation parameter $K$ -- which may be chosen according to the corresponding fraction of variation explained -- and estimates for the scores $\{\xi_{ij}\}_{j=1,\dots,K}$, and the eigenfunction $\varphi_{j}$.
On one hand, estimates for the eigenfunctions $\varphi_j$ of the integral operator corresponding to $\Gamma$ are then obtained by the eigenfunctions $\hat{\varphi_j}$ of the integral operator corresponding to  $\hat{\Gamma}$. On the other, the scores $\xi_{ij}$ may be estimated from sparse data by best linear unbiased prediction.

\section{Computational Aspects}\label{sec : computation}

\subsection{Discretisation}
In practice data are collected in discrete form, both in time and in space, and we now describe how the data are represented/manipulated in a typical discretisation scenario. The first discretisation is relative to the time-interval over which the flow is measured: a continuous flow $\F\::\:[0,1]\rightarrow \KK_{\H}$ over an (infinite dimensional) Hilbert space $\H$ will be measured at a discrete grid of time points $\vartriangle = \{t_k\}_{k=1, \dots, \lvert \vartriangle \rvert}$, which may be for instance equally spaced over the time interval $[0,1]$. Hence, a flow $\F$ will be measured as 
$$
 \begin{bmatrix}
F_1, & \dots & F_{|	\vartriangle|} \end{bmatrix}\in \KK_{\H}^{\lvert \vartriangle \rvert},
$$ for some grid $	\vartriangle$ and where we interpret $F_j=\F(t_j)$ for $ j=1,\dots,\lvert \vartriangle\rvert$.

\medskip 

The discretisation in space, on the other hand, often arises by basis representation. Indeed, let $F_j$ represent the covariance operator of some (mean-zero) random process in $\H$, say $\mathsf{X}_j$: that is, $F_j = \E[\mathsf{X}_j\otimes \mathsf{X}_j]$. A discretisation in space arises naturally when considering a finite representation over a set of basis elements $\{e_1,\dots,e_d\}\subset\H$, i.e.\ by considering the multivariate process given by the decomposition of $\mathsf{X}_j$ on the basis elements 
 $\begin{bmatrix}
X_{j,1} & \dots & X_{j,d} \end{bmatrix}\in \H^{d}$, where $X_{j,k} = \langle \mathsf{X}_j, e_k\rangle$ for $k=1,\dots n$. Of course, there are many different valid ways to do this: one can use pre-specified bases, such
as B-splines, or empirical bases corresponding to Karhunen-Lo\`eve type expansions, for instance. Regardless, once a choice of (complete) basis is made, a discretisation in space of $F_j$, for $j=1,\dots,\lvert \vartriangle \rvert$ arises via the outer products 
$$
 F_j^{(d)} :=\E \left(\begin{bmatrix}
X_{j,1} \\ \vdots \\ X_{j,d} \end{bmatrix}
\begin{bmatrix}
X_{j,1} & \dots & X_{j,d} \end{bmatrix}
\right),
$$ 
which discretise the tensor product $\E[\mathsf{X}_j\otimes\mathsf{X}_j]$.
 Note that, setting $\mathscr{P}_d$ the projection onto the span of the basis elements  $\{e_1,\dots,e_d\}$, we have that:
$$
 F_j^{(d)} 
=
\mathscr{P}_d \F(t_j) \mathscr{P}_d^*, \qquad j=1,\dots,\lvert\vartriangle\rvert.
$$

To summarise, given a flow $\F$ on Hilbert space $\H$ over the time interval $[0,1]$, given a discrete grid of time points $\vartriangle = \{t_k\}_{k=1, \dots, \lvert \vartriangle \rvert}$ and given a finite set of basis elements in $\{e_1,\dots,e_d\}\subset\H$, a discretisation of $\F$ in time and spaces arises as:
$$
\begin{bmatrix}
 F_1^{(d)} & \dots ,&  F_{\lvert \vartriangle \rvert}^{(d)}
\end{bmatrix}
$$
where each $ F_j^{(d)}$ is a $d\times d$ dimensional matrix approximating $\F(t_j)$, for $j= 1,\dots,\vert\vartriangle\rvert.$

\subsection{Approximation}
Given a discretised covariance flow, the natural question is how well this approximates the latent continuous (in time and space) flow. Let $\F,\G\in \FF_C$. We discretise them according to some grid $\vartriangle$ on $[0,1]$, which we assume for simplicity to be comprised of equally spaced points, and some choice of orthonormal basis elements $\{e_1,\dots,e_d\}$ as described in the paragraph above:
$$
\begin{bmatrix}
 F_1^{(d)} & \dots ,&  F_{\lvert \vartriangle \rvert}^{(d)} 
\end{bmatrix}
\quad \text{and} \quad\begin{bmatrix}
 G_1^{(d)} & \dots ,& G_{\lvert \vartriangle \rvert}^{(d)} 
\end{bmatrix}.
$$
We claim that, for $\vartriangle$ sufficiently dense and $d$ sufficiently large, the quantity $\frac{1}{\lvert \vartriangle \rvert}\sum_{i=1}^{\lvert \vartriangle \rvert}\Pi(
F_i^{(d)}, G_i^{(d)})^2$ is a good approximation of $d(\F,\G)^2$.
In other words, that our metric is \emph{stable to discretisation}:
$$
\forall \: \delta> 0 \: \exists \: \lvert\vartriangle\rvert,d<\infty \::\: 
\left\lvert
\frac{1}{\lvert\vartriangle \rvert}\sum_{i=1}^{\lvert \vartriangle \rvert}\Pi(
F_i^{(d)}, G_i^{(d)})^2 - d(\F,\G)^2 \right\rvert < \delta.
$$
By standard approximation of integrals by Riemmanian sums, it is straightforward to see that:
$$
\left\lvert
\frac{1}{\lvert\vartriangle (\delta)\rvert}\sum_{i=1}^{|\vartriangle(\delta)|}\Pi(
\F(t_i), \G(t_i))^2 - d(\F,\G)^2 \right\rvert <\delta 
$$
where $\vartriangle(\delta)$ depends on  $\delta>0$, which can be chosen arbitrarily small.
Next, observe that for all $i=1,\dots,\lvert\vartriangle\rvert$:
$$
\lim_{d\rightarrow\infty}\Pi(F_i^{(d)}, G_i^{(d)}) = \lim_{d\rightarrow\infty}\Pi(\mathscr{P}_d \F(t_i) \mathscr{P}_d^*, \mathscr{P}_d \G(t_j) \mathscr{P}_d^*) = \Pi(\F(t_i), \G(t_i))
$$
This can be easily seen in consequence to \citet[][Lemma 5]{masarotto2019procrustes}, together with the continuity of Wasserstein distances.
In particular, for we may choose $d$ large enough so that:
$$\max_{i=1,\dots,\lvert\vartriangle\rvert}\left\lvert\Pi(F_i^{(d)}, G_i^{(d)}) - \Pi(\F(t_i), \G(t_i)\right\rvert<\delta,$$ which establishes an approximation result. In particular, if the trajectories of the flows $\F,\G$ almost surely lie within a compact set -- for instance when Assumption \ref{assumption_invertible_maps} is satisfies for some constant flow $\mathcal{R}(t) \equiv R$ -- then such convergence holds uniformly over all possible discretisations.

\medskip

Now consider $n$ flows $\F_1,\dots,\F_n$, with (empirical) \Frechet mean flow $\M_n$, and the corresponding $d$-dimensional projections 
$\{F^{(d)}_{i,j}\}_{j=1,\dots, \vartriangle}$ for $i=1,\dots,n$,
 where $F^{(d)}_{i,j} = \mathscr{P}_d \F_i(t_j)\mathscr{P}_d$ for $i=1,\dots,n$ and $j=1,\dots,\lvert \vartriangle \rvert$, with corresponding \Frechet means given by $M^{(d)}_{n,j}$.

 A natural question is then whether $\Pi(M^{(d)}_{n,j}, \M_n(t_j)) \rightarrow 0$ as $d\rightarrow\infty$. We have a positive answer: for every fixed time instant $j \in \vartriangle$, as the dimension of the projections $d$ diverges to infinity, the \Frechet mean of the finite dimensional projections $\{F_{i,j}^{(d)}\}_{i=1,\dots,n}$ converges to the \textit{true} \Frechet mean of the \textit{operators} $\{ \F_{i}(t_j)\}_{i=1,\dots,n}$: see \cite[][Theorem 11]{masarotto2019procrustes}.

\medskip

In particular, in the framework described above it is natural to define a \textit{flow} $\tilde\M_n^{(d)}$ interpolating of the  values $\{M_{n,j}^{(d)}\}_{j=1,\dots,\vert\vartriangle\rvert}$, for instance employing Mc'Cann's interpolant \eqref{eq : McCann's interpolant}.
Let $\varepsilon>0$ be arbitrary. Then, provided the discretisation grid $\vartriangle$ is sufficiently dense, standard approximation of integrals by Riemmanian sums yields that:
$$
d(\M_n, \tilde\M_n^{(d)})^2 \leq \frac{\varepsilon^2}{2}  
+ \lvert \vartriangle\rvert \sum_{j=1,\dots,\vert\vartriangle\rvert} \Pi(\M_n(t_j), \tilde\M_{n}^{(d)}(t_j))^2.
$$
for such grid,  which we remark to be fixed and of finite size, the stability argument above guarantees the existence of $d$ large enough so that:
$$
\max_{j=1,\dots,\vert\vartriangle\rvert} \Pi(M^{(d)}_{n,j},\M_{n}(t_j) \leq \frac{\varepsilon}{\sqrt{2 \vert\vartriangle\rvert}},
$$
and putting the two equations above together yields that
$$
d(\M_n, \tilde\M_n^{(d)}) \leq \varepsilon.
$$ In other words, for any arbitrary small $\varepsilon>0$ the there exists a sufficiently dense grid $\vartriangle$ and large enough projection size $d$ such that the latent empirical flow may be approximated at $\varepsilon$-distance by discretisation and interpolation.

\subsection{Optimisation}
Our methodology relies on the estimation of the \Frechet mean flow of a collection of i.i.d.\ covariance flows as a major stepping stone. 
 As per Lemma \ref{lemma : empirical FMF via pointwise min},  such optimisation problem may be reduced to pointwise minimisation. Therefore, given the discretised data,  we need to be able to compute the \Frechet mean of finite collection of $n$ covariance matrices, and repeat this over all grid points. That is, once the flows $\F_1,\dots,\F_n$ have been discretised on some grid $\vartriangle = \{t_j\}_{j=1\dots,\lvert \vartriangle \rvert}$, and are thus represented by the sequences of matrices
 $$
 \{F^{(d)}_{1,j}\}_{j=1,\dots, \vartriangle}, 
 \quad \dots,\quad 
 \{F^{(d)}_{n,j}\}_{j=1,\dots, \vartriangle},
 $$ 
 where $F^{(d)}_{i,j} = \mathscr{P}_d \F_i(t_j)\mathscr{P}_d$ for $i=1,\dots,n$ and $j=1,\dots,\lvert \vartriangle \rvert$, the discretised \Frechet mean flow at time $t_j$, namely $M^{(d)}_{n,j}$, is computed as the \Frechet mean of the matrices $\{F^{(d)}_{i,j}\}_{i=1,\dots,n}$.

\medskip

Though the \Frechet mean of such covariances does not admit a closed form expression, the gradient of the \Frechet functional does, allowing for a (provably convergent) gradient descent procedure. See \citet{zemel2019frechet}.  
Below we give the pseudo-code for the Gradient Descent algorithm (Algorithm \ref{ALG:GD}) for computing the \Frechet mean $\M^{(d)}$ of $n$ covariances $F_1^{(d)},\dots,F_n^{(d)}$.
\begin{center}
\noindent\begin{minipage}{0.68\textwidth}
\begin{algorithm}[H]
  \caption{Gradient Descent}\label{ALG:GD}
  \begin{algorithmic}[1]
    \Procedure{GD}{$M_0^{(d)}, \{ F_1^{(d)},\dots, F_n^{(d)}\}, K$}
      \For{$k = 1, \ldots, K$}
        \State $T_k^{(d)} \gets n^{-1} \sum_{i=1}^n T^{(d)}_{k,i}$
        \State $M_{k+1}^{(d)}\gets T_k^{(d)} M_k^{(d)} T_k^{(d)}$
      \EndFor
      \State \textbf{return } $M_K^{(d)}$
    \EndProcedure
  \end{algorithmic}
\end{algorithm}
\end{minipage}
\end{center}

The corresponding algorithm (see Algorithm \ref{ALG:GD}) proceeds as follows. Let $M^{(d)}_0$ be an injective initial point and suppose that the current iterate at step $k$ is $M^{(d)}_k$. For each $i$ compute the optimal maps from $M_k^{(d)}$ to each of the prescribed operators $F_i^{(d)}$, namely $$T^{(d)}_{k,i}=\left(M_k^{(d)}\right)^{-\nicefrac{1}{2}}\left[\left(M_k^{(d)}\right)^{\nicefrac{1}{2}} F^{(d)}_i\left(M_k^{(d)}\right)^{\nicefrac{1}{2}}\right]^{\nicefrac{1}{2}}\left(M_k^{(d)}\right)^{-\nicefrac{1}{2}}.$$ 
Define their average $T^{(d)}_k=n^{-1} \sum_{i=1}^n T^{(d)}_{k,i}$, which is a positive matrix, and then set the next iterate to $M^{(d)}_{k+1}(t)=T_k^{(d)} M_k^{(d)} T_k^{(d)}$. The algorithm starts with an initial guess of the \Frechet mean, for instance one of the $F_j^{(d)}$; it then lifts all observations to the tangent space at that initial guess via the log map, and averages linearly on the tangent space; this linear average is then retracted onto the manifold via the exponential map, providing the next guess, and iterates. 

\medskip

In finite dimension \citet{chewi2020gradient} established a linear rate of convergence for gradient descent (GD). Furthermore, they also demonstrated that a stochastic gradient  (SGD) descent algorithm converges at a parametric rate. See also \citet{altschuler2021averaging}.

\begin{center}
\begin{minipage}{0.68\textwidth}
\begin{algorithm}[H]
  \caption{Stochastic Gradient Descend}\label{ALG:SGD}
  \begin{algorithmic}[1]
    \Procedure{SGD}{$M_0,\{F_1^{(d)},\dots,F_n^{(d)}\}, \{\eta_j\}_{j = 1}^K, \{Y_j^{(d)}\}_{j = 1}^K$}
      \For{$t = 1, \ldots, K$}
        \State $
        S_k^{(d)} \gets (1-\eta_k)\Id^{(d)} + \eta_k T_{M_k^{(d)}}^{{Y}^{(d)}_k}
        $
        \State $M^{(d)}_{k+1}\gets S_k^{(d)} M_k^{(d)} S_k^{(d)}$
      \EndFor
      \State \textbf{return } $M^{(d)}_K$
    \EndProcedure
  \end{algorithmic}
\end{algorithm}
\end{minipage}
\end{center}

The SGD additionally requires a sequence $\{\eta_k\}_{k=1,\dots,K}$ of step sizes, as well as an i.i.d.\ sample $\{{Y}_k^{(d)}\}_{k=1\dots,K}$, distributed as  $F_1^{(d)}$. Such a sequence may for instance be obtained by resampling.
At each iteration, SGD moves the iterate along the geodesic
between the current estimate $M_k^{(d)}$ and ${Y}_k^{(d)}$ for a step size $\eta_k$.

As the authors point out in the reference, this stochastic variant exhibits a
slower convergence in the estimation of empirical barycenter, but has a much cheaper iteration cost, and lends itself better to parallelization.

\subsection{PCA}\label{subsec:comp:pca}

We now describe the implementation of the PCA described in Subsection \ref{subsec : PCA}.
\begin{enumerate}
    \item Given flows $\F_1,\dots,\F_n$, which we assume to be observed via $d$-dimensional projections on some grid $\vartriangle = \{t_j\}_{j=1\dots,\lvert \vartriangle \rvert}$, i.e.\ as the family of arrays of matrices
 $\{F^{(d)}_{i,j}\}_{i=1,\dots,n, \: j=1,\dots, \vartriangle}$, we estimate their \Frechet mean flow $\{\widehat M_{n,j}^{(d)}\}_{j=1\dots,\lvert \vartriangle \rvert}$ employing either gradient descend Algorithm \ref{ALG:GD} or \ref{ALG:SGD}.
 
    \item We log-transform to lift our sample to the tangent space. That is, for every $i=1,\dots,n$ we consider the array of log-processes 
    $$
    \left(\log_{\widehat{M}_{n,1}^{(d)}}F^{(d)}_{i,1},
    \quad \dots, \quad \log_{\widehat{M}_{n,\lvert \vartriangle \rvert}^{(d)}}F^{(d)}_{i,\lvert \vartriangle \rvert}\right)
    \in \TT_{\widehat M_{n,1}^{(d)}} \times \dots\times \TT_{\widehat M_{n,\lvert \vartriangle \rvert}^{(d)}},
    $$
    where  for $i=1,\dots,n$ and $j=1,\dots,\lvert \vartriangle \rvert$:
    $$
    \log_{\widehat{M}_{n,j}^{(d)}}F^{(d)}_{i,j} = (\widehat{M}_{n,j}^{(d)})^{\nicefrac{1}{2}}\left[(\widehat{M}_{n,j}^{(d)})^{\nicefrac{1}{2}} 
    F^{(d)}_{i,j}
    (\widehat{M}_{n,j}^{(d)})^{\nicefrac{1}{2}} \right]^{\nicefrac{1}{2}}(\widehat{M}_{n,j}^{(d)})^{\nicefrac{1}{2}} - \Id^{(d)}, 
    $$
    and we denoted by $\Id^{(d)}$ the $d\times d$ identity matrix.
    
    \item Though each $\TT_{\widehat M_{n,j}^{(d)}}$ is proper Hilbert space, for every $j=1,\dots,\lvert \vartriangle \rvert$, these have distinct geometrical structures. Therefore, to ease the implementation of PCA, we appeal to the isometric embedding into the space of flows of Hilbert-Schmidt operators introduced in \eqref{eq : embedding}. In fact, performing the PCA on the log-processes on the product space $ \TT_{\widehat M_{n,1}^{(d)}} \times \dots\times \TT_{\widehat M_{n,\lvert \vartriangle \rvert}^{(d)}}$ with respect to the inherited geometry is fully equivalent to performing PCA on the family of arrays:
    $$  \left({\J}_{\widehat M_{n,1}}\log_{\widehat{M}_{n,1}^{(d)}}F^{(d)}_{i,1},
    \quad \dots, \quad {\J}_{\widehat M_{n,\lvert \vartriangle \rvert}}\log_{\widehat{M}_{n,\lvert \vartriangle \rvert}^{(d)}}F^{(d)}_{i,\lvert \vartriangle \rvert}\right)
    \in \BB_2^{\lvert \vartriangle \rvert},
    $$
    as elements of the space $\BB_2^{\lvert \vartriangle \rvert}$, which is euclidean and allows for a straightforward PCA.
    
    Note that for $i=1,\dots,n$ and $j=1,\dots,\lvert \vartriangle \rvert$, we may explicitly write:
    $$ {\J}_{\widehat M_{n,j}}\log_{\widehat{M}_{n,j}^{(d)}}F^{(d)}_{i,j} 
    = 
    (\widehat{M}_{n,j}^{(d)})^{\nicefrac{1}{2}}\left((\widehat{M}_{n,j}^{(d)})^{\nicefrac{1}{2}}\left[(\widehat{M}_{n,j}^{(d)})^{\nicefrac{1}{2}} 
    F^{(d)}_{i,j}
    (\widehat{M}_{n,j}^{(d)})^{\nicefrac{1}{2}} \right]^{\nicefrac{1}{2}}(\widehat{M}_{n,j}^{(d)})^{\nicefrac{1}{2}} - \Id^{(d)}\right).
    $$
\end{enumerate}

\section{Application to Phase Varying Functional Time Series}
\label{sec : fts}

In this section we show how our methods can be applied to conduct a frequency domain based analysis of functional time series by making use of the spectral representation of stationary sequences.

\medskip

 A functional time series (FTS) is a time-ordered sequence of random elements in a Hilbert space $\H$, say $L^2( [ 0 , 1 ] ,  \mathbb{R})$, usually denoted as $X \equiv\left\{X_t\right\}_{t \in \mathbb{Z}}$. 
If the sample paths are continuous, one can interpret a functional time series as a sequence of random curves $\{X_t(x) \::\: x \in[0,1]\}_{t \in \mathbb{Z}}$ at each point; the discrete index parameter $t$ is often interpreted as a time variable, and the argument variable $x$ represents a continuous spatial location in the domain $[0,1]$.
 
\medskip

In the following we will assume that the series $\left\{X_t\right\}_{t \in \mathbb{Z}}$ is \textit{strictly stationary}: that is, for any finite set of indices $I \subset \mathbb{Z}$ and any $s \in \mathbb{Z}$, the joint law of $\left\{X_t, t \in I\right\}$ coincides with that of $\left\{X_{t+s}, t \in I\right\}$.
 Furthermore, we will assume that $\mathbb{E}\left\|X_0\right\|_2<\infty$. This allows us to define the mean of $X_t$ as an element of $L^2([0,1], \mathbb{R})$,  $m(\tau)=\mathbb{E} X_t(\tau)$, which is independent of $t$ by stationarity. Furthermore, we may define the following lag-$h$ autocovariance operator:
\begin{equation}
    \label{eq : autocovariance operators}
\mathscr{R}_h^X=\mathbb{E}\left[\left(X_h-m_X\right) \otimes\left(X_0-m_X\right)\right]=\mathbb{E}\left[\left\langle\cdot, X_0-m_X\right\rangle\left(X_h-m_X\right)\right] .
\end{equation}
which encapsulates the second-order stochastic dynamics of the process.
 \citet{panaretos2013cramer, panaretos2013fourier} introduced a frequency-domain approach to functional time series analysis. Their method develops a spectral representation for stationary sequences that simultaneously capture both the within curve dynamics (the dynamics of the curve $\{X_0(\tau) \;:\; \tau \in [0, 1]\}$) as well as the between curve dynamics (the dynamics of the sequence $\{X_t \;:\; t \in \mathbb{Z}\}$). The building block of their analysis is the definition of a suitable notion of Fourier transform of the autocovariance operators \eqref{eq : autocovariance operators}, referred to as \textit{spectral density operator} of the time-series $\left\{X_t\right\}$:
\begin{equation}
\label{eq : spectral density operators}
\F_\omega=\frac{1}{2 \pi} \sum_{t \in \mathbb{Z}} e^{-\mathrm{i} \omega t} \RR_t,
\end{equation}
which is well-defined in trace norm, provided the autocovariance operators satisfy the weak dependence condition $
\sum_{h \in \mathbb{Z}}\left\|\mathscr{R}_h^X\right\|_1<\infty.$ 
 Furthermore, the following inversion formula holds:
$$
\RR_t = \int_0^{2\pi} e^{\mathrm{i} \omega t}\F_\omega
$$
 establishing that the autocovariance and the spectral density operators constitute a Fourier pair.

\medskip
Set $\H = L^2([0,1])$, and denote by $\H^{\mathbb{C}}$ its complexification.
Observe that, for all $\omega\in[-\pi,\pi]$, the spectral density operator $\F_\omega$ is an element on the space $\BB(\H^{\mathbb{C}})$ of bounded operators on the complexification of $\H$. In fact, it may be shown that the spectral density operator $\F_\omega$ is self-adjoint and nonnegative. Furthermore, whenever $\sum_{h\in\mathbb{Z}}|\tr(\RR_{h})|<\infty$, one may show that the spectral density operator is uniformly bounded and trace-class: 
$$
\sup_{\omega\in[0,2\pi]}\left\|\F_\omega\right\|_1 \leq \frac{1}{2 \pi} \sum_{h\in\mathbb{Z}}|\tr(\RR_{h})|
$$
and that the mapping 
$$\FF\;:\; [-\pi,\pi]\rightarrow \KK(\H^{\mathbb{C}}), \quad \FF(\omega) =\FF_{\omega}$$
defines a uniformly continuous flow of covariance operators on $\H^{\mathbb{C}}$. Indeed:
\begin{align*}
\left\|\FF_{\omega_1}-\FF_{\omega_2}\right\|_{\BB_1(\H^{\mathbb{C}})} & \leq \sum_{t \in \mathbb{Z}}\left|e^{-\mathrm{i} t \omega_1}-e^{-\mathrm{i} t \omega_2}\right|\left\|\RR_t\right\|_{\BB_1(\H^{\mathbb{C}})} \\
& \leq C \sum_{|t| \leq N}\left|e^{-\mathrm{i} t \omega_1}-e^{-\mathrm{i} t \omega_2}\right|+2 \sum_{|t|>N}\left\|\RR_t\right\|_{\BB_1(\H^{\mathbb{C}})},
\end{align*}
where $C=\max _{t \in \mathbb{Z}}\left\|\RR_t\right\|_{\BB_1(\H)}$. Fixing $\varepsilon>0$, since $\sum_{t \in \mathbb{Z}}\left\|\RR_t\right\|_1<\infty$, we can choose $N=N(\varepsilon)>0$ such that the right-hand sum is smaller than $\varepsilon / 2$, and since for each $t$, the function $\omega \mapsto e^{-\mathrm{i} t \omega}$ is uniformly continuous, we may also choose $\delta=\delta(N, \varepsilon)>0$ such that the left-hand side is smaller than $\varepsilon / 2$, by which uniform continuity follows.

\medskip

It is readily seen that the methodology and theoretical results we presented herein, where we consider a real-valued Hilbert space, may be directly extended to complex Hilbert spaces, starting from the definition of the Bures-Wasserstein distance \eqref{eq : bures-wasserstein distance} all the way to the principal component analysis for covariance flows and the Karhunen-Loève decomposition appearing in Subsection \ref{subsec : PCA}. In particular, our methods concerning flows of covariance operators may be adapted and applied to flows of complex-valued covariance operators, and specifically to the analysis of spectral density operators of functional time series.

\medskip

This in turn allows for the introduction of hierarchical models for collections of heterogeneous time series. 
 Given an $m$-vector whose components are stationary functional time series, say $(X^{(1)},\dots,X^{(m)})$, one can postulate a model where each $X^{(j)}$ is generated conditionally on the realisation of (a random) spectral density $\F^{(j)}$, via a Cram\'er-Karhunen-Lo\`eve expansion \citep{panaretos2013cramer}. Conversely, one may conduct a PCA on the spectral densities, i.e. the flows of spectral density operators $\F^{(1)}, \dots,\F^{(n)}$, interpreted as continuous flows of covariance operators, and proceeding according to the methodology we described herein.
We apply this to real data in Section \ref{subsect : FTS mortality}.

\section{Numerical Simulations}
\label{sec : simulations}

\subsubsection*{A generative model}

\label{subsec freqdom genmod}
Conducting a simulation study testing our methodology requires the generation of independent, identically distributed flows of covariances $\F_1,\dots,\F_n$. We do this by independently applying smooth perturbations to a fixed template flow. 
Our construction of the flows goes as follows.
Consider a fixed template $\M:[0,1]\rightarrow\KK$ and a family of random i.i.d.\ non-negative, continuous ``perturbation curves'' 
$\{T_i: [0,1]\rightarrow \BB(\H)\}_{i=1,\dots,n},$ in that $\lim_{s\rightarrow t}\|T_i(t)-T_i(s)\| =0$ for all $i$ and $t$ We then define a sample of i.i.d.\ flows $\F_1,\dots,\F_n$ by setting:
\begin{equation}\label{eq : template}
\F_i(t) = T_i(t)\M(t)T_i^*(t), \qquad \text{for every } t\in[0,1] \text{ and } i=1,\dots,n,
\end{equation}
with $\E T_1(t) = \Id$ for $t\in[0,1]$, so that $\M$ is \Frechet mean flow of the $\F_j$.

To construct the mean-identity continuous non-negative family of perturbing curves we proceeded as follows.
Consider the harmonics
$
\{\psi_k \::\: x\mapsto \psi_k(x)\}_{k\in\mathbb{N}},
$

$$
\psi_k(x) =\begin{cases}
1 \hspace{2cm} \text{if } k=0\\
\sin(2\pi k x), \quad \text{if } k>0, k \text{ odd}\\
\cos(2\pi k x), \quad \text{if } k>0, k \text{ even},
\end{cases}
$$
which constitute a complete orthonormal system for $L^2([0,1],\R)$.  Each perturbation flow $T_i$, for $i=1,\dots,n$, evaluated at a point in time $t\in[0,1]$, is constructed by setting:
\begin{equation}\label{eq : T freq gen}
    T_i(t) = \sum_{k=0}^\infty \lambda^{(i)}_k(t) \psi_k(\cdot - \theta^{(i)}(t)) \otimes \psi_k(\cdot -  \theta^{(i)}(t)),
    \qquad \text{with}\quad \lambda^{(i)}_k(t) = \frac{1}{\nu} c^{(i)} \cdot W^{(i)}_k(t),
\end{equation}
with $c^{(i)}\sim \chi^2(\nu)$ a chi-squared random variables with $\nu$ degrees of freedom, $t\mapsto W^{(i)}(t)$ a continuous random curve that is everywhere positive and of mean $1$, and
$
t\mapsto \theta^{(i)}(t)
$
 a stochastic process taking values in $[0,2\pi]$ with continuous trajectories. 
 All objects are mutually independent.
Note that
$\E T_i(t) =\Id$, regardless the distribution of each $W^{(i)}, \theta^{(i)}$, with the parameter $\nu$ controlling the concentration of $ T_i(t)$ around the identity.

\medskip 

In the simulations we considered $\H = \R^d$ for $d=100$. The template flow is taken to be the geodesic flow (McCann interpolation) between the standard Brownian motion and standard Brownian Bridge covariances. The unit interval was discretized and subdivided into $100$ sub-intervals. The summation in \eqref{eq : T freq gen} was truncated at the first $50$ summands. 
The empirical \Frechet mean was estimated by the gradient descent Algorithm \ref{ALG:GD}.

\begin{figure}[h!]
    \centering
    \includegraphics[width=1\textwidth]{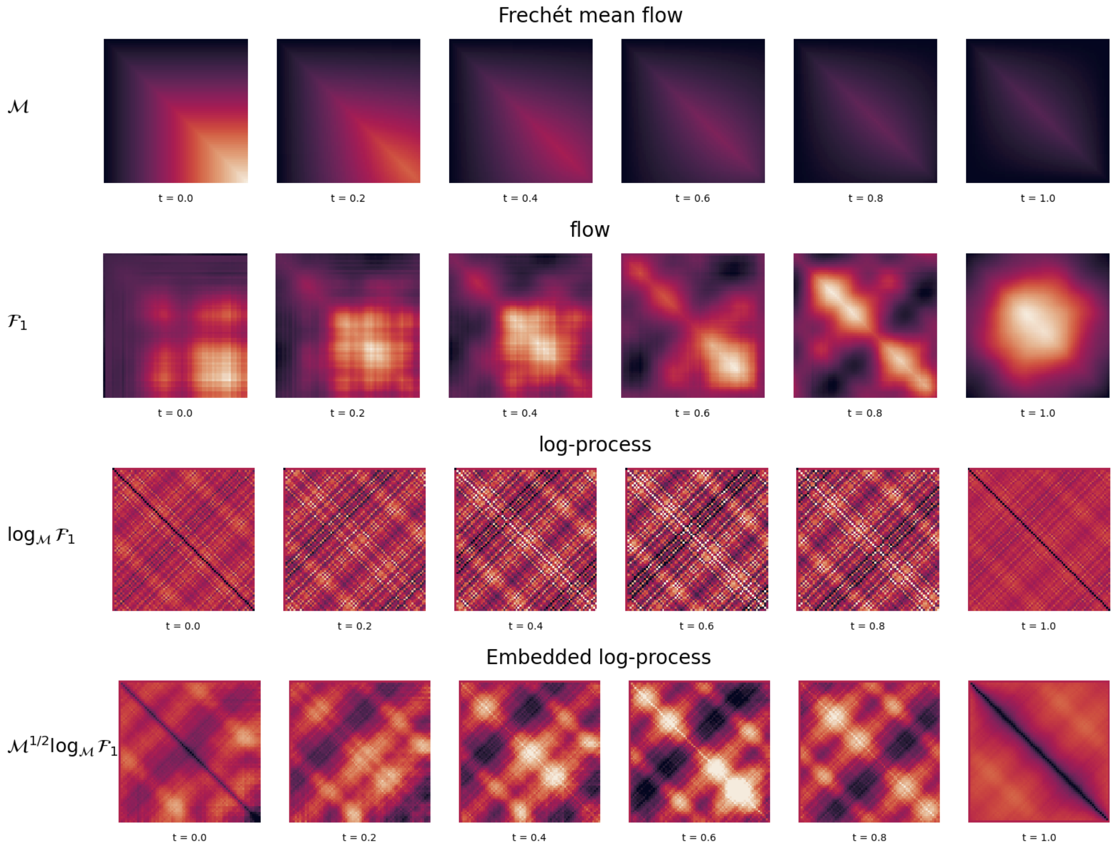}
    \caption{ From top to bottom: the latent template -- obtained by McCann's interpolation of the Brownian Motion covariance and the Brownian Bridge covariance, a random perturbation of it, its logarithmic projection on the tensor Hilbert space of the barycenter, and its embedding into $L^2_{\BB_2}$. }
    \label{fig:flow_log_emblog}
\end{figure}

\subsubsection*{PCA \& Clustering}
In this subsection we show that the functional Principal Component Analysis described in Subsection \ref{subsec : PCA} can be effective in extracting patterns in covariance-valued functional data.
With this aim in mind, we generated a synthetic dataset of covariance flows sampling from a bimodal distribution. Specifically, we generated a dataset by functional geodesic perturbations of a fixed template as in \eqref{eq : template}, with perturbation flows constructed as follows:
$$
T_i(t) = W^{(i)}(t) \sum_{k=0}^\infty 1/k c^{(i)}_k  \psi_k(\cdot - \theta^{(i)}) \otimes \psi_k(\cdot -  \theta^{(i)}),
$$
where, for $i=1,\dots,n$:
the scalars $(c^{(i)}_k)_{k\geq 0}$ are i.i.d.\ $\chi^2(\nu)$ for $\nu =10$;  $\theta^{(i)} \sim \mathrm{Uniform}(0,2\pi)$ are i.i.d.; and $t\mapsto W^{(i)}(t)$ are random i.i.d.\ smooth curves, with mean  $1$, but with a bimodal law.
Specifically, to generate the $W^{(i)}(t)$ we considered two distinct fixed random smooth curves, $t\mapsto g_1(t)$ and $t\mapsto g_1(t)$, with trajectories in $[0,1]$. Then, we generated $n$ two-dimensional points $\{ (a^{(i)}_1,a^{(i)}_2) \}_{i=1,\dots,n} \subset \R^2$ by:
$$
(a^{(i)}_1,a^{(i)}_2) \sim \begin{cases}
\mathrm{Uniform}(0,1) \times \mathrm{Uniform}(-1,0), \quad \text{with probability } 1/2,\\
\mathrm{Uniform}(-1,0) \times \mathrm{Uniform}(0,1), \quad \text{with probability } 1/2.
\end{cases}
$$
Each curve $t\mapsto W^{(i)}(t)$ was obtained as linear combination of $t\mapsto g_1(t)$ and $t\mapsto g_2(t)$, with coefficients given by $(a^{(i)}_1,a^{(i)}_2)$:
$$
 W^{(i)}(t) = 1 + a^{(i)}_1 \cdot g_1(t) +  a^{(i)}_2 \cdot g_2(t), \qquad \text{for } i=1,\dots,n.
$$

\begin{figure}[h!]
    \centering
    \includegraphics[width=.9\textwidth]{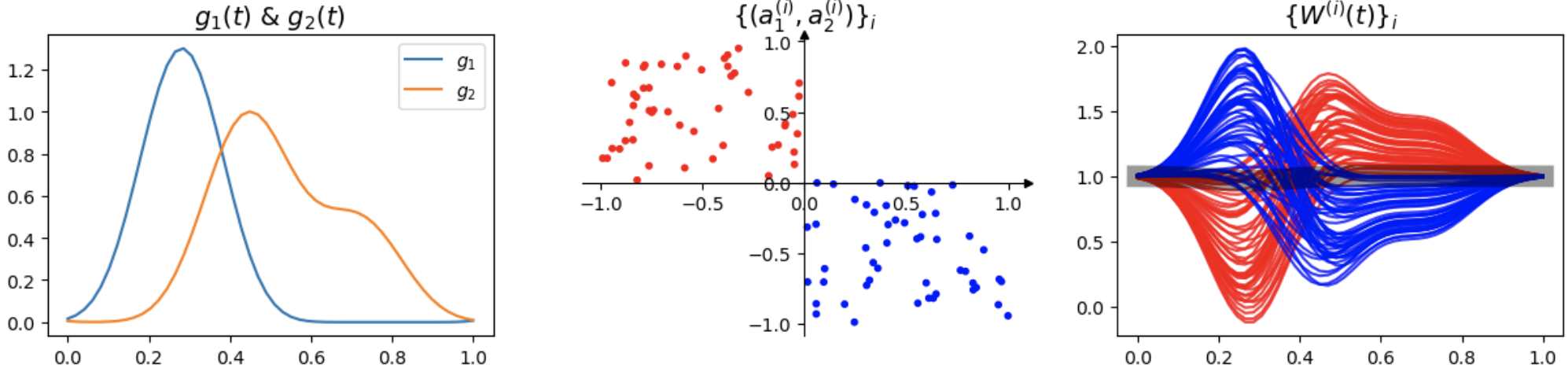}
    \caption{Generation of samples of the  random bimodal curves $W^{(i)}(t)$. The two clusters, associated with the two modes of variation of the random curve, are identified by the color red and blue. }
    \label{fig:_gen_bimodal_clusters}
\end{figure}

Note that the construction ensures that $\E[ W^{(i)}(t)] = 1$ for every $i=1,\dots,n$ and $t\in[0,1]$, but with two modes of fluctuation. Since $W^{(i)}$ essentially determines the rate of the spectral decay of the perturbation map $T_i$, we expected such a construction to generate a dataset of flows with two distinct clusters, and the Principal Component Analysis described in Subsection \ref{subsec:comp:pca} to detect these two modes of variation and distinguish the clusters. 
And indeed generating 100 flows as described above, we found that the information relative to the two distinct clusters is encoded in the variation corresponding to the first principal component, which was in fact sufficient to achieve perfect classification of the data, though explaining only 34\% of the variability in the data.

\subsubsection*{Visualising the modes of variation}
In (real-valued) functional data analysis, it is typical to visualise the effect of the first principal component, say $\phi_1$, by means of comparing the plots of the mean function $m$ with those of $m+\lambda \phi_1$ and $m-\lambda \phi_1$, for some $\lambda>0$. In our setting, such a visualisation is not straightforward, and not (only) for the complicated nature of our kind of data, bur rather because the Principal Component Analysis we developed takes place at the level of the embedded tangent space, and not directly at the level of the ambient space of covariance flows. 

Indeed, with our analysis we obtain principal components $\{\Phi_k\}_{k\geq 1}\in L^2_{\BB_2}$ in consequence to the linear analysis of the embedded log-transformations of the observed flows,
$\bJ_{\M_n}\log_{\M_n}\F_i\}_{i=1,\dots,n}$, where $\M_n$ denotes the \Frechet mean flow of the $\{\F_i\}_{i=1,\dots,n}$. Note however that, by \citet{ocana1999functional}, this is equivalent to performing PCA on the log-processes in the tangent space norm, in the sense that the eigenvalues (i.e., the variances) remain the same, and the eigenfunctions $\Phi_k$ of the embedded log-processes with respect to $\langle\cdot,\cdot\rangle_{\BB_2}$ are simply $\J_{\M_n}$ applied to the eigenfunctions $\Psi_k$ of the covariance of the log-processes in $\langle\cdot,\cdot\rangle_{\TT_{\M}}$. That is, $\Psi_k := {\bJ}_{\M_n}^{\dag}\Phi_k$, for $k\geq 1$.

Given the manifold structure of the space of flows, \textit{locally} the projection on the tangent space provides a faithful representation of $\KK$ near $\M_n$: though the exponential map does not constitute a homeomorphism, for a given tangent vector $V\in\TT_{\M_n}$ and for small values of $\lambda \in\R$, the map $\lambda \mapsto 
 \exp_{\M_n(t)}\{\lambda V(t)\}$ well approximates constant-speed geodesic paths.
Furthermore, for each $k\geq 1$, $t\in[0,1]$ and small enough values of $\lambda$, we may  interpret $\lambda \Psi_k(t) + \Id$ and $-\lambda\Psi_k + \Id$  as optimal transport maps. 
In particular, we may consider:
$$
\lambda \mapsto 
 \exp_{\M_n(t)}\{ \lambda \Psi_k(t)\} 
 = \left( \lambda\Psi_k(t) + \Id\right)\M(t)\left( \lambda\Psi_k(t) + \Id\right),
 \qquad  \lambda \in [-\lambda_{\star}, \lambda_{\star}]
 $$
 for some upper bound $\lambda_{\star}>0$,
which heuristically represent in $\FF_C$ the \textit{positive and negative linear deformations} near $\M_n$ which are caused by the effect of the principal component $\Psi_k$ at the level of the tangent space $\TT_{\M_n}$. See Figure \ref{fig : mean +- pc}.

\begin{figure}[h!]
    \centering
    \includegraphics[width=1\textwidth]{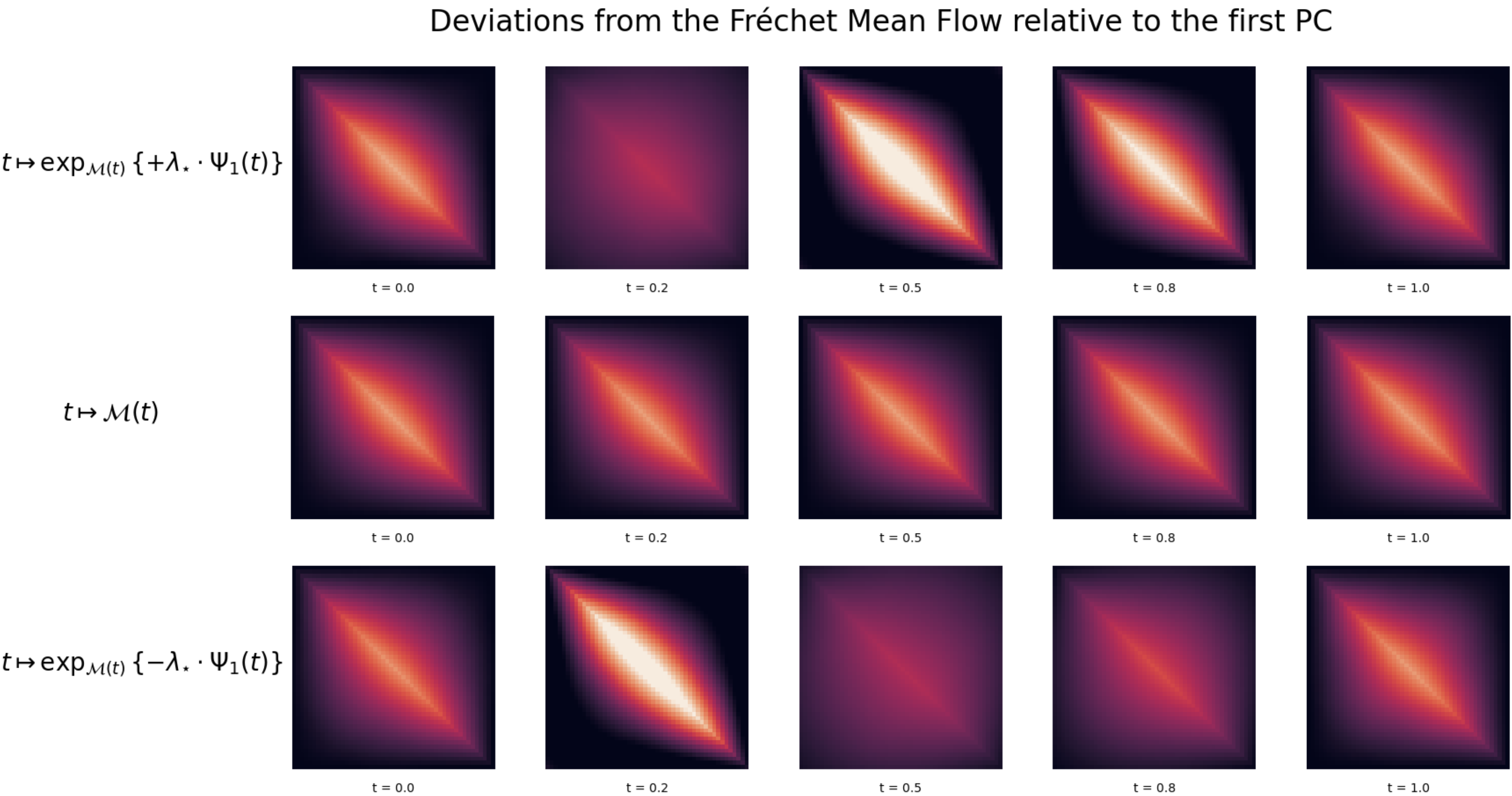}
    \caption{Positive and negative deviation from the \Frechet mean flow relative to the first principal component through the exponential map. It is interesting to notice the consistency with the modulus of the difference between the two clusters in the right plot in Figure \ref{fig:_gen_bimodal_clusters}.}
    \label{fig : mean +- pc}
\end{figure}

\section{Data Analysis:}

\subsection{EEG Movement/Imagery Dataset}
\label{subsect : EEG}
This data set consists of over 1500 one and two minute 64-channel EEG recordings, obtained from 109 volunteers. In this analysis we compare the recordings of resting state EEG signals with the measurements obtained while subjects performed motor/imagery tasks. In particular, we consider:
\begin{itemize}
    \item recordings of a baseline activity, measuring brain signals of the subject at rest with with eyes closed; such data was recorded twice for each volunteer, for a duration of one minute.
    \item recordings of the brain activity during the following task: a target appears on the side of the screen, at which point the subject opens and closes the corresponding fist until the target disappears, after which the subject relaxes; such data was recorded thrice for each volunteer, for a duration of two minutes.
\end{itemize}
During each task, the
brain of each subject is scanned and  brain activities are recorded at 160 samples per second, from 64 electrodes as per the international 10-10 system as shown in Figure \ref{fig : eeg loc}.
\begin{figure}[h!]
    \centering
    \includegraphics[width=.5\textwidth]{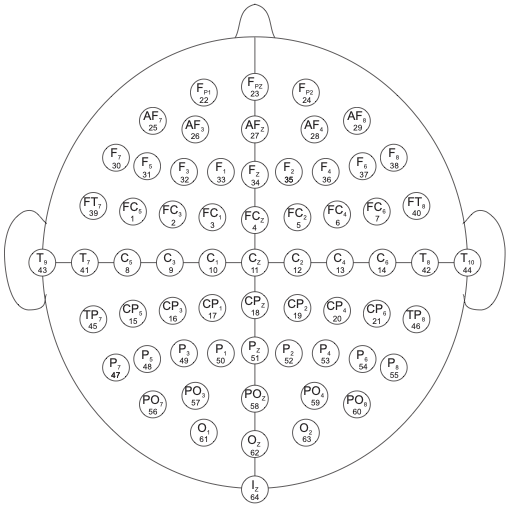}
    \caption{Placement of the electrodes, as per the international 10-10 system; the numbers below each electrode name indicate the order in which they appear in the records.}
    \label{fig : eeg loc}
\end{figure}

At each design time
point $t$, the functional connectivity between brain areas of each subject $i$ is represented by the
covariance matrix $\F_i(t)$ of EEG signals from regions of interest (ROI) at time $t$. 

Let $X_{i,t}$ be the 64-dimensional column vector
that represents the EEG signals at time $t$ from the i-th subject during some experiment. To estimate the connectivity matrix $\F_i(t)$ we adopted a local sliding window approach as in \citet{lin2019intrinsic}:
$$
\F_i(t)=\frac{1}{2 h+1} \sum_{j=t-h}^{t+h}\left(X_{i j}-\bar{X}_{i t}\right)\left(X_{i j}-\bar{X}_{i t}\right)^T \quad \text{with}\quad \bar{X}_{i t}=\frac{1}{2 h+1} \sum_{j=t-h}^{t+h} X_{i j}
$$
where $h$ is a positive integer that represents the length of the sliding window; this is required to be sufficiently large for regularisation, but not too large to avoid significant bias. Our analysis was in fact found to be rather robust with respect to values of $h$ ranging from 100 to 1000. Below we consider $h=500$.
When an experiment was repeated multiple times by a subject, the estimate of the time-varying connectivity matrix was obtained by averaging over the estimates obtained in separate experiments.

\bigskip

We conducted a permutation test to test the equality in distribution of the flows of connectivity matrices at rest and during the movement/imagery task. We considered the distance between the two \Frechet mean flows corresponding to the two subgrups as test statistic, which lead us to reject the null. 

\bigskip

We performed our Principal Component analysis method for random flows of covariances on the pooled samples of time-varying connectivity matrices, at the level of the tangent space of their pooled \Frechet mean flow.
We found that the first two principal component scores -- while explaining for only ~50\% of the data variability -- where sufficient in inferring
whether a sample came from a resting or active subject in over 4/5 of the cases by means of e functional linear discriminant analysis (fLDA)
classifier, which has been previously used on EEG signals in \cite{masak2023separable}. Our results were validated by out of sample leave-one-out cross-validation.
In particular, we observed a cross-validated accuracy ranging from 83\% to 89\%, depending on the number of considered components. But most interestingly, we observed that the second principal component in itself carried most of the relevant information relative to the classification, achieving a classification accuracy of 81\%, while explaining only $\sim$10\% of the data variabuility. We found that the effect of the second principal component mostly concerned changes in the correlations between areas of the frontal lobe (F and AF), which were increased in the subjects performing the activity. This is a reasonable result, as the frontal lobe is known to be relevant in all voluntary movement, and is further evidence that our method is effective in extracting valuable information from time-varying coviariance data.

\subsection{FTS analysis:  Human Mortality Database}
\label{subsect : FTS mortality}

We analyse the functional time series given by the age-at-death rates for $N = 32$ countries between the years 1960 to 2011, 
obtained from the Human Mortality Database of UC Berkeley and the Max Planck Institute for
Demographic Research (openly accessible on \href{www.mortality.org}{\texttt{www.mortality.org}}). 
Death rates are provided by
single years of age up to 109, with an open age interval for 110+.

\begin{figure}[h!]
    \centering
    \includegraphics[width=1\textwidth]{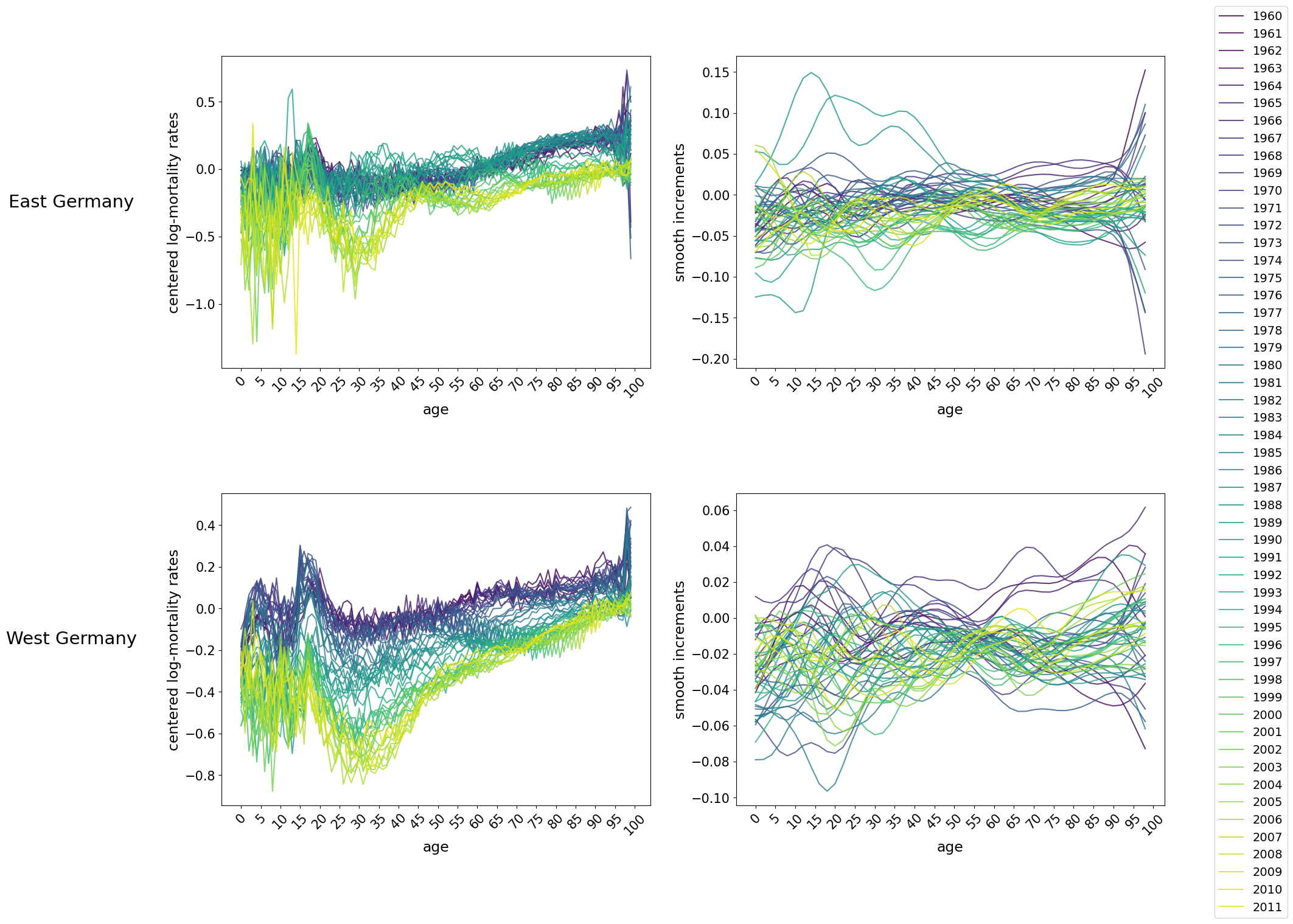}
     \caption{{``Rainbow plot" of the centered logarithmic age-at-death rates (left) and corresponding smoothed increments (right) between 1960 and 2011 in the Eastern and Western bundeslands in Germany. 
     Recall that the Berlin Wall fell in 1989, roughly at the temporal mid point of the time series.}}
    \label{fig : east_west_DE}
\end{figure}

Age-specific mortality rates have been
studied and analysed as functional time series by various authors: see for instance \citet{tang2022clustering} and the references therein. The value of the time series on any given  year can be regarded as a function, with age as its continuous argument. Each distinct country thus yields its own functional time series, indexed by year. 

\medskip

In our analysis, we consider the centered log-transforms of the mortality rates. We aggregate the population aged over 100. Since these are evidently non-stationary, we considered the corresponding increments, obtained by differencing the log-mortality rates. These were then smoothed by kernel smoothing, with bandwith $h=5$ and Gaussian kernel, and averaged over consecutive years. We denote by $\{ (X^{(i)}_t)_{t=1960,\dots,2010}  \}_{i=1,\dots,32}$ the corresponding time series, where the index $i$ represents the corresponding country and $t$ the year.
For each functional time series $X^{(i)}$,  we estimated its corresponding lag-$h$ autocovariance kernel:
$$
\widehat \RR^{(i)}_h = \frac{1}{50-h}\sum_{t=0,\dots,50-h} X^{(i)}_{t+h}X^{(i)}_t, \quad  h =0,\dots,50, \: i=1,\dots,32
$$
and the associated spectral density operators:
$$
\widehat \FF^{(i)}_{\omega} = \sum_{h=0,\dots,50} e^{-it\omega} \widehat\RR^{(i)}_h.
$$

We then clustered the countries running a naive $k$-means algorithm (with 1000 random initial points) on the flows of spectral density operators $\FF^{(1)}_{\cdot}, \dots, \FF^{(32)}_{\cdot}$, with the integrated Bures-Wasserstein distance \eqref{eq : distance between flows} as metric.

Given our data’s nature, to avoid the “curse of dimensionality”, we first apply the Principal Component dimensional reduction method described in Subsection \ref{subsec : PCA} to the spectral density operator flows,
before applying the k-means clustering method. Of course the number of clusters is unknown in
advance, and needs to be determined before clustering. Our criterion for the choice of the number of clusters $k$ was based on the analysis of the dispersion and of inertia of the produced clusters.  Distortion is the average sum of squared distance between each data point to the centroid, while inertia is just the sum of squared distance between the data point to the center of the cluster/centroid. The ``elbow method" suggested $k=5$ for both considered metrics: Figure \ref{fig : elbow}.

\begin{figure}[h!] 
\centering
    \includegraphics[width=.95\textwidth]{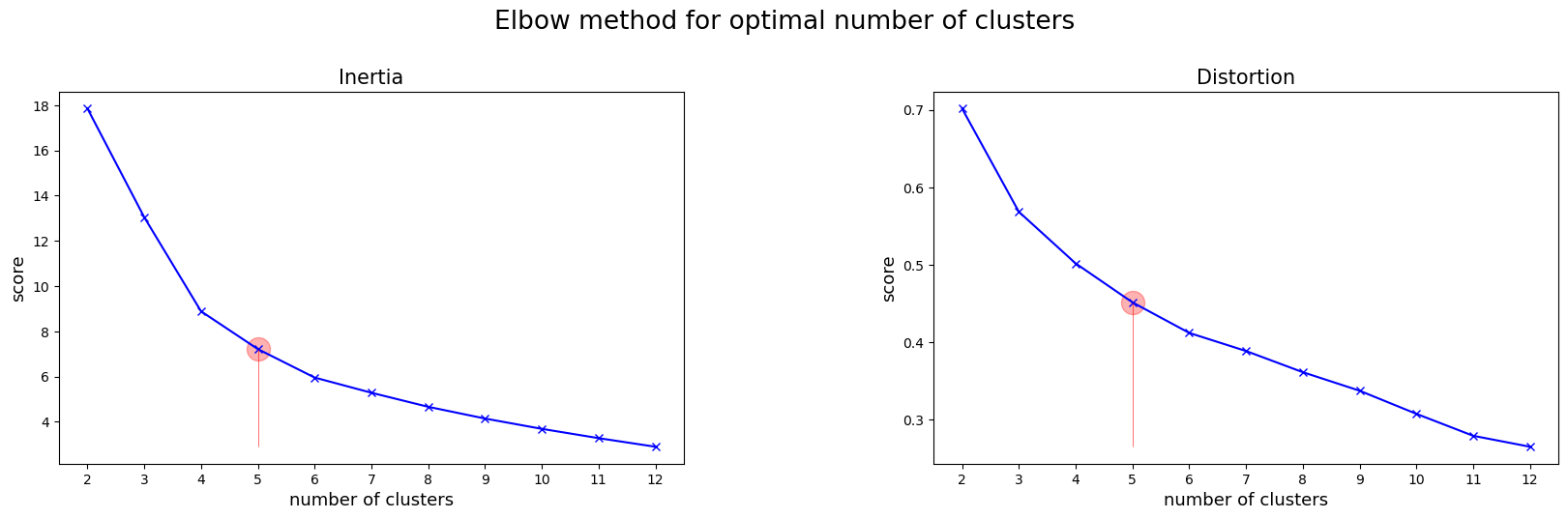}
     \caption{Selection of the optimal number of clusters $k$ to run $k$-means, based on the inertia and distortion scores, on which we evaluated the algorithms output for $k\in\{2,\dots,15\}$. In red the ``elbow" detected by the kneedle algorithm. } 
    \label{fig : elbow}
\end{figure}

The clusters produced by the $k$-means algorithm with $k=5$ applied to the data are illustrated by the pairwise distance matrix in Figure \ref{fig : clust_matrix}. To illustrate the usefulness of a clustering
analysis on multi-country mortality data, we finally report the rainbow plots illustrating the respective common time trends of each detected cluster: see Figure \ref{fig : trends}.
Clusters are shown on a World map in Figures \ref{fig : clust_map_ks} and \ref{fig : clust_map}.

\begin{figure}[h!]
    \centering
    \includegraphics[width=1\textwidth]{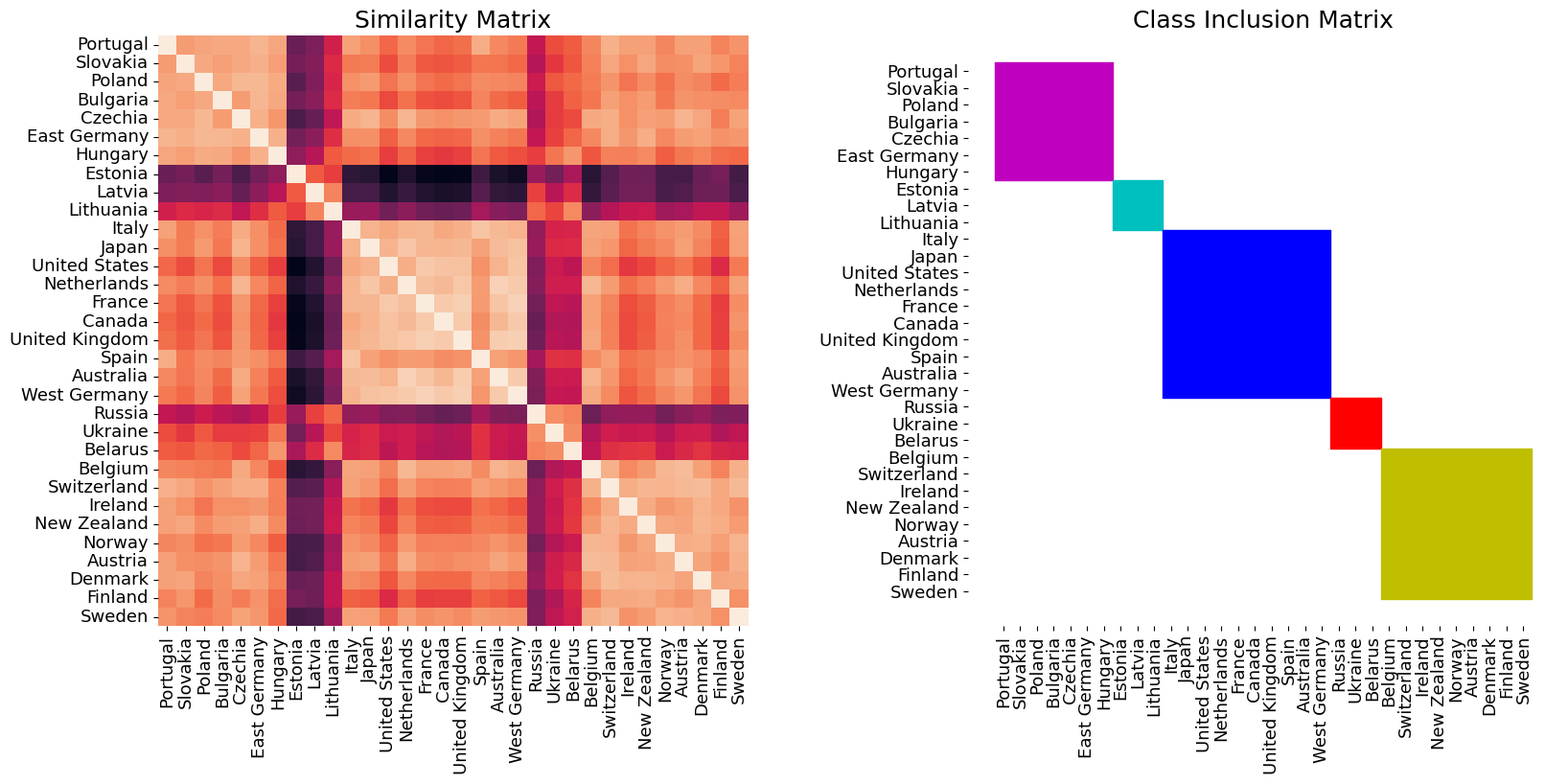}
         \caption{Left: similarity matrix reporting the pairwise distances between power spectral density operators corresponding to the Functional Time Series. Right: cluster inclusion matrix, reporting the clusters obtained via $k$-means with $k=5$. }
    \label{fig : clust_matrix}
    
    \vspace{.5cm}

    \includegraphics[width=1\textwidth]{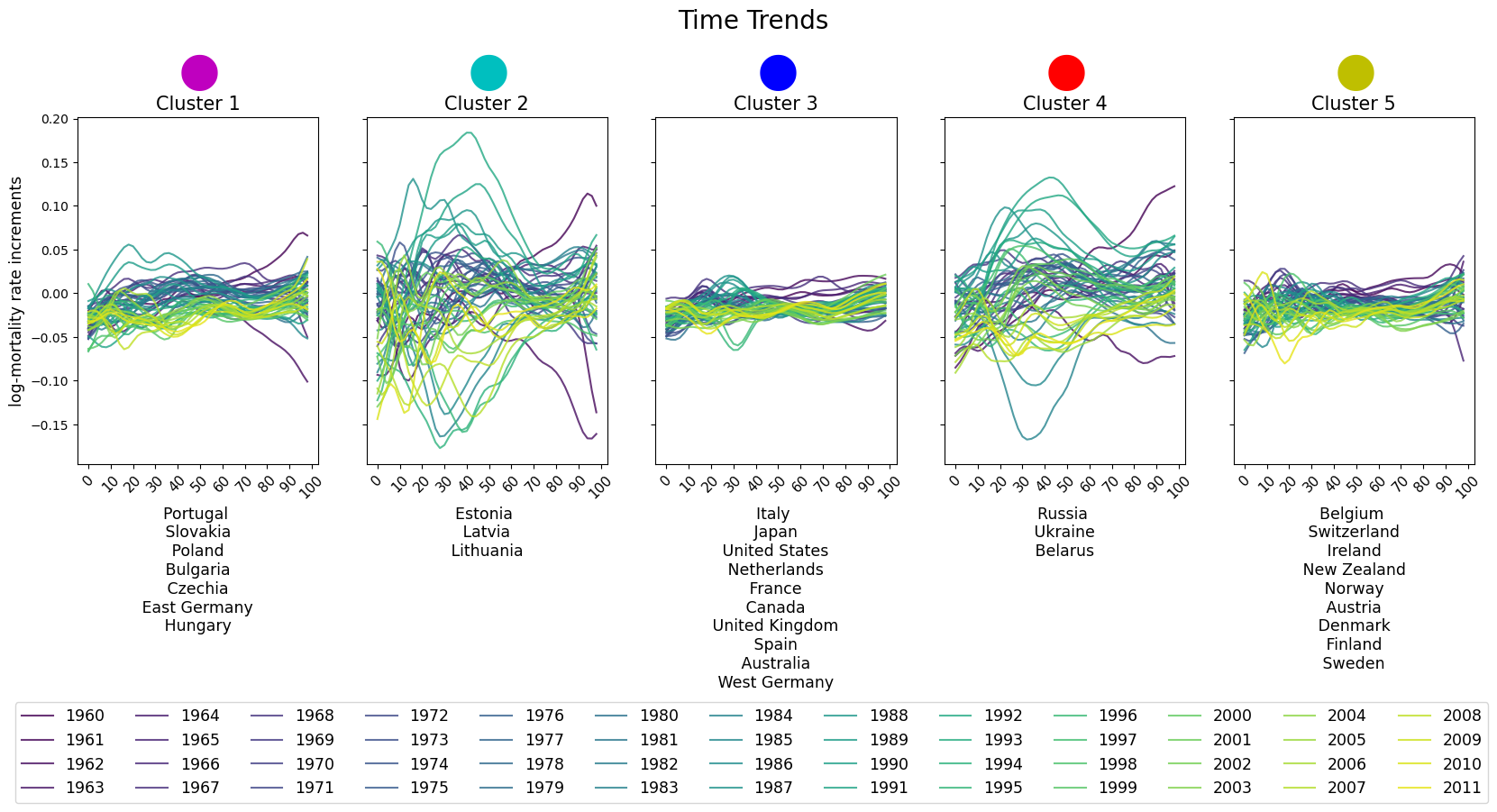}
     \caption{Rainbow plots illustrating the respective common time trends in each cluster detected by $k$-means with $k=5$. }
    \label{fig : trends}
\end{figure}

\begin{figure}[h!]     
    \centering
    \includegraphics[width=1\textwidth]{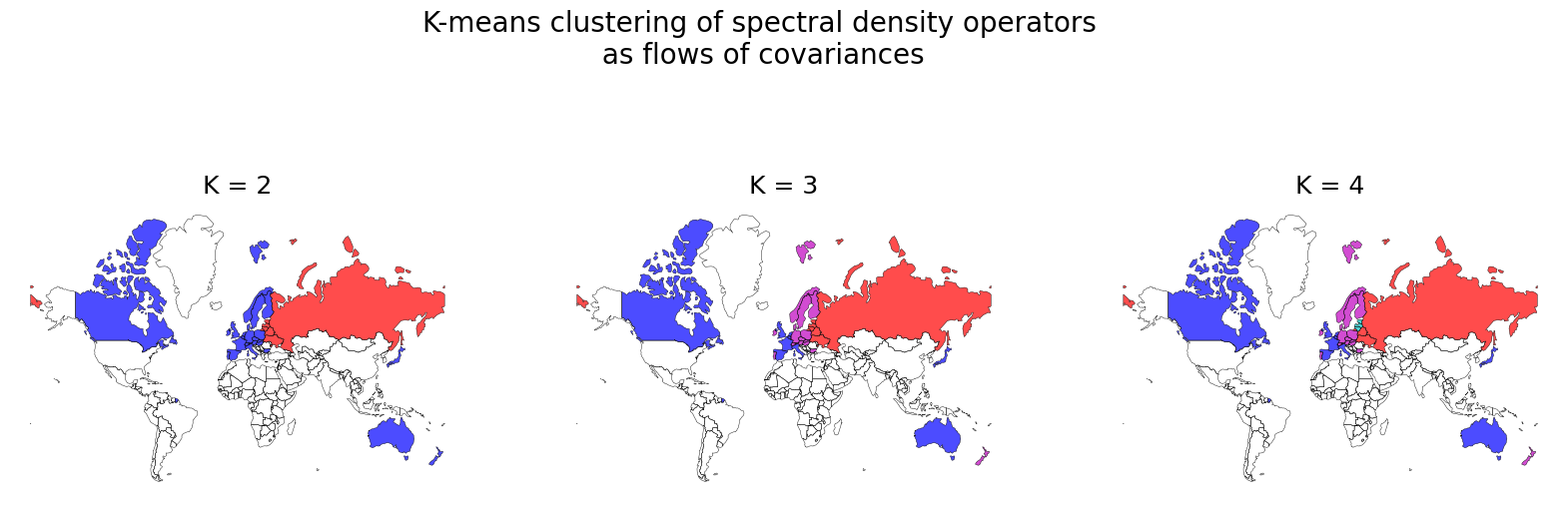}
     \caption{The clustering produced by K-means algorithm with $k=2,3,4$.  }
     \label{fig : clust_map_ks}

\vspace{1.5cm}

    \includegraphics[width=.95\textwidth]{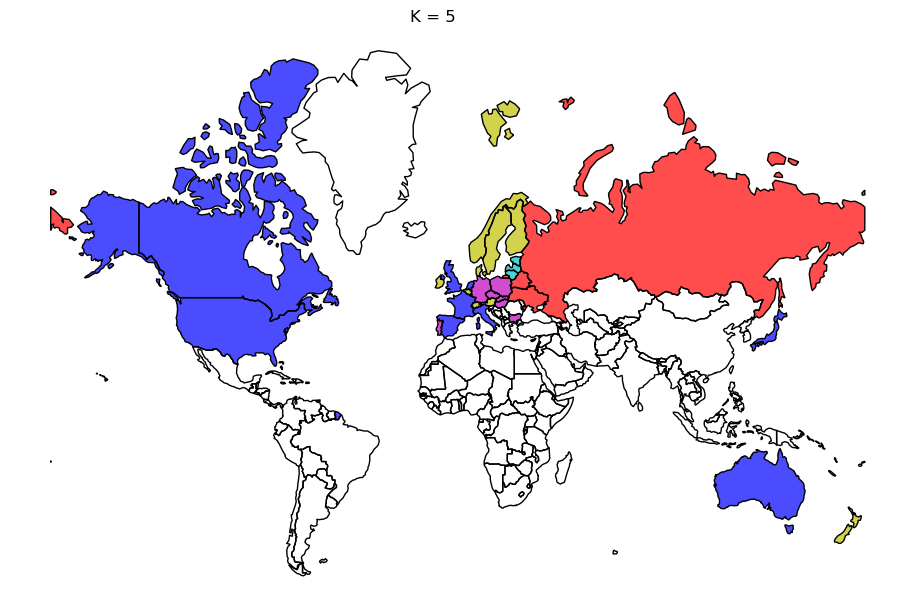}
     \caption{The clustering produced by K-means algorithm with $k=5$. .  }
    \label{fig : clust_map}
\end{figure}

As may be observed from our plots, the clustering produced rather intuitive results, grouping together countries that have shared similar paths in history, and have traditionally been considered to have similar lifestyles and diets. When requiring the algorithm to partition the countries in two groups,  these seem to reflect the partition of the cold war (western and eastern blocks). Increasing the number of clusters to three, certain central and northern Europe countries form their own cluster, while most leading European economies are excluded from this cluster. Interestingly, this is reflected in the case of East and West Germany: while the first remains grouped with the US and more affluent European economies, the latter clusters with most Eastern-European countries, and Portugal.  When $k=4$, the Baltic republics form an additional separate cluster.
Finally increasing $k$ to $5$, which is the optimal value according to inertia and dispersion scores, we see the formation of an additional cluster which includes the northern European countries, together with Ireland, Belgium, Switzerland, Austria and New Zealand.
Our findings present multiple similarities to those obtained by \citet{tang2022clustering}, who develop and apply a functional panel data
model with fixed effects to cluster the time series. Overall, the cluster results reflect a number of factors, such as, geography, lifestyle, ethnicities, socio-economic status and dietary traditions.

\section{Appendix}\label{sec:appendix}
\setcounter{equation}{0}
\renewcommand{\theequation}{A.\arabic{equation}}

\begin{proof}[Proof of Proposition \ref{prop : space pf continous flows is a metric space}]
Note that the distance between continuous flows is continuous. That is, consider $\F,\G\in\FF_C$, $\F= \{F_t\}_{t\in[0,1]}, \G = \{G_t\}_{t\in[0,1]}$; the map $t\mapsto f(t):= \Pi(F_t,G_t)$ is continuous:
\begin{align*}
    |f(t+h)-f(t)|
    \leq & \: | \Pi(F_{t+h}, G_{t+h}) - \Pi(F_{t+h}, G_{t})| \: +\: | \Pi(F_{t+h}, G_t) - \Pi(F_t, G_t) |\\
    \leq & \: | \Pi(G_{t+h}, G_t) - \Pi(F_{t+h}, F_t)|
\end{align*}
which goes to $0$ with $h$ by continuity.
If $d(\F,\G) = 0$, then $\Pi(F_t,G_t)=0$ for almost all $t$. However, by continuity of $t\mapsto \Pi(F_t,G_t)$, we conclude that equality holds everywhere and $\F=\G$.
Symmetry clearly holds. Moreover, essentially following the same steps as in the proof of Minkowski's inequality for ${L}^2$ spaces, one can show that the triangle inequality is satisfied. Indeed, for any $\M=\{M_t\}_{t\in[0,1]}\in\FF_C$:
\begin{align*}
            & \int_{0}^1( \Pi(F_t,M_t) + \Pi(M_t,G_t) )^2 dt \:= \\
            & =\:\int_{0}^1 \Big( \Pi(F_t,M_t) + \Pi(M_t,G_t)\Big) \Pi(F_t,M_t) dt 
            \: +  \: \int_{0}^1\Big( \Pi(F_t,M_t) + \Pi(M_t,G_t) \Big)\Pi(M_t,G_t) dt\\
            & \leq\: \left( \int_{0}^1( \Pi(F_t,M_t) + \Pi(M_t,G_t) ))^2 dt \right)^{\nicefrac{1}{2}}
            \cdot \left[ 
            \left( \int_{0}^1 \Pi(F_t,M_t)^2 dt \right)^{\nicefrac{1}{2}}  \: + \: \left( \int_{0}^1 \Pi(M_t,G_t)^2 dt \right)^{\nicefrac{1}{2}}
            \right].
\end{align*}
where we have used H\"older's inequality.
In particular:
\begin{align*}
     \left(\int_{0}^1\Pi(F_t, G_t)^2dt\right)^{\nicefrac{1}{2}} 
    \leq& \left(\int_{0}^1( \Pi(F_t,M_t) + \Pi(M_t,G_t) )^2 dt\right)^{\nicefrac{1}{2}}\\
    \leq&  \left( \int_{0}^1 \Pi(F_t,M_t)^2 dt \right)^{\nicefrac{1}{2}}  \: + \: \left( \int_{0}^1 \Pi(M_t,G_t)^2 dt \right)^{\nicefrac{1}{2}}
\end{align*}
which yields the triangular inequality.
\end{proof} 

\begin{proof}[Proof of Lemma \ref{lemma : continuity for BW barycenter}] 
Existence follows from coercivity of the \Frechet functional, and uniqueness from its strong convexity: see  \cite[][Theorem 1]{santoro2023large}. As for the pointwise characterisation part, the proof essentially consists in using continuity.
Let $t_n\rightarrow t$ be an arbitrary convergent sequence in $[0,1]$. First, note that for every $m\geq 1$ and $t\in[0,1]$, $\widehat\M(t)\preceq \frac{1}{m}\sum_{i=1}^m\F_i(t)$ \cite{masarotto2019procrustes}  which in turn implies that 
$\M(t_n) \preceq \E[\F(t_n)]$ for each $n \geq 1$ by monotone convergence..
Furthermore, note that  $\E[\F(t_n)] \rightarrow \E[\F(t)]$.
This shows that the sequence $\{\M(t_n)\}_{n\geq 1}$ is relatively compact, as it is enclosed in the compact set $\{F \::\: 0\preceq F \preceq G\}$ for some $G$, in light of the convergence of $\{\E[\F(t_n)]\}_{n\geq 1}$. Furthermore, it is clear that any limit point need be a \Frechet mean of $\F(t)$. Upon observing that such mean need be unique (by assumption \ref{assumption_flow_posdef}), the usual sub-subsequence argument yields the conclusion.
\end{proof}

\begin{proof}[Proof of Proposition \ref{prop : tensor Hilbert space is complete and separable}]
We first show completeness by reproducing the results in \citet[][Theorem 1]{lin2019intrinsic}. Suppose $\left\{\V_{n}\right\}_{n\geq 1}$ is a Cauchy sequence in $\TT_{\F}$. Our next steps show the existence of a subsequence $\left\{\V_{n_{k}}\right\}_{k\geq 1}$ satisfying
\begin{equation}
    \label{eq : 14}
    \sum_{k=1}^{\infty}\left\|\V_{n_{k+1}}(t)-\V_{n_{k}}(t)\right\|_{\Tan_{\F(t)}}<\infty.
\end{equation}
Since $\Tan_{\F(t)}$ is complete, the limit $\V(t)=\lim _{k \rightarrow \infty} \V_{n_{k}}(t)$ is well defined and in $\Tan_{\F(t)}$. 
Fix any $\epsilon>0$ and choose $N$ such that $n, m \geq M$ implies $\left\|\V_{n}-\V_{m}\right\|_{\Tan_{\F(t)}} \leq \epsilon$. An application of
Fatou's lemma to the function $\left|\V(t)-\V_{m}(t)\right|$ yileds that if $m \geq N$, then $\left\|\V(t)-\V_{m}(t)\right\|_{\Tan_{\F(t)}}^{2} \leq \lim\inf _{k \rightarrow \infty}\left\|\V_{n_{k}}-\V_{m}\right\|_{\TT_{\F}}^{2} \leq \epsilon^{2} .$ This shows that $\V-\V_{m} \in \TT_{\F}$. 
Since $\V=\left(\V-\V_{m}\right) + \V_{m}$, we see that $\V \in \TT_{\F}$. 
The arbitrariness of $\epsilon$ implies that $\lim _{m \rightarrow \infty}\left\|\V(t)-\V_{m}(t)\right\|_{\Tan_{\F(t)}}=0$ for $t\in[0,1]$.
Because $\left\|\V-\V_{n}\right\|_{\TT_{\F}} \leq$ $\left\|\V-\V_{m}\right\|_{\TT_{\F}}+\left\|\V_{m}-\V_{n}\right\|_{\TT_{\F}} \leq 2 \epsilon$, we conclude that $\V_{n}$ converges to $\V$ in $\TT_{\F}$.
It remains to show \eqref{eq : 14}. To do so, we choose $(n_{k})_k$ so that $\left\|\V_{n_{k}}-\V_{n_{k+1}}\right\|_{\Tan_{\F(t)}} \leq$ $2^{-k}$. Let $\U \in \TT_{\F}$. 
By Cauchy-Schwarz inequality:
$$\int_{0}^1 \| \U_t\|_{\Tan_{\F(t)}} \cdot\| \V_{n_{k}}(t)-\V_{n_{k+1}}(t)\|_{\Tan_{\F(t)}} \der t  \leq\|\U\|_{\TT_{\F}} \| \V_{n_{k}}- \V_{n_{k+1}}\left\|_{\TT_{\F}} \leq 2^{-k}\right\| \U \|_{\TT_{\F}}$$
and summing over $k$: 
$$\sum_{k} \int_{0}^1 \|\U(t)\|_{\Tan_{\F(t)}} \cdot \|\V_{n_{k}}(t)-\V_{n_{k+1}}(t)\|_{\Tan_{\F(t)}} \der t \leq\|\U\|_{\TT_{\F}}<\infty,$$
which implies  \eqref{eq : 14}. 

\medskip

Next, we show separability following the steps of \citet{zhou2021intrinsic}. First, note that for each $t\in[0,1]$, the tangent space $\Tan_{\F(t)}$ may be identified with a closed subset of the separable Hilbert space $\mathcal{L}^2(\mu(t))$, where $\mu(t) \equiv \N(0,F(t))$, and is therefore itself separable. For $t\in[0,1]$, let $(\Phi_i(t))_{i\geq 1}$ denote a complete orthonormal system (CONS) of $\Tan_{\F(t)}$. In particular, for $i\geq 1$, denote by $\Phi_i$ the equivalence class in $\TT_\F$ of the function $ t\mapsto \Phi_i(t)$.
We may view $(\Phi_i(t))_{i\geq 1, t\in[0,1]}$ as an orthonormal frame along the flow $F(t)$.
Let $\tilde{ \V}^{(1)},\tilde{ \V}^{(2)}$ be two arbitrary representatives of the same equivalence class $\V\in \TT_{\F}$, and denote by $(\tilde a_{1,i}(t))_i$ and $(\tilde a_{2,i}(t)_i$ the corresponding decomposition of $\tilde{ \V}_1(t),\tilde{ \V}_2(t)$ with respect to the CONS $\{\Phi_i(t)\}_i$. That is, for every $t\in[0,1]$ and $j\in\{1,2\}$: $\|\V_j(t) - \sum_{i\geq 1} \tilde a_{j,i}(t)\Phi_i(t)\|_{\Tan_{\F(t)}}=0$. Since $\tilde{ \V}^{(1)},\tilde{ \V}^{(2)}$ belong to the same class $\V$,
$$\|  \tilde{ \V}_{1}(t) - \tilde{ \V}_{2}(t)\|_{\Tan_{\F(t)}}^2 \:=\:\| \sum_{i\geq 1}( \tilde a_{1,i}(t) -  \tilde a_{2,i}(t) )\Phi_i(t)\|_{\Tan_{\F(t)}}^2 \:\:  0\qquad\text{for a.e.\ }t\in[0,1]
$$
and by Parseval's identity:
  $
 \sum_i( \tilde a_{1,i}(t) -  \tilde a_{2,i}(t) )^2= 0
 $
for a.e.\ $t\in[0,1]$. We view $\tilde a_{j,i}$ as element in the space $\ell^{2}$ of square-summable sequences.
 Therefore, we see that each element $\V$ in $\TT_{\F}$ is associated to an element in $L^{2}\left(\T, \ell^{2}\right)$, which we denote by $\Upsilon(\V)$ and interpret as the coordinate representation for $\V$ associated to $(\Phi_i)_i$. It is well known that that $L^{2}\left(\T, \ell^{2}\right)$ is a separable Hilbert space, with inner product
 $\int_{0}^1 \sum_i a_i(t)b_i(t) dt $.
In fact, we can define a map $\Upsilon: \mathscr{T}(\F) \longmapsto L^{2}\left(\T, \ell^{2}\right)$, and we can immediately see that $\Upsilon$ is a linear map, and hence injective. It is also surjective, because for any 
$ \boldsymbol{a} \in L^{2}\left(\T, \ell^{2}\right)$,
the vector field $\U$ along $\M$ given by 
$\U(t)=\sum_{i} {a}_{i}(t) \Phi_{i}(t)$ for $t\in[0,1]$ -- where ${a}_{i}(t)$ denotes the $i$-th component of $\boldsymbol{a}(t)$ --
is an element in $\mathscr{T}(\M)$ satisfying $\Upsilon(\U) = \boldsymbol{a}$. Moreover, $\Upsilon$ preserves the inner product. Indeed,  every $\V_1, \V_2\in\TT_{\F}$:
\begin{align*}
 \langle \V_1,\V_2 \rangle_{\F} \:
    =&\: \int_{0}^1 \langle \V_1(t),\V_2(t)\rangle _{\Tan_{\F(t)}}\\
     =&\: \int_{0}^1 \langle \sum_{i} a_{1,i}(t)\Phi_i(t),  \sum_{i} a_{2,i}(t)\Phi_i(t) \rangle _{\Tan_{\F(t)}}\\
    =&\: \int_{0}^1 \sum_{i}   a_{1,i}(t) a_{2,i}(t) \rangle _{\Tan_{\F(t)}}
    = \: \langle \Upsilon(\V_1) , \Upsilon(\V_2)\rangle_{L^{2}(\T, \ell^2)}.
\end{align*}
 Therefore, it is a Hilbertian isomorphism, implying the separability of $\TT (\F)$.
\end{proof}

\begin{proof}[Proof of Lemma \ref{lemma : empirical FMF via pointwise min}]
It suffices to show that, with probability one, the pointwise minimiser $\widehat{\M}(t)$ eventually (in $n$) exists uniquely, for all $t$. The rest follows as in Lemma 1. For any $t\in [0,1]$, assumption \ref{assumption_flow_posdef} guarantees that the stopping time $N_t=\inf\{m\geq 1:F_m(t)\succ 0\}$ is a well-defined (almost surely finite) geometric random variable. By continuity, there is an open interval $I_t\ni t$ such that $F_{N_t}(u)\succ 0$ for all $u\in I_t$. The open intervals $\{I_t\}_{t\in [0,1]}$ form a cover of $[0,1]$, which admits a finite sub-cover $\{I_t\}_{t\in B}$, for some finite set $B\subset[0,1]$. Now define $M:=\max\{N_t:t\in B\}$, and observe that this is an almost-surely finite random variable. We have thus established that the pointwise minimiser $\widehat\M(t)$ exists uniqueley for all $n\geq M$, as claimed.
\end{proof}

\begin{proof}[Proof of Theorem \ref{thm : FM consistency and CLT}]

We first show \eqref{eq:consistency} by a dominated convergence argument. 
Set $f_n(t):= \Pi^2(\M(t),\widehat\M_n (t))$ for $t\in[0,1]$. 
By \citet[][Theorem 2]{santoro2023large}, $f_n(t)$ almost surely converges to $0$ for every $t\in[0,1]$. Moreover, we may upper bound $f_n$ as follows:
$$
    f_n(t) = \Pi^2(\M(t),\widehat\M_n (t)) \leq \tr(\M(t)) +   \tr(\widehat\M_n (t) )
$$
and since $\widehat\M_n (t) \leq S_n (t) = \frac{1}{n}\sum_{i=1}^n \F_i (t)$   (see \citet[][Theorem 12]{masarotto2019procrustes}), the SLLN yields that:
$$
    f_n(t) \leq \tr(\M(t) ) +   \tr( S_n (t) ) \leq  \tr( \M(t) ) + 2  \tr(\E \F(t)), \qquad \text{almost surely,}
$$
for large $n$. In particular, the sequence of random functions $f_n$ is almost surely dominated by an integrable deterministic function.
Hence, employing the dominated convergence theorem twice -- once for the integral and once for the expectation -- as well as Fubini-Tonelli to exchange integral and expectation, we obtain:
$$
\E \Big[ \int_{0}^1  f_n(t) dt \Big] = \int_{0}^1 \E [\, f_n(t)\, ]dt\rightarrow 0.
$$
In particular, since convergence in expectation implies convergence in probability :
$$
 d(\M, \widehat\M_n)^2 = \int_{0}^1  \Pi^2(\M(t),\widehat\M_n (t)) dt = 
 \int_{0}^1  f_n(t)dt
\: \overset{\PP}{\longrightarrow} \:0, \quad \text{as } n\rightarrow 0,
$$
proving the consistency result in \eqref{eq:consistency}

\medskip

Finally we prove the stronger result \eqref{eq:consistency+rate}. Observe that by \ref{assumption_invertible_maps}, via \citet[Theorem 3]{santoro2023large}, we obtain that:
\begin{equation}\label{eq:pointwiseCLT}
\Pi(\M(t), \widehat\M_n(t))^2 = O_p(n^{-1}),\qquad \forall\:t\in[0,1].
\end{equation}
We claim that:
$$
\forall \varepsilon>0 \:\exists\: C_\varepsilon>0, \tilde\Omega_{\varepsilon}\::\: \PP(\tilde\Omega_{\varepsilon})>1-\varepsilon \quad \text{s.t.:}\quad 
\limsup_{n\rightarrow\infty} \sup_{t\in[0,1]}n^{\nicefrac{1}{2}-\alpha}\Pi(\M(t),\widehat\M_n(\omega, t)) <C, \quad \forall\:\omega\in\tilde\Omega_{\varepsilon}.$$
Indeed, suppose the contrary.
Then, for some $\varepsilon$ and any positive $C>0$ we could find a set $A$ with $\PP(A)>\varepsilon$ and a sequence $\{t_n\}_{n\geq 1}\subset[0,1]$ such that:
$$
\limsup_{n\rightarrow\infty} \Pi(\M(t_n),\widehat\M_{n}(\omega, t_n)) > C, \quad \forall\:\omega\in A.
$$
Note that $t_n$ stays in the compact $[0,1]$, and we may thus extract a subsequence $\{t_{n_k}\}_{k\geq 1}$ converging to some $t^{\star}$. And since $s\mapsto \Pi(\M(s),\widehat\M_n(s))$ is almost surely continuous, for $n$ large enough we have that:
$$
 \Pi(\M(t^{\star}),\widehat\M_n(\omega, t^{\star}) \geq \Pi(\M(t_{n_k}),\widehat\M_n(\omega, t_{n_k}) - C/2 
$$
and hence:
$$
\limsup_{k\rightarrow\infty} n_k^{\nicefrac{1}{2}-\alpha}\Pi(\M(t^{\star}),\widehat\M_n(\omega, t^{\star}) \geq C/2,  \quad \forall\:\omega\in A,
$$
clearly contradicting \eqref{eq:pointwiseCLT}.

Hence, for any $\delta>0$:
\begin{align*}
    \limsup_{n\rightarrow\infty} \PP(  n^{\nicefrac{1}{2}-\alpha}d(\M,\widehat\M_n) > \delta) 
    & \leq  \lim_{n\rightarrow\infty} \PP(    n^{\nicefrac{1}{2}-\alpha}d(\M,\widehat\M_n) > \delta \:\cap\: \tilde\Omega_{\epsilon}) + {\PP(\tilde\Omega_{\epsilon}^c)}\\
    & \leq  \lim_{n\rightarrow\infty} \PP( n^{\nicefrac{1}{2}-\alpha}d(\M,\widehat\M_n >\delta )\:|\: \tilde\Omega_{\epsilon}] + \epsilon \\
    & \leq \frac{1}{\delta}  \lim_{n\rightarrow\infty} \E[ n^{\nicefrac{1}{2}-\alpha}d(\M,\widehat\M_n)\:|\: \tilde\Omega_{\epsilon}] + \epsilon \\ 
    & \leq \ \epsilon.
\end{align*}
where we have used Markov's inequality, and that the conditioning on $\tilde\Omega_\epsilon$ allows us to apply dominated convergence to prove that $\E[ n^{\nicefrac{1}{2}-\alpha}d(\M,\widehat\M_n)\:|\: \tilde\Omega_{\epsilon}] =o(1)$.

Since $\epsilon$ can be chosen to be arbitrary small, $n^{\nicefrac{1}{2}-\alpha}d(\M,\widehat\M_n) = o_p(1)$ for all $\alpha>0$, and hence
$d(\M,\widehat\M_n) = O_p(n^{-\nicefrac{1}{2}})$, proving the result.

\end{proof}

\begin{proof}[Proof of Theorem \ref{thm : cov consistency}]
We split the term in 2 summands, which may be controlled separately.
\begin{equation}\label{eq:proof:covconv:firstsplit}
\begin{split}
    \| {\bJ}^{\otimes}_\M\CC  - {\bJ}^{\otimes}_{\widehat\M_n}\widehat\CC_n  \|_{\BB( L^2_{\BB_2})} 
    =&\: 
    \Big\Vert \E \left[ (\log_\M\F) \M^{\nicefrac{1}{2}} \otimes (\log_\M\F) \M^{\nicefrac{1}{2}} \right] 
    \\
        &\qquad \qquad
        - \frac{1}{n} \sum_{i=1}^n (\log_{\widehat\M_n}\F_i)\widehat\M_n^{\nicefrac{1}{2}}\otimes (\log_{\widehat\M_n}\F_i)\widehat\M_n^{\nicefrac{1}{2}} \Big\Vert_{\BB( L^2_{\BB_2})}  
        \nonumber \\ 
    \leq & \:
    \Big\Vert \E \left[ (\log_\M\F) \M^{\nicefrac{1}{2}} \otimes (\log_\M\F) \M^{\nicefrac{1}{2}} \right] 
            \\
        &\qquad \qquad
        - \frac{1}{n} \sum_{i=1}^n (\log_{\M}\F_i)\M^{\nicefrac{1}{2}}\otimes (\log_{\M}\F_i)\M^{\nicefrac{1}{2}} \Big\Vert_{\BB( L^2_{\BB_2})} \nonumber\\
    & \:  + \Big\Vert \frac{1}{n} \sum_{i=1}^n  (\log_{\M}\F_i)\M^{\nicefrac{1}{2}}\otimes (\log_{\M}\F_i)\M^{\nicefrac{1}{2}}  
        \\
        &\qquad \qquad
        -  \frac{1}{n} \sum_{i=1}^n (\log_{\widehat\M_n}\F_i)\widehat\M_n^{\nicefrac{1}{2}}\otimes (\log_{\widehat\M_n}\F_i)\widehat\M_n^{\nicefrac{1}{2}}  \Big\Vert_{\BB( L^2_{\BB_2})}
    \nonumber
\end{split}
\end{equation}

The first term vanishes by the strong law of large numbers.
The second summand in the rhs requires some additional work. First, we split it as follows:
\begin{align*}
    \Big\Vert &\frac{1}{n} \sum_{i=1}^n \left( \log_{\M}\F_i\M^{\nicefrac{1}{2}}\otimes \log_{\M}\F_i\M^{\nicefrac{1}{2}}  -  \log_{\widehat\M_n}\F_i\widehat\M_n^{\nicefrac{1}{2}}\otimes \log_{\widehat\M_n}\F_i\widehat\M_n^{\nicefrac{1}{2}} \right) \Big\Vert_{\BB( L^2_{\BB_2})}
     \\ & \leq  
     \Big\Vert \frac{1}{n} \sum_{i=1}^n (\log_{\M}\F_i\M^{\nicefrac{1}{2}} - \log_{\widehat\M_n}\F_i\widehat\M_n^{\nicefrac{1}{2}}) \otimes \log_{\M}\F_i\M^{\nicefrac{1}{2}} \Big\Vert_{\BB( L^2_{\BB_2})}    
    \\ &\quad +
    \Big\Vert \frac{1}{n} \sum_{i=1}^n \log_{\widehat\M_n}\F_i\widehat\M_n^{\nicefrac{1}{2}}\otimes (\log_{\M}\F_i\M^{\nicefrac{1}{2}}  - \log_{\widehat\M_n}\F_i\widehat\M_n^{\nicefrac{1}{2}} ) \Big\Vert_{\BB( L^2_{\BB_2})}
\end{align*}
We show in detail how to handle the first term; the second one may be treated equivalently, given the evident symmetry.
\begin{multline*}
    \Big\Vert  \frac{1}{n} \sum_{i=1}^n (\log_{\M}\F_i\M^{\nicefrac{1}{2}} - \log_{\widehat\M_n}\F_i\widehat\M_n^{\nicefrac{1}{2}}) \otimes (\log_{\M}\F_i)\M^{\nicefrac{1}{2}} \Big\Vert_{\BB( L^2_{\BB_2})} 
    \hspace{3cm}\\
\leq  
    \frac{1}{n^2}\sum_{i,j=1}^n 
        \Big( \|  \log_{\M}\F_i\M^{\nicefrac{1}{2}}  \|_{L^2_{\BB_2}} + \|  \log_{\M}\F_j\M^{\nicefrac{1}{2}}  \|_{L^2_{\BB_2}}   \Big)
        \Big(\|  \log_{\M}\F_i\M^{\nicefrac{1}{2}} - \log_{\widehat\M_n}\F_i\widehat\M_n^{\nicefrac{1}{2}}  \|_{L^2_{\BB_2}} 
        \\
        \qquad + \|  \log_{\M}\F_j\M^{\nicefrac{1}{2}} -
        \log_{\widehat\M_n}\F_j\widehat\M_n^{\nicefrac{1}{2}}  \|_{L^2_{\BB_2}}
\\
\leq \frac{1}{n^2}\sum_{i,j=1}^n 
        \Big( \|  \log_{\M}\F_i  \|_{\TT_{\M}} + \|  \log_{\M}\F_j  \|_{\TT_{\M}}   \Big)
        \Big(\|  \log_{\M}\F_i\M^{\nicefrac{1}{2}} - \log_{\widehat\M_n}\F_i\widehat\M_n^{\nicefrac{1}{2}}  \|_{L^2_{\BB_2}}
        \\
        + \|  \log_{\M}\F_j\M^{\nicefrac{1}{2}} - \log_{\widehat\M_n}\F_j\widehat\M_n^{\nicefrac{1}{2}}  \|_{L^2_{\BB_2}}
        \Big)     
\\
        \leq \frac{2}{n^2}\sum_{i,j=1}^n  
        \|  \log_{\M}\F_i  \|_{\TT_{\M}} \|  \log_{\M}\F_j\M^{\nicefrac{1}{2}} - \log_{\widehat\M_n}\F_j\widehat\M_n^{\nicefrac{1}{2}}  \|_{L^2_{\BB_2}}
        \\
        +     \frac{1}{n}\sum_{i=1}^n
        \|  \log_{\M}\F_i  \|_{\TT_{\M}} \|  \log_{\M}\F_i\M^{\nicefrac{1}{2}} - \log_{\widehat\M_n}\F_i\widehat\M_n^{\nicefrac{1}{2}}  \|_{L^2_{\BB_2}}.
\end{multline*}

Therefore, it suffices to show that:
\begin{equation}\label{eqCovEst_vanishingAverageLog}
\frac{1}{n}\sum_{i=1}^n\| \log_{\M}\F_i\M^{\nicefrac{1}{2}} - \log_{\widehat\M_n}\F_i\widehat\M_n^{\nicefrac{1}{2}}  \|_{L^2_{\BB_2}}\: \overset{\PP}{\longrightarrow} \:0 \quad \text{as} \quad n\rightarrow\infty.
\end{equation}
Observe that $\{ \log_{\M}\F_i\M^{\nicefrac{1}{2}} - \log_{\widehat\M_n}\F_i\widehat\M_n^{\nicefrac{1}{2}}\}_{j=1,\dots,n}$ comprise a family of \textit{dependent, though identically distributed, random variables. Therefore:}
$$
\E\left[ \Big|\frac{1}{n}\sum_{i=1}^n\| \log_{\M}\F_i\M^{\nicefrac{1}{2}} - \log_{\widehat\M_n}\F_i\widehat\M_n^{\nicefrac{1}{2}}  \|_{L^2_{\BB_2}} \Big| \right] = \E\left[ \| \log_{\M}\F_1\M^{\nicefrac{1}{2}} - \log_{\widehat\M_n}\F_1\widehat\M_n^{\nicefrac{1}{2}}  \|_{L^2_{\BB_2}}  \right],
$$
and thus \eqref{eqCovEst_vanishingAverageLog} follows if we can show that $\|\log_{\M}\F_1\M^{\nicefrac{1}{2}} - \log_{\widehat\M_n}\F_1\widehat\M_n^{\nicefrac{1}{2}}\|_{L^2_{\BB_2}}$ converges to zero in expectation.

In the following, we write $T_i = T_{\M}^{\F_i}$ and $\widehat T_{n,i} = T_{\widehat\M_n}^{\F_i}$ for the time-evolving optimal maps from the population and sample barycenters to $\F_i$, respectively, for $i=1,\dots,n$. Writing the logarithms explicitly in terms of the transport maps, and employing the triangle inequality, we obtain:
\begin{equation}
\label{eqCovEst_2terms}
\| \log_{\M}\F_i\M^{\nicefrac{1}{2}} - \log_{\widehat\M_n}\F_i\widehat\M_n^{\nicefrac{1}{2}}  \|_{L^2_{\BB_2}} 
\leq
\| T_i\M^{\nicefrac{1}{2}} -  \widehat T_{n,i}\widehat\M_n^{\nicefrac{1}{2}}\|_{L^2_{\BB_2}} 
+
\|\M^{\nicefrac{1}{2}} - \widehat\M_n^{\nicefrac{1}{2}} \|_{L^2_{\BB_2}}.
\end{equation}
It is easy to see that the second term in \eqref{eqCovEst_2terms} vanishes. Indeed, we may expand it by:
\begin{equation}
    \label{eqCovEst_convHSnormBARY}
\begin{split}
\|\M^{\nicefrac{1}{2}} - \widehat\M_n^{\nicefrac{1}{2}}\|_{L^2_{\BB_2}}^2 = & \:
 \int_{0}^1  \| \M(t)^{\nicefrac{1}{2}} - \widehat\M_n (t)^{\nicefrac{1}{2}} \|_2^2 dt  
 \\ \leq &  \:
 \int_{0}^1  \Big( \tr(\M(t))^{\nicefrac{1}{2}} + \tr(\widehat\M_n (t))^{\nicefrac{1}{2}} \Big)
 \Big( \Pi(\M(t),\widehat\M_n (t) )\Big)dt
  \\ \leq &  \:
  4 \Big( \int_{0}^1  \tr(\M(t)) dt \Big)^{\nicefrac{1}{2}}
 \Big(  \int_{0}^1  \Pi^2(\M(t),\widehat\M_n (t)dt )\Big)^{\nicefrac{1}{2}}
   \\ \leq &  \:
 c\cdot 
 d(\M,\widehat\M_n)\: \overset{\PP}{\longrightarrow} \:0, \quad \text{as} \: n\rightarrow \infty,
\end{split}
\end{equation}
where we have used that $\| F-G\|_1 \leq (\tr(F)^{\nicefrac{1}{2}} + \tr(G)^{\nicefrac{1}{2}}) \Pi(F , G )$ for any two $F,G\in\KK$ \cite[][Lemma 1]{santoro2023large}, H\"older's inequality, and the convergence in \eqref{eq:consistency}. It remains to show that the first term in \eqref{eqCovEst_2terms} vanishes as well. We proceed as follows:
\begin{equation}\label{eqCovEst_2terms2_B}
    \begin{split}
     \| T_i\M^{\nicefrac{1}{2}} -  \widehat T_{n,i}\widehat\M_n^{\nicefrac{1}{2}}\|_{L^2_{\BB_2}}^2  = & \:
    \int_{0}^1 \| T_i(t)\M(t)^{\nicefrac{1}{2}} -  \widehat T_{n,i}(t)\widehat\M_n (t)^{\nicefrac{1}{2}} \|_2 \: dt 
     \\ \leq & \:
      \int_{0}^1 \| T_{i}(t)\big(\M(t)^{\nicefrac{1}{2}} -  \widehat\M_n (t)^{\nicefrac{1}{2}} \big)\|_2 \: dt
      +
     \int_{0}^1 \| \big(T_i(t) - \widehat T_{n,i}(t)\big)\widehat\M_n(t)^{\nicefrac{1}{2}}  \|_2 \: dt. 
\end{split}
\end{equation}
We handle the two terms appearing in the last line of \eqref{eqCovEst_2terms2_B} separately. Employing H\"older's inequality, the second term may be bounded by:
$$
   \left( \int_{0}^1 \| T_{i}(t)\|_{\infty}^2  dt \right)^{\nicefrac{1}{2}} \cdot  \|\M^{\nicefrac{1}{2}} - \widehat\M_n^{\nicefrac{1}{2}}\|_{L^2_{\BB_2}}
$$
and thus consists in the product of a term that is bounded in probability by assumption \ref{assumption_flow_boundedmap}, and a second one vanishing in probability. Therefore their product vanishes in probability by Slutsky's theorem.

Next we show that the second term in \eqref{eqCovEst_2terms2_B} vanishes as well. Note that, by \citet[][Corollary 2]{santoro2023large}, the integrand converges pointwise to $0$ in probability. That is, that for every $t\in[0,1]$ we have $ \| \big(T_i(t) - \widehat T_{n,i}(t)\big)\M(t)^{\nicefrac{1}{2}}  \|_2^2\: \overset{\PP}{\longrightarrow} \:0$. By dominated convergence, this implies pointwise convergence to zero in expectation as well. In particular:
$$
\E \int_{0}^1 \| \big(T_i(t) - \widehat T_{n,i}(t)\big)\M(t)^{\nicefrac{1}{2}}  \|_2 ^2dt 
=
\int_{0}^1 \| \E \big(T_i(t) - \widehat T_{n,i}(t)\big)\M(t)^{\nicefrac{1}{2}}  \|_2 ^2 dt \rightarrow 0 
$$
where we have employed Fubini-Tonelli to exchange integral and expectation and integrability of $t\mapsto \E \tr(\M(t))$ for the dominated convergence.  Since convergence in expectation implies convergence in probability, we obtain that \eqref{eqCovEst_2terms2_B} vanishes in probability, and consequently \eqref{eqCovEst_vanishingAverageLog}, hence proving \eqref{eq:covariance_consistency}.

\end{proof}

\begin{proof}[Proof of Theorem \ref{thm : FM consistency - finite dim}]
The proof of \eqref{eq:consistency-finite dim} is identical to that of \eqref{eq:consistency+rate}, where it is presented in the infinite-dimensional setting: in that general case, the stronger assumptions where needed to get the pointwise rate of convergence, but the proof for deducing the rate for the integral metric $d$ is identical.

\medskip

Next we prove the uniform convergence in \eqref{eq:consistency-finite dim-uniform}. 
The assumption \ref{assumption_invertible_maps} with $\mathcal{R}(t)=R\succ 0 $ allows us to bound the minimum and maximum eigenvalues of $\F(t)$ almost surely, uniformly over $t$, from above and below, away from $0$ and $\infty$ respectively. That is, there exist $0<c_m\leq c_M<\infty$ such that:
$$
\inf_{t\in[0,1]}\lambda_{\min}(\F(t)) >c_m, 
\qquad \text{and}\qquad
\sup_{t\in[0,1]}\lambda_{\max}(\F(t)) <c_M, \qquad \text{almost surely.}
$$
Define $\kappa:= {c_M}/{c_m}$. By \citet[][Corollary 17]{le2022fast}, we may see that:
\begin{align*}
    \sup_{t\in[0,1]}\Pi(\M(t), \widehat\M_n(t))^2 \leq & \frac{1}{n}\frac{4}{(1-\kappa+\kappa^{-1})^2} \sup_{t\in[0,1]} \E\left[\Pi^2(\M(t), \F(t))\right] = O(n^{-1}).
\end{align*}
    
\end{proof}

\begin{proof}[Proof of Lemma \ref{lemma : smoothness of emblog}]
Let us write logarithms explicitly in terms of the transport maps:
\begin{equation}\label{eq:taylor_inproof}
\J_M\log_{M}F =
(T_{M}^{F}-\Id)M^{\nicefrac{1}{2}} = M^{-\nicefrac{1}{2}}(M^{\nicefrac{1}{2}}FM^{\nicefrac{1}{2}})^{\nicefrac{1}{2}} - M^{\nicefrac{1}{2}}.
\end{equation}
Taking derivatives, we see that:
\begin{align*}
    {d_{M^{\nicefrac{1}{2}}}}[\J_M\log_{M}F](H) =& 
    M^{-\nicefrac{1}{2}}\der_{ M^{\nicefrac{1}{2}}FM^{\nicefrac{1}{2}}}g (H FM^{\nicefrac{1}{2}} + M^{\nicefrac{1}{2}}F H) - M^{-\nicefrac{1}{2}}HM^{-\nicefrac{1}{2}}(M^{\nicefrac{1}{2}}FM^{\nicefrac{1}{2}})^{\nicefrac{1}{2}} - H.
\end{align*}
where $\der_{ A^2 }g$ denotes the \Frechet derivative of the square-root operator $A\mapsto A^{\nicefrac{1}{2}}$.
For convenience, we define $Q^2 = M^{\nicefrac{1}{2}}FM^{\nicefrac{1}{2}}$. Then:
\begin{align*}
    {d_{M^{\nicefrac{1}{2}}}}[\J_M\log_{M}F](H) 
    =& 
    M^{-\nicefrac{1}{2}}\der_{ Q^2}g (H M^{-\nicefrac{1}{2}}Q^2 + Q^2M^{-\nicefrac{1}{2}} H) - M^{-\nicefrac{1}{2}}HM^{-\nicefrac{1}{2}}Q - H\\ 
    =& 
    M^{-\nicefrac{1}{2}}QM^{-\nicefrac{1}{2}} H - M^{-\nicefrac{1}{2}}\der_{ Q^2}g (QH M^{-\nicefrac{1}{2}}Q + QM^{-\nicefrac{1}{2}} HQ) - H\\ 
    =& 
    T_M^F H - M^{-\nicefrac{1}{2}}\der_{ Q^2}g (QH M^{-\nicefrac{1}{2}}Q + QM^{-\nicefrac{1}{2}} HQ) - H 
\end{align*}
Note that $M^{-\nicefrac{1}{2}}Q^{\nicefrac{1}{2}}$ defines a bounded operator, and that $X\mapsto Q^{\nicefrac{1}{2}}d_{Q^2}g(X)Q^{\nicefrac{1}{2}}$ and $ X \mapsto Qd_{Q^2}g(X)$ are both bounded operators acting over $\BB_2$, with operator norm  less then $1$.
Therefore:
$$
\| {d_{M^{\nicefrac{1}{2}}}}[\J_M\log_{M}F]\| \leq \|T_{M}^{F}\| + 3.
$$
\end{proof}

\begin{proof}[Proof of Theorem \ref{thm : cov consistency - finite dim}] The proof in the finite dimensional case mimics the steps of the general case, though some specifications occur. First, we split the difference as in \eqref{eq:proof:covconv:firstsplit}, and reduce the problem to that of controlling the two terms:
\begin{enumerate}
    \item $\Big\Vert \E \left[ (\log_\M\F) \M^{\nicefrac{1}{2}} \otimes (\log_\M\F) \M^{\nicefrac{1}{2}} \right] 
        - \frac{1}{n} \sum_{i=1}^n (\log_{\M}\F_i)\M^{\nicefrac{1}{2}}\otimes (\log_{\M}\F_i)\M^{\nicefrac{1}{2}} \Big\Vert_{\BB( L^2_{\BB_2})}$
    \item $\Big\Vert \frac{1}{n} \sum_{i=1}^n  (\log_{\M}\F_i)\M^{\nicefrac{1}{2}}\otimes (\log_{\M}\F_i)\M^{\nicefrac{1}{2}}  
        -  \frac{1}{n} \sum_{i=1}^n (\log_{\widehat\M_n}\F_i)\widehat\M_n^{\nicefrac{1}{2}}\otimes (\log_{\widehat\M_n}\F_i)\widehat\M_n^{\nicefrac{1}{2}}  \Big\Vert_{\BB( L^2_{\BB_2})}$
\end{enumerate}

The first term in the rhs is almost surely of order ${O}(n^{-\nicefrac{1}{2}})$ by the Central Limit Theorem in Hilbert spaces (see \cite{dauxois1982asymptotic}). We then proceed as in the proof of the infinite dimensional case, and bound the second term in the rhs by:
$$
O(1)\frac{1}{n}\sum_{i=1}^n\| \log_{\M}\F_i\M^{\nicefrac{1}{2}} - \log_{\widehat\M_n}\F_i\widehat\M_n^{\nicefrac{1}{2}}  \|_{L^2_{\BB_2}}.
$$
Finally, Lemma \ref{lemma : smoothness of emblog} allows us to establish that:
$$
\frac{1}{n}\sum_{j=1}^n \| \log_{\M}\F_j\M^{\nicefrac{1}{2}} - \log_{\widehat\M_n}\F_j\widehat\M_n^{\nicefrac{1}{2}}  \| = O_p(d(\widehat\M_n, \M)) = O_p(n^{-\nicefrac{1}{2}}),
$$
which thus establishes \eqref{eq:covariance_consistency_FINDIM}.
\end{proof}

\begin{proof}[Proof of Proposition \ref{prop:NWsmoothing}]
We adapt the argument in \citet[][Theorem 5.2]{gyorfi2002distribution}
For simplicity let us consider the case where $W$ is the uniform kernel $W(x) = \1_{|x|<h}$. Denote by $\mu_r$ the empirical distribution of the sampling times $\{T_i\}_{i=1,\dots,r}$ Then:
$$
\hat{\F}_r^{\mathrm{NW}}(t) = \frac{\sum_{i=1}^{r} F_i \1_{\lvert  t  - T_i \rvert<h }}{r\mu_r(B_t(h))}
$$
where $B_t(h)$ denotes the open interval $(t-h,t+h)$.  Let us define:
$$
\check{\F}_r^{\mathrm{NW}}(t):= \frac{\sum_{i=1}^{r}\F(T_i)\1_{\lvert  t  - T_i \rvert<h }}{r\mu_r(B_t(h))}.
$$
Then:
\begin{align*}
    \E\left[ \| \check{\F}_r^{\mathrm{NW}}(t) - \hat{\F}_r^{\mathrm{NW}} (t))\|_1\:\vert\:  T_1,\dots,T_r \right]
    \:=&\:
     \E\left[   (r\mu_r(B_t(h)))^{-1}\sum_{i=1}^{r} \|(F_i - \F(T_i))\|_1 \1_{\lvert  t  - T_i \rvert<h }\:\vert \: T_1,\dots,T_r \right]\\
         \:\leq &\:
     \left[   (r\mu_r(B_t(h)))^{-1}\sum_{i=1}^{r} \E\|(F_i - \F(T_i))\|_1 \1_{\lvert  t  - T_i \rvert<h }\:\vert \: T_1,\dots,T_r \right]\\
     \: \leq &\: \mathcal{R}(K).
     \end{align*}
By the Lipschitz property, denoting by $A_{t,h}$ the event that at least one $T_i$ lies in the ball $B_t(h)$:
\begin{align*}
   \| \check{\F}_r^{\mathrm{NW}}(t) - \F(t) \|_1 
   \:=&\: \| (r\mu_r(B_t(h)))^{-1}{\sum_{i=1}^{r}(\F(T_i) - \F(t) ) \1_{\lvert  t  - T_i \rvert<h }}\1_{A_{t,h}} + \F(t) \1_{A_{t,h}^c}\|_1 \\
         \:\leq &\:
   Lh\1_{A_{t,h}}  + \|\F(t)\|_1 \1_{A_{t,h}^c}
\end{align*}
and consequently:
\begin{align*}
    \E\left[ \int_0^1  \| \F(t) - \hat{\F}_r^{\mathrm{NW}} (t))\|_1 dt\:\vert\:  T_1,\dots,T_r \right]
    \:\leq &\:
    \E\left[ \int_0^1 \| \F(t) - \hat{\F}_r^{\mathrm{NW}} (t))\|_1 dt \:\vert\:  T_1,\dots,T_r \right] \\
    &\qquad \qquad + \E\left[ \int_0^1 \|  \hat{\F}_r^{\mathrm{NW}} (t) - \check{\F}_r^{\mathrm{NW}} (t))\|_1 dt \:\vert\:  T_1,\dots,T_r \right]\\
    \:\leq &\:
     Lh\int_0^1\1_{A_{t,h}}dt  + \int_0^t\|\F(t)\|_1 \1_{A_{t,h}^c}dt  +  \mathcal{R}(K)\\
    \:\leq &\:
     Lh  + (rh)^{-1}\sup_t\|\F(t)\| +  \mathcal{R}(K) 
     \end{align*}
     and since we may bound $\Pi(A,B)^2\leq \|A-B\|_1$ for any $A,B\in\KK$  \cite[][Lemma 1, 2)]{santoro2023large} we obtain:
     $$
      \E\left[ d(\F(t), \hat{\F}_r^{\mathrm{NW}})^2\right]  = 
    \mathcal{O}\left(Lh  + (rh)^{-1} +  \mathcal{R}(K) 
 \right).
     $$
\end{proof}

\begin{proof}[Proof of Proposition \ref{prop:loclin FM from sparse}]
Let us fix some $t\in[0,1]$ and define the following simplified notation, for ease of readability:
$$
M_t:= \argmin_{G\in\KK}\E[ \Pi^2(\F(t),G)] 
$$
Then, let us also consider:
$$
\check M_t :=  \argmin_{G\in\KK}\E[s(T,t,F_T) \Pi^2(F_T,G)].
$$
Note that, up to renormalisation, we may assume without loss of generality that $\E [s(T,t,F_T) ]=1$, so that writing the Bures-Wasserstein distance explicitly we may write:
\begin{align*}
\E[s(T,t,F_T) \Pi^2(F_T,G)] =&
\E[s(T,t,F_T)\tr(F_T)]+\tr(G)-2 \E[s(T,t,F_T) \tr(F_T^{\nicefrac{1}{2}} G F_T^{\nicefrac{1}{2}})^{\nicefrac{1}{2}}] 
\\ =&
\Pi^2[\tilde F_T,G)]
\end{align*}
where $\tilde F_T:=s(T,t,F_T)F_T$ is a scaled modification of $F_T$, which thus still satisfies 
\ref{assumption_flow_trbound},\ref{assumption_flow_posdef} and \ref{assumption_invertible_maps}.
Therefore, $\check M_t$ may be seen as the \Frechet mean of the scaled covariances $\tilde\F_T$. 
Similarly, defining:
$$
\tilde F_{ij} := s(T_ij,t,F_{ij})F_{ij}
$$
the Local \Frechet regression estimate at time $t$ may be expressed as:
\begin{equation}\label{eq:empirical loclin fm}
     \hat M_t :=  \argmin_{G\in\KK}\sum_{i=1}^n\sum_{j=1}^r \Pi^2(\tilde F_{ij},G)
\end{equation}
Therefore, we see that $\Pi(\check M_t,\hat M_t) = o(1)$ for all $t\in[0.1]$. In particular, in finite dimensions we have that:
$$
\sup_{t\in[0,1]}\Pi(\hat M_t,\check M_t) = O(\sqrt{nh}),
$$
as $\sqrt{nh}$ is the average number of non-zero addends appearing in the sum \eqref{eq:empirical loclin fm}.

\medskip

Following the proof of \citet[][Theorem 3]{petersen2019frechet} we see that
$$
\E[ \Pi^2(\F(t),G)] - \E[s(T,t,F_T) \Pi^2(F_T,G)] = O(h^2)
$$
where the $O(h^2)$ is uniform in $G$. Hence we may employ the general result in \cite[][Corollary 3.2.3 (ii)]{van1996weak} directly obtain that
$
\Pi(\check M_t, M_t)\rightarrow 0$ as $h\rightarrow 0$.
In fact, in the finite dimensional case we may derive something more; the fixed point characterisation of Bures-Wasserstein barycentres yields that:
$$
M_t = \E[( M_t^{\nicefrac{1}{2}}\F(t)  M_t^{\nicefrac{1}{2}})^{\nicefrac{1}{2}} ]
\qquad \text{and} \qquad
\check M_t = \E[s(T,t,F_T) (\check M_t^{\nicefrac{1}{2}}F\check M_t^{\nicefrac{1}{2}})^{\nicefrac{1}{2}} ]
$$
Define $\Psi_t(A) := \E[( A^{\nicefrac{1}{2}}\F(t)  A^{\nicefrac{1}{2}})^{\nicefrac{1}{2}} ]-A$ and $\check\Psi_t(A) := \E[( A^{\nicefrac{1}{2}}\tilde F_T A^{\nicefrac{1}{2}}]-A$. Then:
\begin{align*}
   0 
    = \Psi_t(M_t) - \check \Psi_t(\check M_t) 
    = \Psi_t(M_t)- \Psi_t(\check M_t) + \Psi_t(\check M_t)- \check \Psi_t(\check M_t) 
    = \Psi_t(M_t)- \Psi_t(\check M_t) + O(h^2)
\end{align*}
which yields:
$$
d_{M_t}\Psi_t(\check M_t - M_t) + o(\check M_t - M_t)  =  O(h^2)
$$
and to conclude it suffices to argue the bounded invertibility of $d_{M_t}\Psi_t$, uniformly in $t$, which can be easily obtained provided $\inf_t \lambda_{\min}(M_t) > 0$, which in turn follows from Assumption \ref{assumption_invertible_maps}.

\end{proof}

\begin{proof}[Proof of Proposition \ref{prop:pace}]
We follow \citet[][Theorem 2]{dai2021modeling}. The optimal value admits the following expression:
    $$
\hat{\Gamma}(s, t)=\frac{\left(S_{20} S_{02}-S_{11}^2\right) R_{00}-\left(S_{10} S_{02}-S_{01} S_{11}\right) R_{10}+\left(S_{10} S_{11}-S_{01} S_{20}\right) R_{01}}{\left(S_{20} S_{02}-S_{11}^2\right) S_{00}-\left(S_{10} S_{02}-S_{01} S_{11}\right) S_{10}+\left(S_{10} S_{11}-S_{01} S_{20}\right) S_{01}},
$$
where for $a, b=0,1,2$,
$$
\begin{aligned}
S_{a b} & =\frac{1}{nr}\sum_{i=1}^n \sum_{1 \leq j \neq l \leq r} K_h\left(T_{i j}-s\right) K_h\left(T_{i l}-t\right)\left(\frac{T_{i j}-s}{h}\right)^a\left(\frac{T_{i j}-t}{h}\right)^b, \\
R_{a b} & =\frac{1}{nr}\sum_{i=1}^n \sum_{1 \leq j \neq l \leq r} K_h\left(T_{i j}-s\right) K_h\left(T_{i l}-t\right)\left(\frac{T_{i j}-s}{h}\right)^a\left(\frac{T_{i j}-t}{h}\right)^b \Gamma_{i j l} .
\end{aligned}
$$
where recall that:
$$
\Gamma_{ij\ell}
 = \J_{\hat M_{ij}}\log_{\hat M_{ij}}F_{ij} \otimes \J_{\hat M_{i\ell}}\log_{\hat M_{i\ell}}F_{i\ell}.
$$
Let:
$$
\delta_{i j l}=  \left(\J_{M_{i j}}\log _{M_{i j}} F_{i j} \right) \otimes \left( \J_{M_{i l}} \log _{M_{i l}} F_{i l}\right)$$
and also define:
$$
R_{a b}^{\prime}=\frac{1}{nr}\sum_{i=1}^n \sum_{1 \leq j \neq l \neq m_i} K_h\left(T_{i j}-s\right) K_h\left(T_{i l}-t\right)\left(\frac{T_{i j}-s}{h}\right)^a\left(\frac{T_{i j}-t}{h}\right)^b \delta_{i j l} .
$$
Then:
\begin{equation}\label{eq:splitforsparse}
\begin{aligned}
& R_{00}^{\prime} =R_{00} + \frac{1}{nr}\sum_{i=1}^n \sum_{j \neq l }K_h\left(T_{i j}-s\right) K_h\left(T_{i l}-t\right)\left(\frac{T_{i j}-s}{h}\right)^a\left(\frac{T_{i j}-t}{h}\right)^b \Big(\\&
\left(\J_{ M_{ij}} \log _{ M_{ij}} F_{i j}-\J_{\hat{M}_{i j}}\log _{\hat M_{i j}} F_{i j}\right)\otimes\left(\J_{\hat M_{i\ell}} \log _{\hat M_{i\ell}} F_{i \ell}\right) + \\
& \left(\J_{\hat M_{ij}} \log _{\hat M_{ij}} F_{i j}\right)\otimes \left( \J_{\hat M_{i\ell}} \log _{\hat M_{i\ell}} F_{i \ell}-\J_{\hat M_{i\ell}}\log _{M_{i \ell}} F_{i \ell}\right) + \\
& \left( \J_{\hat M_{ij}}\log _{\hat M_{ij}} F_{i j}-\J_{M_{i j}}\log _{M_{i j}} F_{i j}\right)\otimes\left( \J_{\hat M_{i \ell}}\log _{\hat M_{i\ell}} F_{i \ell}-\J_{M_{i \ell}}\log _{M_{i \ell}} F_{i \ell}\right) \Big)
\end{aligned}
\end{equation}

If $\H$ is finite dimensional, the embedded logarithm is continuously differentiable by the result in Lemma \ref{lemma : smoothness of emblog}, and we may thus bound 
$$
\left\|R_{00}-R_{00}^{\prime}\right\|^2 \leq O(1) \sup _{t \in \mathcal{T}} \Pi^2(\hat{\M}(t), \M(t))
$$
and similarly for $R_{a b}$ for $a, b=0,1,2$. Setting
$$
\tilde{\Gamma}(s, t)=\frac{\left(S_{20} S_{02}-S_{11}^2\right) R_{00}^{\prime}-\left(S_{10} S_{02}-S_{01} S_{11}\right) R_{10}^{\prime}+\left(S_{10} S_{11}-S_{01} S_{20}\right) R_{01}^{\prime}}{\left(S_{20} S_{02}-S_{11}^2\right) S_{00}-\left(S_{10} S_{02}-S_{01} S_{11}\right) S_{10}+\left(S_{10} S_{11}-S_{01} S_{20}\right) S_{01}},
$$
by the uniform convergence of $\hat{\M}$ to ${\M}$, which holds in the finite dimensional case, we may see that 
$$
\sup_{s,t\in [0,1]}\|\hat{\Gamma}(s, t)-\tilde{\Gamma}(s, t)\|^2 = 
    O_P\left(h_{\M}^4 + {h_{\M}n}\right) 
    \quad \text{if }\dim(\H)<\infty.
$$
 Then, by the same argument as in the main result in \cite{mohammadi2024functional}:
$$
\sup _{s, t \in \mathcal{T}} \| \tilde{\Gamma}(s, t)- \Gamma(s, t) \|^2=O_p\left[h_{\Gamma}^{-4} \frac{\log n}{n}\left(h_{\Gamma}^4+\frac{h_{\Gamma}^3}{r}+\frac{h_{\Gamma}^2}{r^2}\right)\right]^{1 / 2}+O_p\left(h_{\Gamma}^2\right)
$$
which yields the finite-dimensional rate of convergence.

\medskip

In the general case $\dim(\H)=\infty$, note that the convergence of $\hat{\M}$ to ${\M}$ is only \textit{pointwise}, almost surely, so we need to proceed differently. By dominated convergence we may see that:
$$
\E \| \J_{\hat M_{1\ell}} \log _{\hat M_{1\ell}} F_{i \ell}-\J_{\hat M_{1\ell}}\log _{M_{1 \ell}} F_{1 \ell}\| \rightarrow 0, \qquad \forall\: \ell = 1,\dots, r
$$
and thus every addend in \eqref{eq:splitforsparse} may be bounded by:
\begin{align*}
O(1)  \max_{\ell=1,\dots,r}\E \|  \J_{\hat M_{1\ell}} \log _{\hat M_{1\ell}} F_{1 \ell}-\J_{\hat M_{1\ell}}\log _{M_{1 \ell}} F_{1 \ell} \| 
\end{align*}
 Since convergence in expectation implies convergence in probability, we obtain that $\hat\Gamma(s, t)$ is consistent for $\Gamma(s, t)$.

\end{proof}

\phantomsection
\addcontentsline{toc}{section}{References}
\bibliography{bib}

\end{document}